\documentclass[12pt]{amsart}
\usepackage{graphicx}
\usepackage{graphicx}
\usepackage{epsfig}
\usepackage{pstricks}
\usepackage{hyperref}
\usepackage{caption}
\usepackage{latexsym}
\usepackage{verbatim}

\newtheorem{thm}{Theorem}[section]

\newtheorem{defn}[thm]{Definition}
\newtheorem{lemma}[thm]{Lemma}

\newtheorem{cor}[thm]{Corollary}

\newtheorem{example}[thm]{Example}

\usepackage{amsmath}
\usepackage{amsxtra}
\usepackage{amscd}
\usepackage{amsthm}
\usepackage{amsfonts}
\usepackage{amssymb}
\usepackage{eucal}
\newcommand{\bmb}{\left( \begin{array}{rr}}
\newcommand{\enm}{\end{array}\right)}

\newcommand{\Z}{{\mathbb Z}}

\newcommand{\R}{{\mathbb R}}

\newcommand{\al}{{\alpha}}

\numberwithin{equation}{section}

\begin{document}
\title{Arctic curves of the octahedron equation}
\author{Philippe Di Francesco} 
\address{Department of Mathematics, University of Illinois at Urbana-Champaign, MC-382, 
1409 W\ Green St., Urbana, IL 61801, U.S.A.
e-mail: philippe@illinois.edu}
\author{Rodrigo Soto-Garrido}
\address{Department of Physics, University of Illinois at Urbana-Champaign, 1110 W.\ Green St., Urbana, IL  61801, U.S.A. 
e-mail: sotogar1@illinois.edu.}

\date{\today}
\begin{abstract}

We study the octahedron relation (also known as the $A_{\infty}$ $T$-system), obeyed in particular by the partition 
function for dimer coverings of the Aztec Diamond graph. For a suitable class of doubly periodic initial conditions,
we find exact solutions with a particularly simple factorized form. For these, we show that the density function that 
measures the average dimer occupation of a face of the Aztec graph, obeys a system of linear recursion relations 
with periodic coefficients. This allows us to explore the thermodynamic limit of the corresponding dimer models 
and to derive exact ``arctic" curves separating the various phases of the system.
\end{abstract}

\maketitle
\date{\today}
\tableofcontents

\section{Introduction}

The octahedron recurrence is a system of non-linear equations describing the evolution of a quantity $T_{i,j,k}$, $i,j,k\in \Z$
corresponding to discrete space $(i,j)$ and time $k$. This equation first arose in the context  of integrable quantum spin chains with a Lie group symmetry \cite{KNS,KNS11}, and is obeyed by the corresponding quantum transfer matrices. In this language, the octahedron equation corresponds to the so-called $T$-system (for $A_\infty$). In this formulation, the indices $i,k,j$ respectively stand for representation indices, and a discrete spectral parameter. The $A$-type
$T$-systems have remarkable properties, depending on the choice of boundary conditions, such as discrete integrability \cite{DFK09a}, and
periodicity properties \cite{AH,DFK13} as well as the positive Laurent property (solutions are Laurent polynomials of the initial data with 
non-negative integer coefficients) in relation to cluster algebras \cite{DF}.

This equation or some restrictions thereof appear to be central in a number of combinatorial constructs, 
such as the Desnanot-Jacobi relation between minors of a given matrix and the Dodgson condensation of determinants \cite{DOD},
the lambda-determinant and alternating sign matrices \cite{Bressoud,RR,DFLD}, the puzzles for computing Littlewood-Richardson coefficients  
\cite{KTW}, various generalizations of Coxeter-Conway frieze patterns \cite{Cox,FRISES,BER}, and cluster algebra 
\cite{DFK08} to name a few.

A great progress in understanding the combinatorics of the octahedron equation was due to Speyer \cite{SPY}, 
who worked out the general solution in terms of a dimer model on a graph, also equivalent to 
the domino tiling problem of the Aztec diamond \cite{EKLP}. This establishes the connection between 
a general set of solutions of the octahedron equation with given initial data, and the partition functions 
of statistical lattice models of dimers, whose local Boltzmann  weights are defined in terms of these data.
This was recently extended to more general initial conditions, giving rise to dimer models on specific graphs \cite{DF13}.
Note also that a large class of dimer models on periodic graphs was recently shown to have both integrable and cluster algebra structures
as well \cite{GK}. Finally, the study of the so-called pentagram map, an integrable dynamical system 
on polygons of projective plane displayed intriguing connections \cite{Glick} with solutions of the octahedron equation with special periodic
initial conditions. An analogous connection exists for the generalization to higher pentagram maps \cite{GSTV,DFK13,KV}.

Dimer models  were the subject of a lot of attention, starting with the so-called arctic circle theorem for the 
domino tilings of large Aztec diamonds \cite{JPS}, and later culminating in the global understanding of the arctic curve
phenomenon in the continuum limit \cite{KO} \cite{KOS}, where the phase diagram of the model was shown to exhibit separations
between frozen, disordered, and liquid phases. Analogous phenomena were observed for groves \cite{KPS}, for the 
double-dimer model connected to the hexahedron recurrence \cite{KEN}, for random walks \cite{ADM},
for square Young tableaux \cite{Romik}, 
and for the six-vertex model \cite{CNP}.

In the present paper, we revisit the case of domino tilings/dimer coverings of the Aztec diamond,  from the point of view of 
Speyer's general solution of the octahedron relation. We construct explicit exact solutions of the octahedron equation 
for an infinite class of initial data with special periodicity conditions. 
These in turn are partition functions for dimer models with periodic weights.
For these solutions,
we show that a certain local density function, that measures the average dimer occupation of a face of the Aztec graph, obeys a system of 
linear recursion relations with {\it periodic} coefficients. This allows to compute the density generating function explicitly in the form 
of a rational fraction. Following \cite{PW0204,PW08,PW13book,KEN}, the study of the denominator of this function allows to explore the singularity 
structure of the dimer models  in the thermodynamic limit of large size, and to confirm their phase structure, displaying frozen, 
disordered and facet-like phases separated by generalized ``arctic" curves.

The paper is organized as follows. 

In Section 2, we recall facts on the $A_\infty$ $T$-system/octahedron equation and its initial conditions, as well as its solution as a dimer partition function on the Aztec graph. As a preparatory exercise, we compute the density generating function for the uniform initial data and show how to extract the arctic circle form the explicit solution.

In Section 3, we study in detail $2\times 2$ periodic initial data, for which the $T$-system is wrapped on the torus generated by $(2,0)$ and $(0,2)$ in $\Z^2$. The density generating function is found to solve a linear $4\times 4$ system, whose explicit solution displays two disconnected pieces of ``arctic" curve, separating the dimer configurations in the thermodynamic limit into three phases: (i) four frozen corners with a single dominant configuration induced by the geometry of corners; (ii) a disordered ``temperate region" analog to the inside of the arctic circle, and (iii) a new ``facet"-like central phase, where the configurations are pinned to the sub-lattice corresponding to the faces with the smallest Boltzmann weight. 

Section 4 is devoted to the class of so-called $m$-toroidal boundary conditions. We first construct the explicit solution to the $T$-system, and then show that it leads to linear systems for the density function with {\it triply} periodic coefficients in $\Z^3$. As a result, the density generating function is shown to satisfy a $4m\times 4m$ linear system, whose determinant captures the information on the generalized arctic curves that separate different phases. For this model, we find generically
the same phases: (i) frozen corners; (ii) disordered region; (iii) in general $m-1$ facets  whose position and size vary with the initial data.

We gather a few concluding remarks in Section 5.

\medskip
\noindent{\bf Acknowledgments.} 
We thank E. Fradkin, M. Gekhtman, R. Kedem,  R. Kenyon and G. Musiker for discussions. We thank the referees for a 
very thorough reading of the manuscript, and many useful remarks.
P. D. F. thanks the Simons Center for Geometry and Physics 
for hospitality during the semester ``Conformal Geometry" in the early stages of this work. P. D. F. is supported by the NSF grant 
DMS 13-01636 and the Morris and Gertrude Fine endowment. R. S. G. is supported by the NSF grant DMR-1064319 and the DOE grant 
DE-FG02-07ER46453 at the University of Illinois.

\section{T-system, dimers and arctic curve}

\subsection{The $A_\infty$ $T$-system.}

The unrestricted $A_\infty$ $T$-system (from now on we will drop the $A_\infty$ label), also known as the octahedron recurrence, 
is given by the following difference equation (for a detailed review of the $T$-system see \cite{KNS11}):
\begin{equation}
 T_{i,j,k+1}T_{i,j,k-1}=T_{i+1,j,k}T_{i-1,j,k}+T_{i,j+1,k}T_{i,j-1,k}
 \label{Tsystem}
\end{equation}
where $i,j,k\in \Z$ and $T_{i,j,k}\in\R$. Notice that the system conserves parity $i+j+k=0,1$ mod 2. Let us fix it to $1$ 
throughout the paper. 
One way to think about the $T$-system is to consider $i,j$ as labeling points on the square lattice and $k$ as a discrete time 
(as the vertical axis of a cubic lattice). 
In this sense the $T$-system (\ref{Tsystem}) describes the evolution in time of a given initial data (see Ref. \cite{DF13}). 
In this paper
we will work with a flat initial data, meaning that the value of the $T$ variables is specified on the $(i,j,0)$ and $(i,j,1)$ planes,
namely we fix:
\begin{equation}\label{initdat} T_{i,j,i+j+1\, {\rm mod}\, 2}=t_{i,j} \qquad (i,j\in \Z) \end{equation}
In this paper, we consider the solution $T_{i,j,k}$ to the $T$-system \eqref{Tsystem} subject to various restrictions of the
initial condition \eqref{initdat}. These will be simply the result of imposing extra periodicity conditions on the initial data $t_{i,j}$,
the simplest of which being the uniform case when all $t_{i,j}=1$.

\subsection{Dimer formulation}

\begin{figure}
 \includegraphics[width=12.cm]{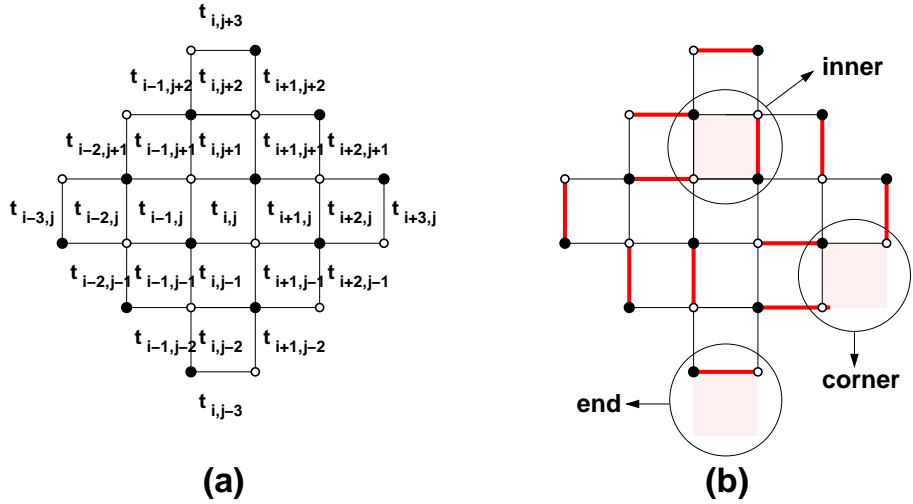}
\caption{\small A typical Aztec diamond graph ${\mathcal A_{i,j,k}}$ for $k=4$ (a) with its face labeling, and a sample dimer configuration on 
${\mathcal A_{i,j,k}}$ (b). We have shaded the three types of faces: inner, corner, end, respectively adjacent to 4,2,1 edges of ${\mathcal A}_{i,j,k}$.}
\label{fig:diamondim}
\end{figure}

The solution $T_{i,j,k}$, for $i+j+k=1$ mod 2, $i,j\in\Z$, $k\geq 0$, of the $A_\infty$ $T$-system \eqref{Tsystem} with initial data \eqref{initdat} 
was identified as the partition function for 
domino tilings of the Aztec diamond, or dually to the dimer coverings of the Aztec diamond graph ${\mathcal A}_{i,j,k}$ \cite{SPY,DF13}. 
The graph ${\mathcal A}_{i,j,k}$ has vertices at points of the lattice $\Z^2$.
Let us  label the faces of $\Z^2$ by the coordinate $(i,j)$ of their lower left corner vertex. Then ${\mathcal A}_{i,j,k}$ has faces $(a,b)$ such 
that $|a-i|+|b-j|\leq k-1$. Each such face receives the label $t_{a,b}$, the initial data assignment of the $T$-system. The edges of the graph are 
inherited from those of the underlying square lattice, 
however the boundary faces have only one or two adjacent edges depending on whether they are corner faces or end faces 
(see Fig. \ref{fig:diamondim} for an illustration).
The vertices of ${\mathcal A}_{i,j,k}$
are naturally bicolored black/white, according to the parity of $i+j=0/1$ mod 2. 

\begin{defn}
The dimer model on ${\mathcal A}_{i,j,k}$ is defined as follows. The configurations of the model are matchings of pairs of vertices connected
via an edge (dimers) such that any vertex of ${\mathcal A}_{i,j,k}$ belongs to exactly one  dimer. Each configuration is weighted by a product 
over all faces $(a,b)$ of ${\mathcal A}_{i,j,k}$ of local weights $w_{a,b}$. For any face $(a,b)$ (including boundary faces) the weight is 
$$w_{a,b}=(t_{a,b})^{1-D_{a,b}}$$
where $D_{a,b}\in \{0,1,2\}$ is the total number of dimers occupying the edges of the square face. The partition function of the model is
$$Z_{i,j,k}=\sum_{{\rm dimer}\atop {\rm configurations}} \prod_{{\rm faces}\ (a,b)} w_{a,b} $$
\end{defn}

The analysis of the present paper is based on the following main result:
\begin{thm}{\cite{SPY,DF13}}
The solution $T_{i,j,k}$, for $i+j+k=1$ mod 2, $i,j\in\Z$, $k\geq 0$, of the $A_\infty$ $T$-system \eqref{Tsystem} with initial data 
\eqref{initdat} is the partition function of the dimer model on ${\mathcal A}_{i,j,k}$, namely
$$ T_{i,j,k}=Z_{i,j,k}$$
\end{thm}

\subsubsection{Density}

Consider the solution $T_{i,j,k}$ to the $T$-system with initial data \eqref{initdat}.
If we think of $T_{i,j,k}$ as a partition function, then the derivative 
$t_{\epsilon,\eta}\partial_{t_{\epsilon,\eta}} {\rm Log}\, T_{i,j,k}$ corresponds to some susceptibility 
or density, where $t_{\epsilon,\eta}$ acts as a source,
here a magnetic field attached to dimers around the $(\epsilon,\eta)$ face. More precisely, we have
$$t_{\epsilon,\eta}\partial_{t_{\epsilon,\eta}} {\rm Log}\, T_{i,j,k} =\langle 1-D_{\epsilon,\eta} \rangle \ ,$$
the statistical average of $1-$ the number of dimers surrounding the face $(\epsilon,\eta)$ 
within the set of dimer configurations on ${\mathcal A}_{i,j,k}$.
Assume we further restrict the initial values to $t_{a,b}=t_{a,b}^*$, $a,b\in \Z$.

\begin{defn}
We define the density function $\rho_{i,j,k}^{(\epsilon,\eta)}$ as:
\begin{equation}
 \rho_{i,j,k}^{(\epsilon,\eta)}=t_{\epsilon,\eta}\left. \partial_{t_{\epsilon,\eta}}{\rm Log}\, T_{i,j,k}\right\vert_{t_{a,b}=t_{a,b}^*}
 \label{density}
\end{equation}
\end{defn}

We can easily derive a linear recurrence relation for $\rho_{i,j,k}^{(\epsilon,\eta)}$ by taking the derivative with respect to 
$t_{\epsilon,\eta}$ of the $T$-system relation \eqref{Tsystem}.
After some straightforward algebra we have:
\begin{equation} 
\rho_{i,j,k+1}^{(\epsilon,\eta)}+ \rho_{i,j,k-1}^{(\epsilon,\eta)}=L_{i,j,k}(\rho_{i+1,j,k}^{(\epsilon,\eta)}+\rho_{i-1,j,k}^{(\epsilon,\eta)})
+R_{i,j,k}(\rho_{i,j+1,k}^{(\epsilon,\eta)}+\rho_{i,j-1,k}^{(\epsilon,\eta)}) 
\label{densityequation}				
\end{equation}
where we used the notation:
\begin{equation}\label{LRdef}
L_{i,j,k}=\displaystyle\frac{T_{i+1,j,k}T_{i-1,j,k}}{T_{i,j,k+1}T_{i,j,k-1}}\qquad{\rm and }\qquad 
R_{i,j,k}=1-L_{i,j,k}=\displaystyle\frac{T_{i,j+1,k}T_{i,j-1,k}}{T_{i,j,k+1}T_{i,j,k-1}} 
\end{equation}
where $T_{i,j,k}$ are evaluated at $t_{a,b}=t_{a,b}^*$ for all $a,b\in \Z$.
The recurrence relation is supplemented with the following initial data. Define $\varphi=(\epsilon+\eta+1\ {\rm mod}\ 2)$. Then:
\begin{equation}\label{iniro}
\rho_{i,j,\varphi}^{(\epsilon,\eta)} =\delta_{i,\epsilon}\delta_{j,\eta} \qquad 
\rho_{i,j,1-\varphi}^{(\epsilon,\eta)}=0 \qquad (i,j\in \Z;i+j+\epsilon+\eta=0\, {\rm mod}\, 2)
\end{equation}

Notice that in order to solve the recurrence relation for the density we need to know $L_{i,j,k}$ and $R_{i,j,k}$ appearing 
in \eqref{densityequation}. As we shall see below, the density $\rho_{i,j,k}^{(\epsilon,\eta)}$ is the variable that we will use to explore 
the behavior of the dimer model for large $k$.

\subsection{Arctic curve: the uniform case}

We start with the recurrence relation (\ref{densityequation}) for the uniform initial data:
\begin{equation}
 t_{i,j}^*=1\qquad(i,j \in \Z)
\end{equation}
In this case the solution for the T-system is simply given by $T_{i,j,k}=2^{k(k-1)/2}$: this coincides with the partition function of
uniform domino tilings of an Aztec diamond of size $k$ \cite{EKLP}. The
two ratios $L_{i,j,k}$, $R_{i,j,k}$ appearing in eq. (\ref{densityequation}) are equal to $1/2$. 

To this uniform initial data we add up a source on the face $(\epsilon,\eta)=(0,0)$ or $(0,1)$,
namely consider the densities $\rho_{i,j,k}^{(0,0)}$, $\rho_{i,j,k}^{(0,1)}$.
Both densities obey the following recurrence relation:
\begin{equation}
\rho_{i,j,k+1}+\rho_{i,j,k-1}=\frac{1}{2}(\rho_{i+1,j,k}+\rho_{i-1,j,k}+\rho_{i,j+1,k}+\rho_{i,j-1,k}) \qquad (i,j\in \Z;k\geq 1)
 \label{densitytrivialdata}
\end{equation}
with initial data $\rho_{i,j,0}^{(0,0)}=0$ and $\rho_{i,j,1}^{(0,0)}=\delta_{i,0}\delta_{j,0}$, while
$\rho_{i,j,1}^{(0,1)}=0$ and $\rho_{i,j,0}^{(0,1)}=\delta_{i,0}\delta_{j,0}$. 
Extending the validity of \eqref{densitytrivialdata} to $k=0$ allows to define 
\begin{equation}\label{minusonero}\rho_{i,j,-1}^{(0,0)}=-\delta_{i,0}\delta_{j,0}\end{equation}
and 
\begin{equation}\label{minustworo}\rho_{i,j,-1}^{(0,1)}= \frac{1}{2}(\delta_{i,-1}+\delta_{i,1})\delta_{j,0}+ \frac{1}{2}(\delta_{j,-1}+\delta_{j,1})\delta_{i,0}
\end{equation}
For $(\epsilon,\eta)=(0,0)$ or $(0,1)$, we define generating functions $\rho^{(\epsilon,\eta)}(x,y,z)$ as:
\begin{equation}
 \rho^{(\epsilon,\eta)}(x,y,z)=\sum_{i,j\in\Z,k\geq 0}\rho_{i,j,k}^{(\epsilon,\eta)}\, x^i\, y^j\, z^k
 \label{genfun}
\end{equation}
Multiplying both sides of \eqref{densitytrivialdata} by $x^iy^jz^k$ and then summing over $i,j\in \Z,k\geq 0$ we get:
$$
\sum_{i,j\in \Z,k\geq 0}(\rho_{i,j,k+1}+\rho_{i,j,k-1})x^iy^jz^k
=\frac{1}{2}\sum_{i,j\in \Z,k\geq 0}(\rho_{i+1,j,k}+\rho_{i-1,j,k}+\rho_{i,j+1,k}+\rho_{i,j-1,k})x^iy^jz^k
$$
for both density functions. Substituting the values of $\rho_{i,j,-1}$ from (\ref{minusonero}-\ref{minustworo}), we get:
\begin{eqnarray*}
(z^{-1}+z)\rho^{(0,0)}(x,y,z)&=&\frac{1}{2}(x^{-1}+x+y^{-1}+y)\rho^{(0,0)}(x,y,z)+1 \\
(z^{-1}+z)\rho^{(0,1)}(x,y,z)&=&\frac{1}{2}(x^{-1}+x+y^{-1}+y)\rho^{(0,1)}(x,y,z)+z^{-1}-\frac{1}{2}(x^{-1}+x+y^{-1}+y)
\end{eqnarray*}

So the density generating functions for the uniform initial data are given by:
\begin{eqnarray}
\rho^{(0,0)}(x,y,z)&=&\frac{z}{1+z^2-\frac{z}{2}(x^{-1}+x+y^{-1}+y)}\nonumber \\
\rho^{(0,1)}(x,y,z)&=&\frac{1-\frac{z}{2}(x^{-1}+x+y^{-1}+y)}{1+z^2-\frac{z}{2}(x^{-1}+x+y^{-1}+y)}\nonumber \\
&=&1-z \, \rho^{(0,0)}(x,y,z)(x,y,z) \label{genfuntrivialdata}
\end{eqnarray}

Recall that $\rho^{(0,0)}_{i,j,k}$, $i+j+k=1$ mod 2, is the average $\langle 1-D_{0,0}\rangle_{i,j,k}$ in the dimer model on ${\mathcal A}_{i,j,k}$.
It is the same as the average $\langle 1-D_{-i,-j}\rangle_{0,0,k}$ (for $i+j$ even) in the dimer model on ${\mathcal A}_{0,0,k}$ when $k$ is odd
and as $\langle 1-D_{-i,-j+1}\rangle_{0,1,k}$ (for $i+j$ odd) in the dimer model on ${\mathcal A}_{0,1,k}$ when $k$ is even.
Therefore the generating function for the averages $\langle 1-D_{i,j}\rangle_{0,0,2k-1}$ on the even faces (with $i+j=0$ mod 2) 
of the dimer model on $A_{0,0,2k-1}$ reads:
$$\sum_{i,j\in \Z\atop i+j\ {\rm even}} \langle 1-D_{i,j}\rangle_{0,0,2k-1}\,  x^i \, y^j= \rho^{(0,0)}(x^{-1},y^{-1},z)\vert_{z^{2k-1}}
=\rho^{(0,0)}(x,y,z)\vert_{z^{2k-1}} $$
by use of the obvious symmetries $x\leftrightarrow x^{-1}$ and $y\leftrightarrow y^{-1}$ of $\rho^{(0,0)}$ \eqref{genfuntrivialdata},
and where the notation $f\vert_{z^m}$ stands for the coefficient of $z^m$ in the power expansion of $f$ as a series of $z$.
Similarly, the generating function for the averages $\langle 1-D_{i,j}\rangle_{0,1,2k}$ on the even faces
of the dimer model on $A_{0,1,2k}$ reads:
$$\sum_{i,j\in \Z\atop i+j\ {\rm even}} \langle 1-D_{i,j}\rangle_{0,1,2k} \, x^i \, y^j= y\, \rho^{(0,0)}(x^{-1},y^{-1},z)\vert_{z^{2k-1}}
=y\, \rho^{(0,0)}(x,y,z)\vert_{z^{2k}} $$
We have similar expressions for the averages on odd faces, involving $\rho^{(0,1)}$.

\begin{figure}
 \includegraphics[scale=.8]{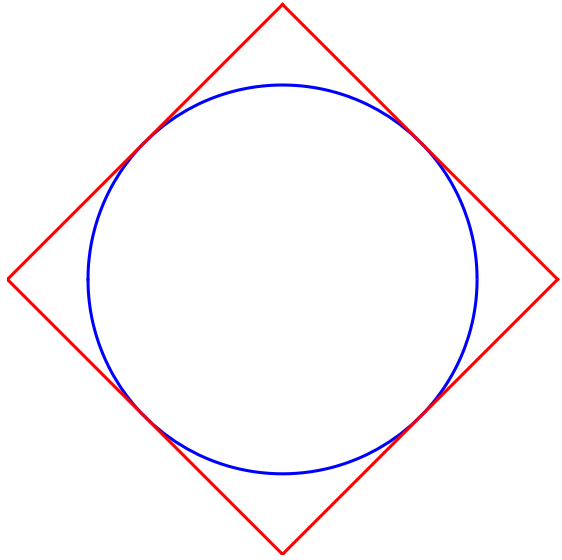}
\caption{\small Arctic circle for the trivial initial data. The corners are frozen and the center is disordered.}
\label{Arcticcircle}
\end{figure}

The singularities of these two expressions are determined by the denominator of $\rho^{(0,0)}(x,y,z)$ given by
(\ref{genfuntrivialdata}).

We wish to explore the behavior of the coefficients $\rho^{(0,0)}_{i,j,k}$ for large $i,j,k$ and $\frac{i}{k}=u$, $\frac{j}{k}=v$ finite.
Following the general theory of singularities of multivariate series \cite{PW0204,PW08,PW13book}, 
the limit is singular along the planar dual curve to the homogeneous polynomial part of the denominator 
of the generating functions at the critical point $x=y=z=1$.
To compute this curve, we may
blow up the singular point $x=y=z=1$ by taking $x\rightarrow 1-tx$, $y\rightarrow 1-ty$ and $z\rightarrow 1+t(ux+vy)$ and 
then expand the denominator in powers of $t$. Up to order $t^2$ we have:
\begin{equation*}
 z^{-1}+z-\frac{1}{2}(x^{-1}+x+y^{-1}+y)=\frac{t^2}{2}((2u^2-1)x^2+(2v^2-1)y^2+4uvxy)+O(t^3)
\end{equation*}
Let us define $H(x,y)=(2u^2-1)x^2+(2v^2-1)y^2+4uvxy$. The dual curve is obtained by imposing $H(x,y)=0$ and 
$\frac{\partial}{\partial x}H(x,y)=\frac{\partial}{\partial y}H(x,y)=0$. However, here and in the following, $H(x,y)$ is always a homogeneous
polynomial, henceforth $(\frac{\partial}{\partial x}+y\frac{\partial}{\partial y})H=m H$, where $m$ is the total degree, $m=2$ here.
We may therefore simply impose $H(x,y)=0$ and $\frac{\partial}{\partial x}H(x,y)=0$, and the last equation $\frac{\partial}{\partial y}H(x,y)=0$
is automatically satisfied.
Eliminating $x$ and $y$, we end up with the singularity locus:
\begin{equation}
 P(u,v)=2(u^2+v^2)-1=0
 \label{arcticcircle}
\end{equation}
This defines the arctic circle (see Fig. \ref{Arcticcircle}). This is the circle inscribed into the
square  domain $|u|+|v|=1$, which corresponds to the limiting domain of non-zero values of $\rho_{i,j,k}$ for $k\to\infty$
while $\frac{i}{k}=u$, $\frac{j}{k}=v$.


\section{Toroidal initial data I: the 2$\times$2 case}\label{sectwotwo}

\begin{figure}
\centering
\includegraphics[scale=.42]{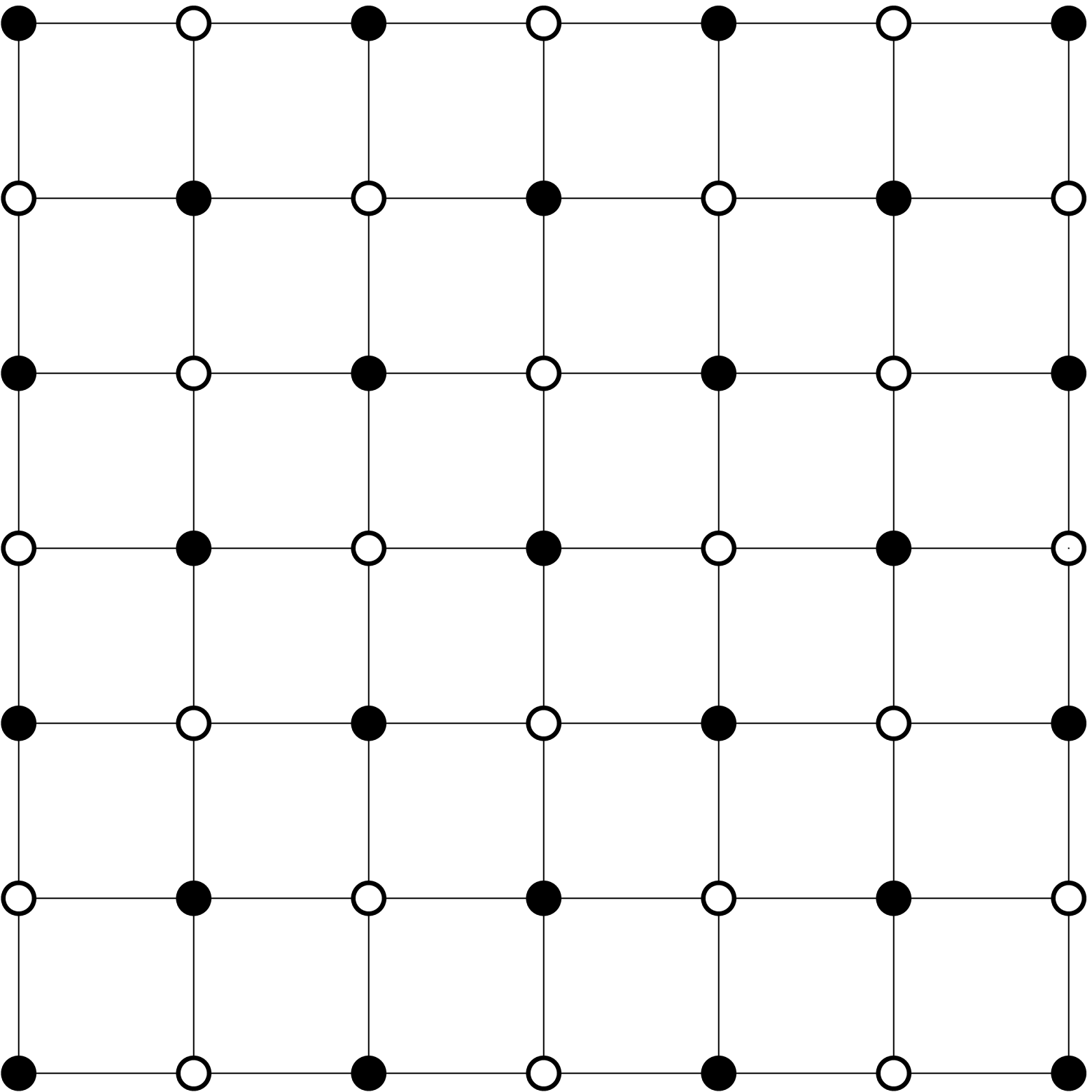}
\setlength{\unitlength}{0.5cm}
\multiput(-13,.4)(4.3,0){4}{\large{$b$}}
\multiput(-13,4.7)(4.3,0){4}{\large{$b$}}
\multiput(-13,9)(4.3,0){4}{\large{$b$}}
\multiput(-13,13.3)(4.3,0){4}{\large{$b$}}
\multiput(-13,2.6)(4.3,0){4}{\large{$d$}}
\multiput(-13,6.9)(4.3,0){4}{\large{$d$}}
\multiput(-13,11.2)(4.3,0){4}{\large{$d$}}
\multiput(-10.8,2.6)(4.3,0){3}{\large{$a$}}
\multiput(-10.8,6.9)(4.3,0){3}{\large{$a$}}
\multiput(-10.8,11.2)(4.3,0){3}{\large{$a$}}
\multiput(-10.8,.4)(4.3,0){3}{\large{$c$}}
\multiput(-10.8,4.7)(4.3,0){3}{\large{$c$}}
\multiput(-10.8,9)(4.3,0){3}{\large{$c$}}
\multiput(-10.8,13.3)(4.3,0){3}{\large{$c$}}
\put(-14,-.7){\vector(1,0){2}}
\put(-14,-.7){\vector(0,1){2}}
\put(-14,1.7){$j$}
\put(-11.9,-0.9){$i$}
\put(-6.7,6){\small{(0,0)}}
\caption{\small Initial data with period 2 in directions $i,j$. The white and black points corresponds to the planes $k=0$ and $k=1$ 
respectively.}
\label{periodicinitialdata}
\end{figure}

In this section we will focus on a very specific initial data, which has period two in 
both $i$ and $j$ directions, namely:
$$ t_{i+2,j}=t_{i,j}\qquad t_{i,j+2}=t_{i,j} $$
More precisely, we set (see Fig. \ref{periodicinitialdata}):
\begin{equation}\label{ini22}
 t_{i,j}=\left\{ \begin{matrix} 
a & {\rm if}\, i=0,j=0\, {\rm mod}\, 2\\
b & {\rm if}\, i=1,j=1\, {\rm mod}\, 2\\
c & {\rm if}\, i=0,j=1\, {\rm mod}\, 2\\
d & {\rm if}\, i=1,j=0\, {\rm mod}\, 2
\end{matrix}\right. 
\end{equation}
Remarkably, this particular initial data allows us to still find
an exact and simple solution for the $T$-system. At the same time it provides us with an illustrative example on how to
compute the arctic curves for non-uniform initial data. We will see that the ratios $L_{i,j,k}$ and $R_{i,j,k}$ appearing in (\ref{densityequation})
have a certain periodicity that allows to reduce the problem to a finite linear system of 4 equations (each for different values of the ratios).

\subsection{Exact solution of the $T$-system with 2$\times$2 periodic initial data}
The corresponding $T$-system actually coincides with the so-called $Q$-system for $\widehat{A_1}$.
The exact solution to this T-system with doubly periodic initial data is given by:

\begin{lemma}
\small{\begin{equation}\label{exasol22}
  T_{i,j,k}=
\left(\frac{a^2+b^2}{c d}\right)^{\left\lfloor \frac{k}{2}\right\rfloor  \left\lfloor
   \frac{k+1}{2}\right\rfloor } \left(\frac{c^2+d^2}{a b}\right)^{\left\lfloor
   \frac{k-1}{2}\right\rfloor  \left\lfloor \frac{k}{2}\right\rfloor }
   \times
\begin{cases}
t_{i,j}&, \text{ if } k=0,1,\ {\rm mod}\ 4  \\
t_{i+1,j+1}&, \text{ if }  k=2,3,\ {\rm mod}\ 4 
\end{cases}
\end{equation}}
\normalsize
\noindent for $t_{i,j}$ as in \eqref{ini22}.
\end{lemma}
\begin{proof} By direct substitution into the octahedron recurrence \eqref{Tsystem}, 
and inspection of the cases $k=0,1,2,3$ mod 4.
\end{proof}

\subsection{Density: exact derivation}
As explained before, we may consider various density functions $\rho^{(\epsilon,\eta)}_{i,j,k}$ that measure the average
$\langle 1-D_{\epsilon,\eta}\rangle_{i,j,k}$ on the face $(\epsilon,\eta)$ of the dimer model on ${\mathcal A}_{i,j,k}$.
For odd $k$ this is equal to the average $\langle 1-D_{\epsilon-i,\eta-j}\rangle_{0,0,k}$ of the dimer model on
${\mathcal A}_{0,0,k}$ (in which the central face is of $a$ type). 

Let us consider only even faces of type $a$ or $b$, namely $\epsilon=\eta=0$ or $1$, and $k$ odd.
Then the density $\rho^{(0,0)}(x^{-1},y^{-1},z)\vert_{x^{\rm even}y^{\rm even}}$ generates $\langle 1-D_{i,j}\rangle_{0,0,k}$ on 
$a$-type faces with both $i,j$ even, while $x y \left(\rho^{(1,1)}(x^{-1},y^{-1},z)\vert_{x^{\rm even}y^{\rm even}}\right)$
generates $\langle 1-D_{i,j}\rangle_{0,0,k}$ on $b$-type faces  with both $i,j$ odd. 
By considering odd powers of $x,y$ instead, we have also access to averages of $1-D_{i,j}$ for ${\mathcal A}_{1,1,k}$ 
(in which the central face is of $b$ type).

Now that we have a solution for the $T$-system, we can directly compute the ratios in the recurrence relation
for the density (\ref{densityequation}). We note that the ratio $L_{i,j,k}$ is periodic (as well as $R_{i,j,k}$, from the relation
$L_{i,j,k}+R_{i,j,k}=1$). We have indeed an obvious periodicity on the constant
$k$ planes. Defining $e_1=(2,0,0)$ and $e_2=(0,2,0)$, we have that $L_{(i,j,k)}=L_{(i,j,k)+me_1+ne_2}$ where $m,n\in\Z$
(the same for $R_{i,j,k}$).
However, by using the exact solution \eqref{exasol22}, we find another less obvious periodicity in the direction 
$e_3=(1,1,2)$ as well. This is summarized in the following:

\begin{lemma}\label{perio22lem}
The coefficients $L_{i,j,k},R_{i,j,k}$ corresponding to the $2\times 2$ periodic solution \eqref{exasol22} of the $T$-system
have the following periodicity:
$$L_{(i,j,k)}=L_{(i,j,k)+me_1+ne_2+pe_3}\qquad (m,n\in \Z,p\geq 0) . $$
and similarly for $R_{i,j,k}=1-L_{i,j,k}$.
\end{lemma}
\begin{proof} By inspection.
\end{proof}

Let us first derive $\rho^{(0,0)}(x,y,z)$.
Define the following partial generating functions:
\begin{equation}
 \rho^{(i_0,j_0,k_0)}(x,y,z)=\sum_{m,n\in \Z,p\geq 0}\rho^{(0,0)}_{(i_0,j_0,k_0)+me_1+ne_2+pe_3}\, x^{i_0+2m+p}y^{j_0+2n+p}z^{k_0+2p} \label{genfun2}
\end{equation}
where $(i_0,j_0,k_0)$ are in the unit cell for the periodicities of Lemma \ref{perio22lem}, 
namely $(i_0,j_0,k_0)\in P=\{(0,0,1),(1,1,1),(1,0,0),(0,1,0)\}$. 
In terms of these the full density generating function
is simply 
$$\rho^{(0,0)}(x,y,z)=\sum_{(i_0,j_0,k_0)\in P}\rho^{(i_0,j_0,k_0)}(x,y,z)$$
Using the recurrence relation for the density (\ref{densityequation}) and the initial conditions $\rho^{(0,0)}_{i,j,1}=\delta_{i,0}\delta_{j,0}$,
$\rho^{(0,0)}_{i,j,0}=0$, we end up with the following linear system:
\begin{equation}
\left(
\begin{array}{cccc}
 \frac{1}{z} & \frac{\left(x^2+1\right) (\tau -1)}{x} & -\frac{\left(y^2+1\right) \tau }{y} &
   z \\
   z & \frac{\left(y^2+1\right) (\tau -1)}{y} & -\frac{\left(x^2+1\right) \tau }{x} &
   \frac{1}{z} \\
 -\frac{\left(y^2+1\right) \sigma }{y} & \frac{1}{z} & z & \frac{\left(x^2+1\right) (\sigma
   -1)}{x} \\
 -\frac{\left(x^2+1\right) \sigma }{x} & z & \frac{1}{z} & \frac{\left(y^2+1\right) (\sigma
   -1)}{y}
\end{array}
\right)
\left(
\begin{array}{c}
\rho^{(0,0,1)}(x,y,z)\\
\rho^{(0,1,0)}(x,y,z)\\
\rho^{(1,0,0)}(x,y,z)\\
\rho^{(1,1,1)}(x,y,z)
\end{array}
\right)=
\left(
\begin{array}{c}
1\\
0\\
0\\
0
\end{array}
\right)
\label{2x2system}
\end{equation}
\noindent where we have used the parametrizations of weights:
$$\sigma=L_{1,0,1}=\frac{a^2}{a^2+b^2},\qquad  \tau=R_{0,0,0}=\frac{c^2}{c^2+d^2}\, ,$$ 
or equivalently:
$$a=b\sqrt{\frac{\sigma}{1-\sigma}}, \qquad {\rm and}\qquad  c=d\sqrt{\frac{\tau}{1-\tau}}\, .$$ 
It is worth noticing that even though we started with 4 arbitrary values  $a,b,c,d$ in our initial data in the $2\times2$
torus, the corresponding system for the density only depends on 2 parameters, the ratios $a/b$ and $c/d$. 
As we saw in the case of the uniform initial data, the arctic curve is determined by the zero locus of the denomiator
of the density functions. 
Here, this denominator is given by the determinant of the above system of 4 equations. 
Defining 
\begin{equation}\label{aldef}
\alpha=16\sigma(1-\sigma)\tau(1-\tau)=\frac{16}{\left(\frac{a}{b}+\frac{b}{a}\right)^2\left(\frac{c}{d}+\frac{d}{c}\right)^2}
\end{equation}
then, up to a factor of $xyz$, the determinant reads:
\small{
\begin{equation}\label{denomD}
D(x,y,z)=\frac{\alpha}{16}  (x^2-y^2)^2(x^2 y^2-1)^2z^4
-x^2 y^2 \left(x y-z^2\right)\left(y-xz^2\right)\left(x-y z^2\right) \left(1-xyz^2\right)
\end{equation}
}
\normalsize
\vskip-12pt
\noindent 
The actual density $\rho^{(0,0)}(x,y,z)$ however depends explicitly on $\sigma,\tau$, not just on $\alpha$. It has the form 
$\rho^{(0,0)}(x,y,z)=\frac{Q^{(0,0)}(x,y,z)}{D(x,y,z)}$,
with $D$ as in \eqref{denomD}, and $Q^{(0,0)}$ the following polynomial of $x,y,z$:
\footnotesize{\begin{eqnarray}Q^{(0,0)}(x,y,z)&=& x y z \left\{(-x^2 y^2 (1 - z^2) ( z (x (1 - x y z^2) + y (x y - z^2)) + x y (1 - z^4))\right. \nonumber \\
&&-x  y (x - y) (1 - x y) z^2 (x (1 - x y z^2) + y (x y - z^2)) \sigma + x^2  y^2 (x - y) (1 - x y) z (1 - z^4) \tau  \nonumber \\
&&\left. + x y (x^2 - y^2) (1 - x^2 y^2) z^2 (1 - z^2)\sigma\tau 
- (x - y) (x^2 - y^2) (1 - x y) (1 - x^2 y^2) z^3 \sigma\tau (1 -\tau)\right\} \label{numQ}
\end{eqnarray}}
\normalsize

Similarly, the density $\rho^{(1,1)}(x,y,z)$
solves the same system \eqref{2x2system}, but with the r.h.s. replaced by $(0,x y,0,0)^t$, due to the initial conditions
$\rho^{(1,1)}_{i,j,1}=\delta_{i,1}\delta_{j,1}$ and $\rho^{(1,1)}_{i,j,0}=0$. Alternatively, $\rho^{(1,1)}/(x y)$ is
obtained by interchanging $a\leftrightarrow b$ and $c\leftrightarrow d$ in the expression for $\rho^{(0,0)}$, namely by performing
the substitutions $\sigma\to 1-\sigma$ and $\tau\to 1-\tau$. These produce a new numerator $Q^{(1,1)}$, 
but leave the denominator $D(x,y,z)$ unchanged.

\subsection{Arctic curve}
Using the same procedure as for the uniform initial data, we expand the denominator $D(x,y,z)$ around the critical point $x=y=z=1$.
Taking $x\rightarrow1-tx$, $y\rightarrow1-ty$ and $z\rightarrow1+t(ux+vy)$, we find 
$D(1-tx,1-ty,1+t(ux+vy))=t^4 H(x,y)+O(t^5)$ at leading order in $t$ (which in this case turns out to be $t^4$). 
Imposing again $H(x,y)=0$ and $\frac{\partial}{\partial x}H(x,y)=0$, we can eliminate $x$ and $y$. 
We finally get the singularity curve $P_\al(u,v)=0$, where:
\footnotesize{
\begin{equation}\label{fortress}
\begin{split}
 P_\al(u,v)=&(1-\alpha)^3+16 \alpha ^2 (u^8+v^8)+8(4-5 \alpha ) \alpha(u^6+v^6)
 +32 \left(\alpha ^2+2 (2-\alpha )^2\right) u^4 v^4\\
 &+\left((4-\alpha )^2-24 \alpha \right) (1-\alpha ) \left(u^4+v^4\right)+8\left(6 \alpha ^2-(4-\alpha )^2\right) u^2 v^2 
 \left(u^2+v^2\right)\\
 &+2\left(48-(4-\alpha )^2\right) (1-\alpha ) u^2 v^2-2 (1-\alpha )^2 (4-\alpha )(u^2+v^2)+64(2-\alpha)\alpha u^2v^2(u^4+v^4)
\end{split}
\end{equation}}

\normalsize
\vskip-1pt
\noindent
Notice that this polynomial depends only on the single parameter $\al$ of \eqref{aldef}. 
Let us examine a few limiting cases of interest.

\begin{figure}
        \centering
             \hbox{ 
                \includegraphics[width=3.85 cm]{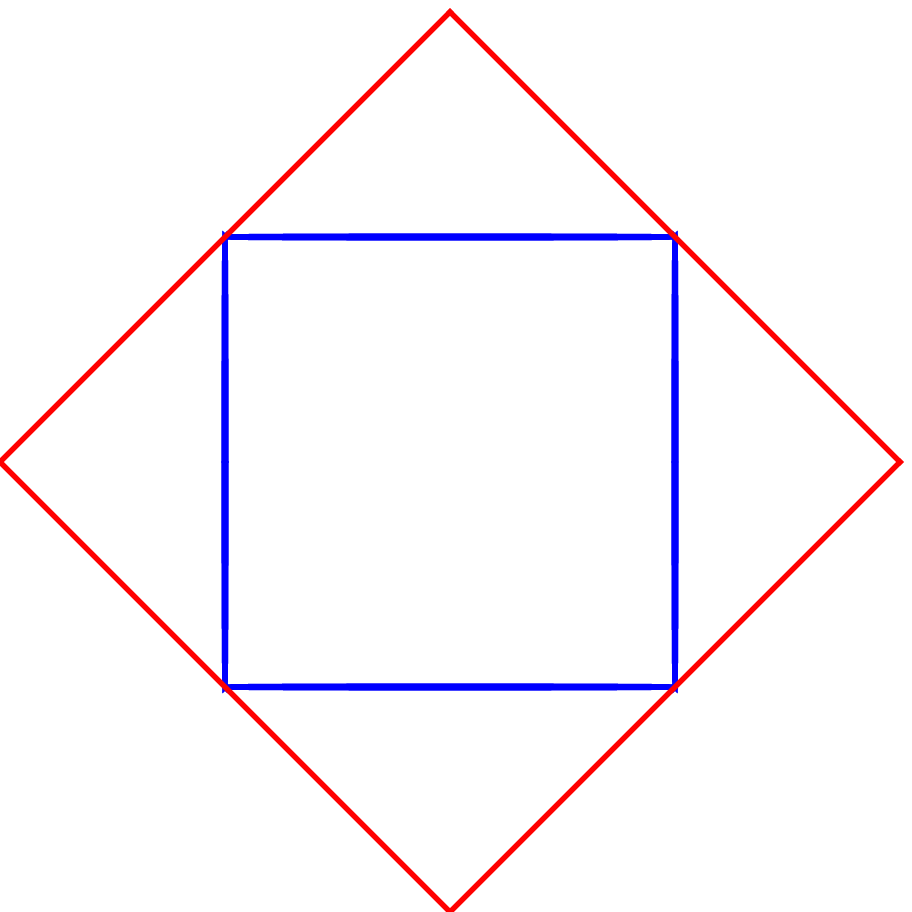}
               \put(-70,-15){$\alpha\rightarrow0$}
              \hskip .2cm
               \includegraphics[width=3.85 cm]{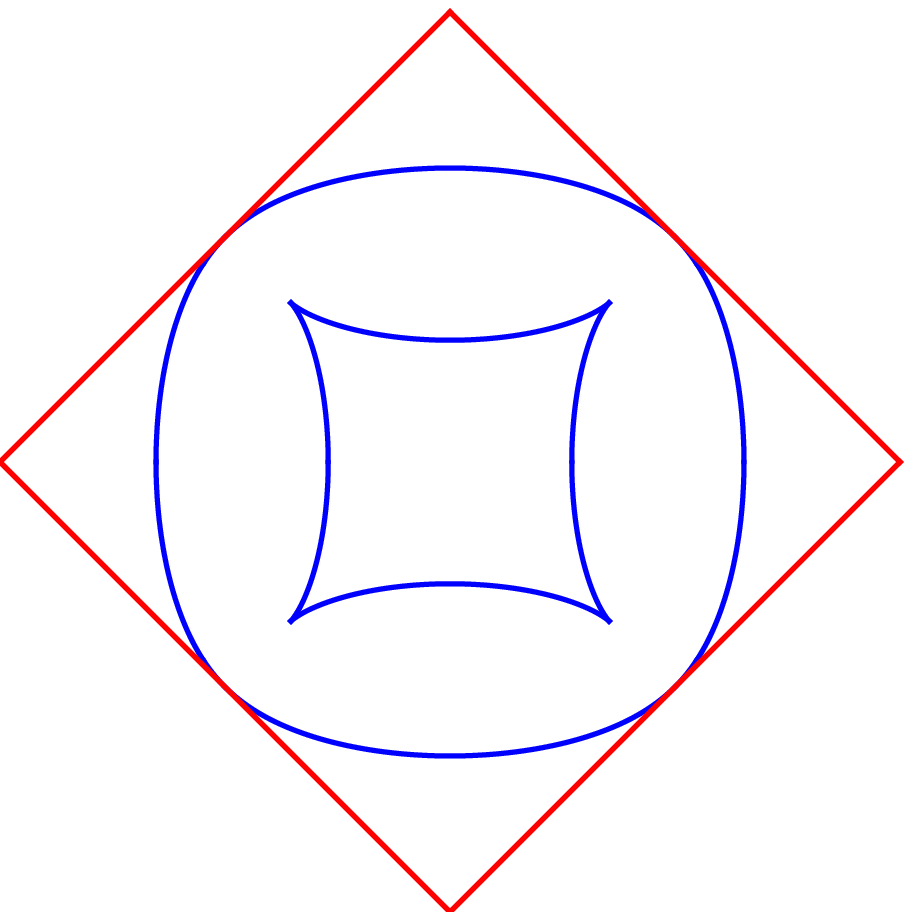}
               \put(-70,-15){$\alpha=1/2$}
	      \hskip .2cm
              \includegraphics[width=3.85 cm]{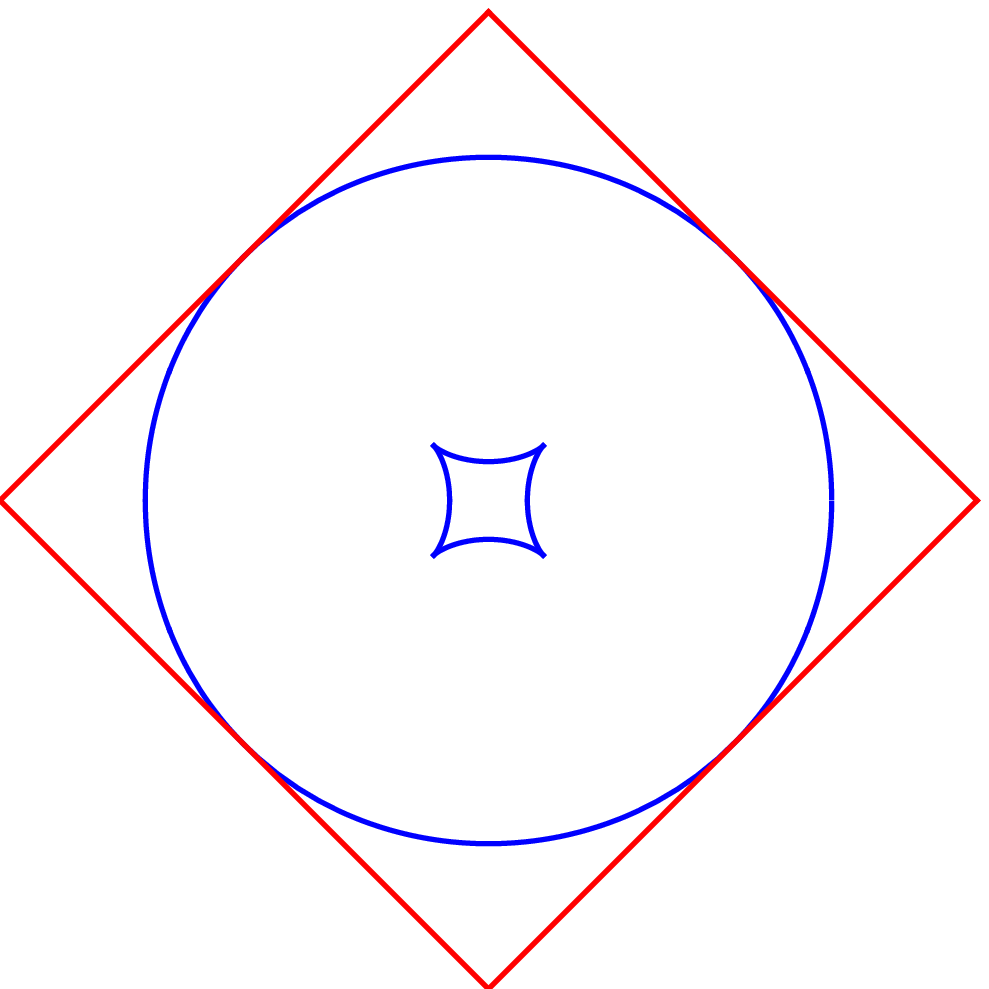}
              \put(-70,-15){$\alpha=19/20$}
              \hskip .2cm
              \includegraphics[width=3.85 cm]{ArcticCircle}
              \put(-69,-15){$\alpha=1$}}
        \caption{\small Arctic curves for the periodic initial data corresponding to different values of $\alpha$.}
        \label{periodicarcticcurve}
\end{figure}

For $\al=1$, the curve degenerates into:
$$P_1(u,v)= 8(u^2+v^2)^3(2u^2+2v^2-1) $$
and we recover the same result as in the uniform case ($\sigma=\tau=1/2$), 
namely the arctic circle $2u^2+2v^2=1$ of Fig.\ref{Arcticcircle}.

For $\al=0$, the curve degenerates into:
$$P_0(u,v)=(4u^2-1)^2(4v^2-1)^2$$
namely into the square with edges $u=\pm \frac{1}{2}$, $|v|\leq \frac{1}{2}$ and $v=\pm \frac{1}{2}$, $|u|\leq \frac{1}{2}$,
inscribed into the domain $|u|+|v|=1$. 

We have represented two more somewhat generic cases, with $\alpha=\frac{1}{2}$ and $\al=1-\frac{1}{20}$ in 
Fig. \ref{periodicarcticcurve}. 
In addition to the actual external arctic curve tangent to the square $|u|+|v|=1$ at the 4 points
$(u,v)=(\pm\frac{1}{2},\pm\frac{1}{2})$, we note the existence of an internal curve with 4 cusps 
at positions $(u,v)=(\pm \frac{\sqrt{1-\al}}{2},\pm \frac{\sqrt{1-\al}}{2})$ along the $u=\pm v$ lines. 
In addition to the frozen and temperate regions, we obtain a new ``bubble" inside, often called the {\it facet} domain.
The facet domain disappears exactly at $\al=1$, in which case 
we are left with simply the arctic circle, whereas it is ``maximal" at $\alpha=0$, where it becomes an 
inscribed square, and gets identified with the external arctic curve, so that the temperate region is squeezed and disappears.
The parameter $\al$ clearly governs the size of this facet domain.

The phase structure with a central facet shown in Fig.\ref{periodicarcticcurve}
coincides with that found for the ``square-octagon fortress" of Example 5.2 of \cite{KO} (see also Figure 18), 
for the value $\alpha=16/25$. This model indeed corresponds to a uniform solution of the octahedron equation,
but with a different initial data stepped surface, namely $T_{i,j,k_{i,j}}=1$, with
$$ k_{i,j}=\left\{ \begin{matrix} 1& {\rm if}\ i+j=0\ {\rm mod}\ 2\\
0 & {\rm if}\ (i,j)=(0,1)\ {\rm mod}\ 2\\
2 & {\rm if}\ (i,j)=(1,0)\ {\rm mod}\ 2
\end{matrix}\right.$$
which in turn corresponds to $a=b=c=1$, $d=2$.

\subsection{Physical interpretation and phase diagram/limit shape}
\begin{figure}
\centering
\includegraphics[width=7.cm]{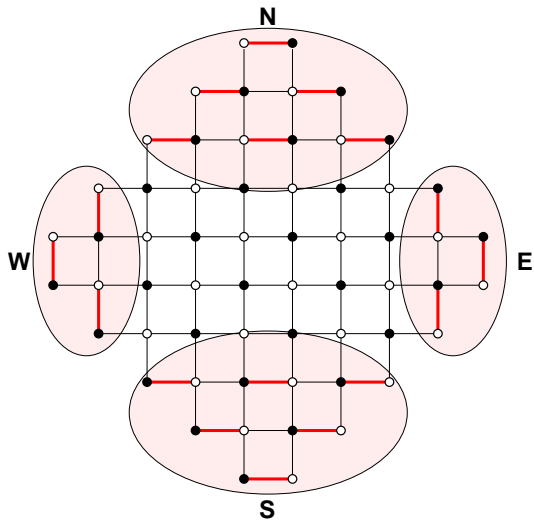}
\caption{\small The four frozen dimer configurations pertaining to the four (N,S,E,W) corners of the Aztec diamond graph.}
\label{fig:corners}
\end{figure}

\subsubsection{Uniform case}
We have seen that the density $\rho_{i,j,k}$ \eqref{density} is a measure of the expectation value, within the statistical ensemble of dimer 
configurations of the Aztec domain of size $k$,
of the observable $1-D_{i,j}$, where $D_{i,j}$ is the number of dimers occupying the edges around a given face $(i,j)$
of the domain with fixed parity of $i+j$. It therefore measures the imbalance between the empty squares configurations ($D_{i,j}=0$) and the 
maximally occupied ones ($D_{i,j}=2$).

The asymptotics of the coefficients  $\rho_{i,j,k}^{(0,0)}$ for large $i,j,k$ with $i/k=u$ and $j/k=v$ 
of the density generating series $\rho^{(0,0)}(x,y,z)$ \eqref{genfuntrivialdata} for the trivial initial data 
can be extracted by using for instance general theorems of Baryshnikov and Pemantle  
\cite{BP} (Theorem 3.7).
The result reads:
\begin{equation}\label{genfuntrivialdata2}
 \rho^{(0,0)}_{i,j,k}=\sim \frac{2}{\pi k} 
\frac{\delta^{[2]}_{i+j+k,1}}{\sqrt{1-2(u^2+v^2)} }
\end{equation} 
where we use the notation $\delta^{[p]}_{i,j}=\delta_{i-j,0\ {\rm mod}\, p}$.
This scaling function $\rho_{i,j,k}\sim \nu(i,j,k)$ appeared in \cite{CEP} (see p.26, where it is found to obey the differential equation 
$\frac{\partial^2}{\partial z^2}\nu=\frac{1}{2}\Big(\frac{\partial^2}{\partial x^2}+\frac{\partial^2}{\partial y^2}\Big)\nu$).

\begin{figure}
 \includegraphics[width=8.cm]{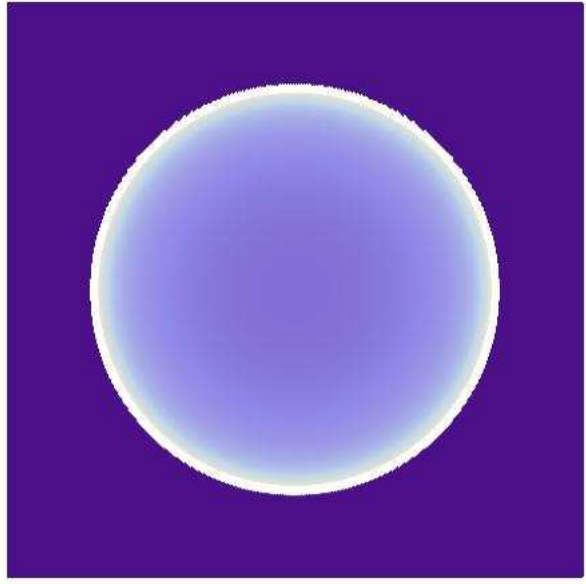}
\caption{\small Density profile for the coefficients $\rho_{i,j,k}^{(0,0)}$ given in \eqref{genfuntrivialdata2} 
for $k=211$ and $-k\leq i,j \leq k$.}
\label{figdensitycoeff}
\end{figure}

\noindent We display in Fig. \ref{figdensitycoeff} the density profile for this asymptotic value of $\rho_{i,j,k}^{(0,0)}$.

\begin{figure}
 \includegraphics[width=8.cm]{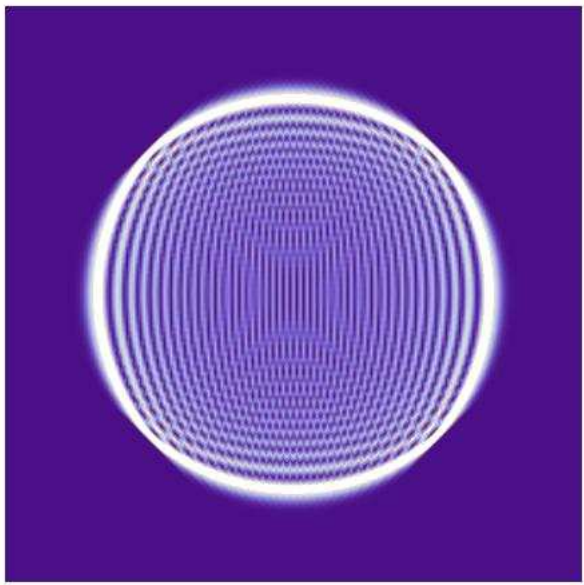}
\caption{\small Density profile for the trivial initial data. The dark color corresponds to $\rho_{i,j,k}=0$ and brighter color to larger 
$\rho_{i,j,k}$. The scale of colors is arbitrary.}
\label{densitycircle}
\end{figure}

The standard explanation for the arctic circle in the uniform initial data case is that the dimer configurations 
that contribute to the partition function $T_{i,j,k}$ tend
to be in a fundamental crystalline state in the vicinity of the corners of the square domain $|i|+|j|\leq |k|$. There are
four distinct such states, each corresponding to a (N,S,E,W) corner, characterized by an occupation number $D_{i,j}=1$
on each face $(i,j)$ (see Fig.\ref{fig:corners}).
Away from the corners, the dimer model has a non-trivial entropy,
and the competition between order and disorder gives rise to a separating critical curve in the continuum thermodynamic limit
when $k\to\infty$ with $\frac{i}{k}=u$ and $\frac{j}{k}=v$ fixed, between a frozen phase (next to the corners) and a so-called temperate 
phase (in the center). 
Outside of the critical curve, the density decays exponentially as $k\to\infty$
to $0$, as each square tends to be occupied by a single dimer, while
inside the curve it decreases as a power law $\propto \frac{1}{k}$, whereas the coefficient tends to the non-zero function 
$2/(\pi\sqrt{1-2(u^2+v^2)})$, singular on the arctic curve. 
This function indicates a growing local disorder
in the dimer configurations, maximum at the center of the Aztec domain.
We have represented the values of the rescaled density function $k|\rho_{i,j,k}^{(0,0)}|$ for fixed $k=85$, 
$-k\leq i,j \leq k$ and $i+j+k=1$ mod 2.

\subsubsection{$2\times 2$ periodic case}\label{22periosec}

To understand the emergence of a new (facet) central phase, let us first consider the simple case $\sigma=0$. 
We saw that in that case, the arctic curve degenerates to an inscribed square and the temperate region disappears.
This is attained
for instance by fixing $b,c,d>0$ and letting $a\to 0$ in the various density functions. 

\begin{figure}
\centering
\includegraphics[width=5.cm]{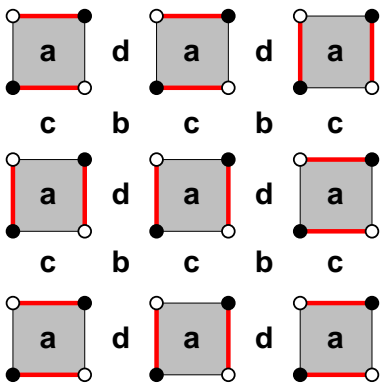}
\caption{\small The typical dominant dimer configurations for odd $k$ when $a\to 0$: each $a$-type face is occupied by two dimers,
with two arbitrary choices of orientation on each $a$-type face.}
\label{fig:ground22}
\end{figure}

It is clear that if $a$ is very small, the Boltzmann weight of maximally occupied dimer configurations around  the
$a$ faces becomes the dominant contribution to the partition function (see Fig.\ref{fig:ground22} for an illustration). 
We expect therefore a phase where the $a$ type
faces are occupied by two dimers, with arbitrary (vertical or horizontal) orientation. This phase is globally crystalline, from
the pinning of the dimers to the $a$-type faces that form a square sublattice, but retains some non-trivial entropy, from the
arbitrary orientation of the pair of dimers at each site, hence the name facet. 
Note that this also imposes another square sublattice of empty faces.
The corresponding value of $\langle 1-D\rangle$ for $k$ odd is $-1$ on the former sublattice, and $1$ on the latter.

\begin{figure}
\centering
\includegraphics[width=8.cm]{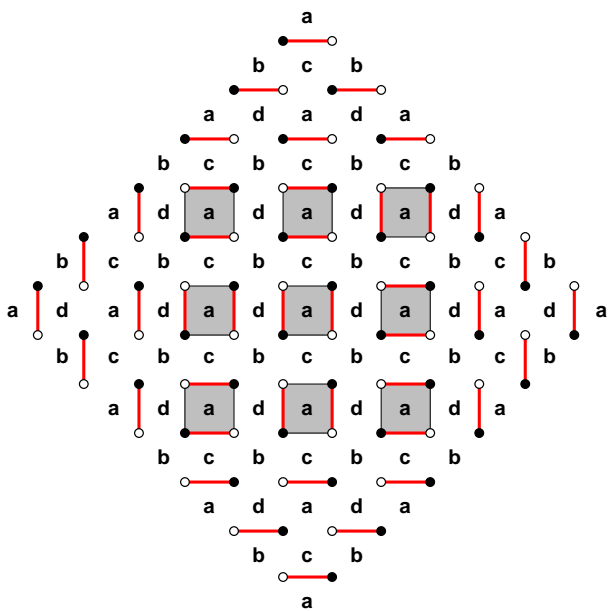}
\caption{\small The typical dominant dimer configurations for $k=7$ when $a\to 0$: each $a$-type face in the inscribed square region
is occupied by two dimers,
with two arbitrary choices of orientation on each $a$-type face (shaded).}
\label{fig:insquare}
\end{figure}

More precisely, for $\sigma=0$ the solution of the system \eqref{2x2system} and its companion for $\epsilon=\eta=1$ above lead to:
\footnotesize{
\begin{eqnarray*}\rho^{(0,0)}(x,y,z)&=& \frac{z}{(xy -z^2)(x-yz^2)(y-xz^2)(1-xyz^2)} \times \\
&&\left\{ x^2 y^2 - x y z^2 - x^3 y z^2 - x^2 y^2 z^2 - x y^3 z^2 - 
 x^3 y^3 z^2 - x^2 y z^3 - x y^2 z^3 - x^3 y^2 z^3 \right.  \\
 &&  \left. - x^2 y^3 z^3 +
 x^2 z^4 + x^2 y^2 z^4 + x^2 y^4 z^4 + x^2 y z^5 + x y^2 z^5 + 
 x^3 y^2 z^5 + x^2 y^3 z^5 + x^2 y^2 z^6 \right\}\\
&&+\tau\frac{z^3(x^2-y^2)(1-x^2y^2)}{(xy -z^2)(x-yz^2)(y-xz^2)(1-xyz^2)} 
\end{eqnarray*}
}
\normalsize
\vskip-12pt
\noindent 
and:
\footnotesize{
\begin{eqnarray*}
\frac{\rho^{(1,1)}(x,y,z)}{x y}&=&  \frac{z}{(xy -z^2)(x-yz^2)(y-xz^2)(1-xyz^2)} \times \\
&&\left\{ x^2 y^2 + x^2 y z + x y^2 z + x^3 y^2 z + x^2 y^3 z + y^2 z^2 + 
 x^2 y^2 z^2 + x^4 y^2 z^2 - x^2 y z^3  \right.  \\
&&  \left.- x y^2 z^3- x^3 y^2 z^3 - 
 x^2 y^3 z^3 - x y z^4 - x^3 y z^4 - x^2 y^2 z^4 - x y^3 z^4 - 
 x^3 y^3 z^4 + x^2 y^2 z^6 \right\}\\
&&-\tau\frac{z^3(x^2-y^2)(1-x^2y^2)}{(xy -z^2)(x-yz^2)(y-xz^2)(1-xyz^2)} 
\end{eqnarray*}
}
\normalsize
\vskip-12pt
\noindent 
with $\tau=\frac{c^2}{c^2+d^2}$.
As explained above, the final generating function for $\langle 1-D_{i,j}\rangle_{0,0,2k-1}$ for $i+j$ even is obtained by  
extracting the even powers of $x,y$ from $\rho^{(0,0)}$ and the odd powers of $x,y$ from $\rho^{(1,1)}/(x y)$,
namely by forming:
\begin{eqnarray*}U(x,y,z)&=&\frac{1}{4}(A(x,y,z)+A(-x,y,z)+A(x,-y,z)+A(-x,-y,z))\\
&&+\frac{1}{4}(B(x,y,z)-B(-x,y,z)-B(x,-y,z)+B(-x,-y,z))
\end{eqnarray*}
where $A=\rho^{(0,0)}$, and $B=\rho^{(1,1)}/(x y)$ above. 
Similarly, the generating function for $\langle 1-D_{i,j}\rangle_{1,1,2k}$ for $i+j$ even is obtained by  
extracting the odd powers of $x,y$ from $\rho^{(0,0)}$ and the even powers of $x,y$ from $\rho^{(1,1)}/(x y)$,
namely by forming:
\begin{eqnarray*}V(x,y,z)&=&\frac{1}{4}(A(x,y,z)-A(-x,y,z)-A(x,-y,z)+A(-x,-y,z))\\
&&+\frac{1}{4}(B(x,y,z)+B(-x,y,z)+B(x,-y,z)+B(-x,-y,z))
\end{eqnarray*}

Let us denote by $f_k(x,y)$ the coefficient of $z^k$ in the series expansion of $f(x,y,z)$, and by $[n]_x=\frac{x^n-x^{-n}}{x-x^{-1}}$
and similarly for $y$. Noting the generating functions
\begin{eqnarray*}\sum_{n\geq 0} [n]_x[n]_y z^n &=&\frac{x^2 y^2z(1-z^2)}{(xy -z^2)(x-yz^2)(y-xz^2)(1-xyz^2)}\\
\sum_{n\geq 1} [n+1]_x[n-1]_y z^n&=&\frac{y z^2
 (y + x^2 y + x^4 y - x z - x^3 z - x y^2 z - x^3 y^2 z + x^2 y z^2)}{(xy -z^2)(x-yz^2)(y-xz^2)(1-xyz^2)} 
\end{eqnarray*}
and expressing $U(x,y,z)$ and $V(x,y,z)$ in terms of these,
we finally get:
\begin{eqnarray*}
U_{4k-1}(x,y)&=&[2k]_x[2k]_y-[2k-1]_x[2k-1]_y\\
U_{4k-3}(x,y)&=&\tau ([2k-2]_x[2k]_y-[2k-3]_x[2k-1]_y)\\
&&\qquad +(1-\tau)([2k]_x[2k-2]_y-[2k-1]_x[2k-3]_y)\\
V_{4k-1}(x,y)&=&\tau ([2k-1]_x[2k+1]_y-[2k-2]_x[2k]_y)\\
&&\qquad +(1-\tau)([2k+1]_x[2k-1]_y-[2k]_x[2k-2]_y)\\
V_{4k-3}(x,y)&=&[2k-1]_x[2k-1]_y-[2k-2]_x[2k-2]_y
\end{eqnarray*}
Recall that the $U$'s correspond to the averages in the dimer model on Aztec graphs with a central face of $a$ type, while
the $V$'s correspond to the averages in the dimer model on Aztec graphs with a central face of $b$ type.
The case $U_{4k-1},V_{4k-3}$ display an alternance of $\pm 1$ on (even, even)/(odd,odd) faces.
This is in agreement with the typical dominant configuration represented in Fig.\ref{fig:insquare}, corresponding to $U_7(x,y)$:
the facet occupies exactly the inscribed square  $|i|,|j|\leq 3$, while outside this domain each face is occupied by a single dimer,
i.e. we have four corners frozen in their respective fundamental states with zero entropy.
The cases $U_{4k-3},V_{4k-1}$ also have a central facet square region with the same alternance of $\pm 1$, but have a 
thin boundary region around the square where the averages explicitly depend on $\tau$.

We conclude that asymptotically the density is identically zero outside of the inscribed square $|u|,|v|\leq \frac{1}{2}$, while inside 
it takes finite nonzero values that alternate on two square sublattices. 
The arctic curve is nothing but the exact phase separation, here reduced to the inscribed square $|u|=\frac{1}{2},|v|\leq \frac{1}{2}$
and $|v|=\frac{1}{2},|u|\leq \frac{1}{2}$.

\begin{figure}
        \centering
             \hbox{ 
                \includegraphics[width=3.85 cm]{2x2densitycircle}
               \put(-70,-15){$\alpha=1$}
              \hskip .2cm
               \includegraphics[width=3.85 cm]{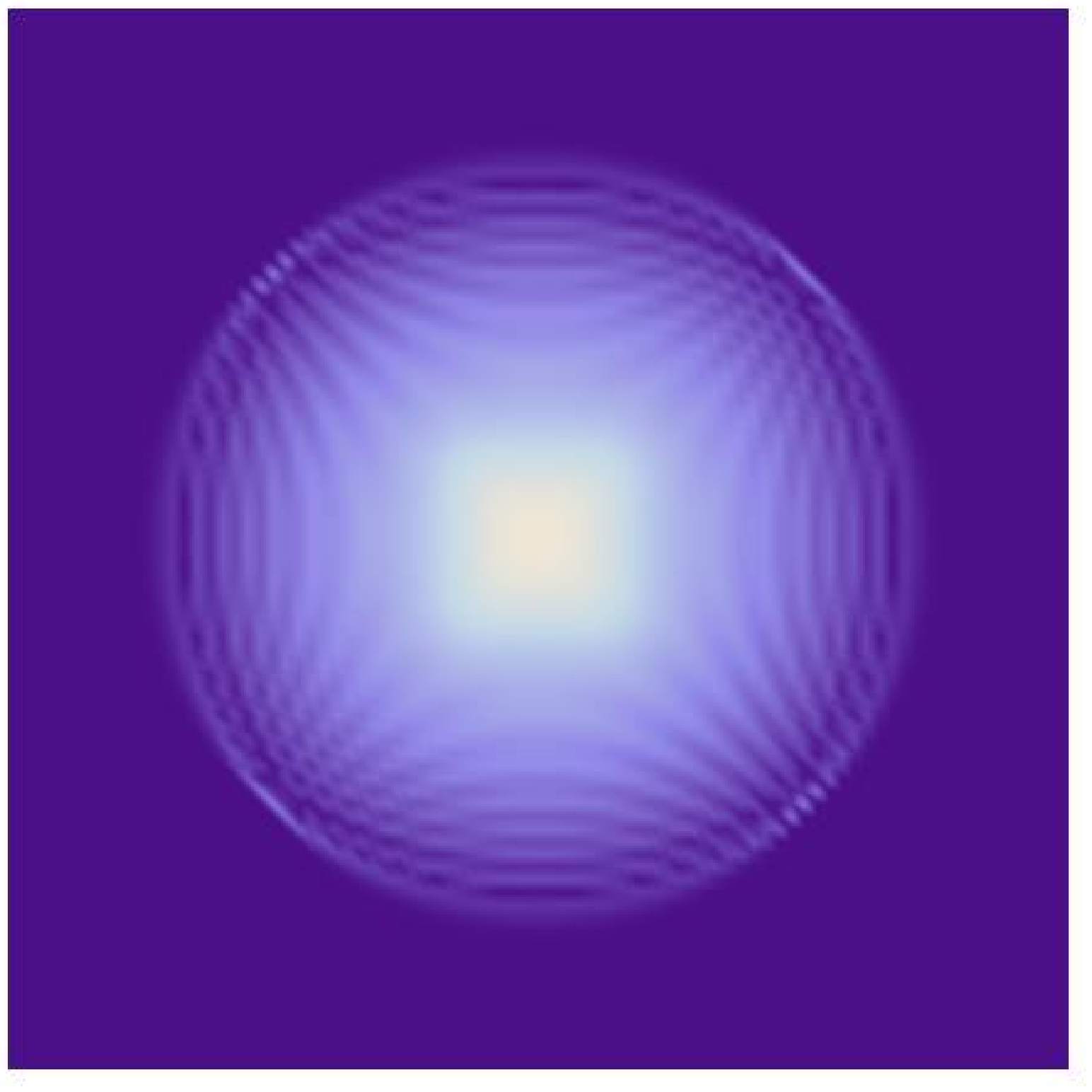}
               \put(-75,-15){$\alpha=0.98$}
	      \hskip .2cm
              \includegraphics[width=3.85 cm]{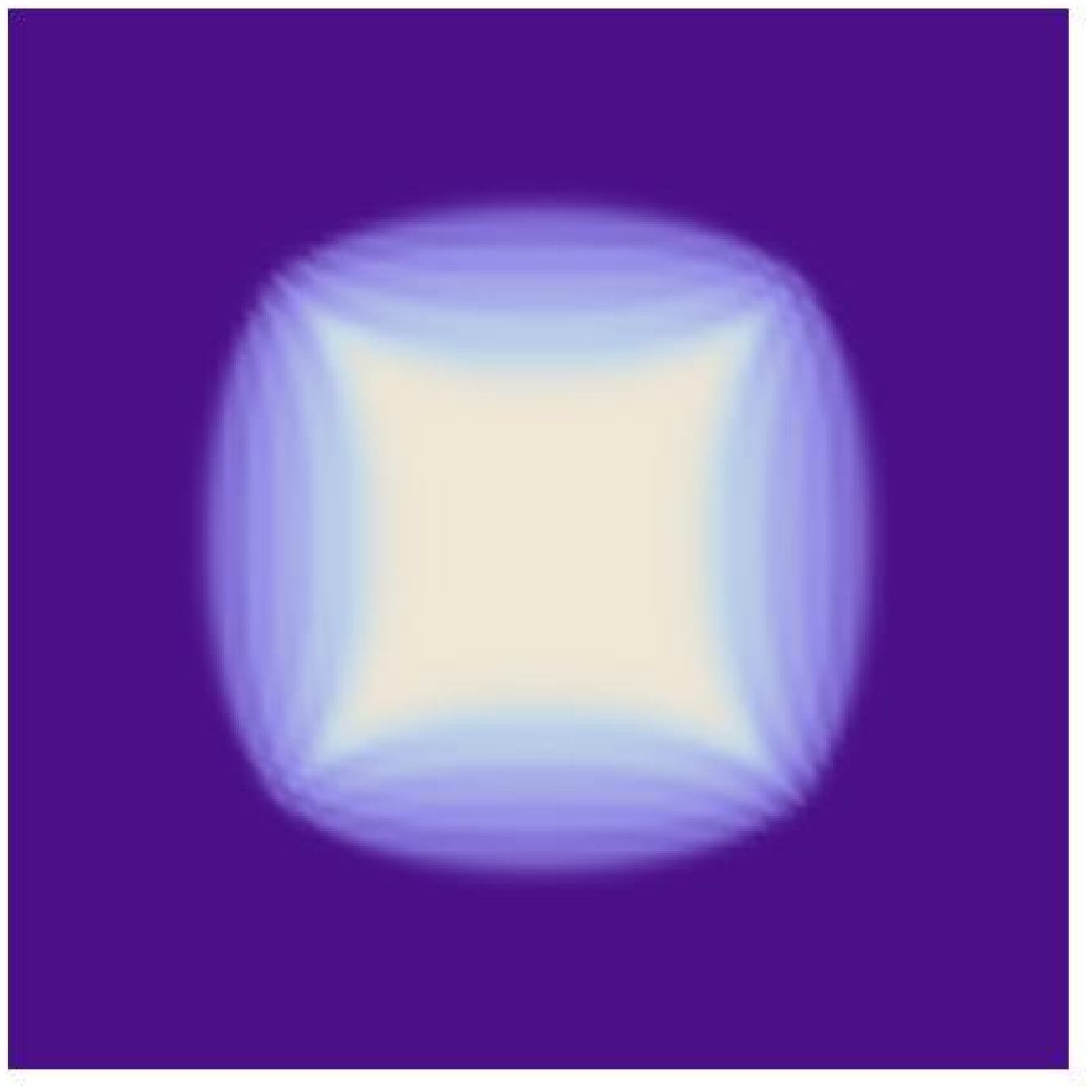}
              \put(-72,-15){$\alpha=0.5$}
              \hskip .2cm
              \includegraphics[width=3.85 cm]{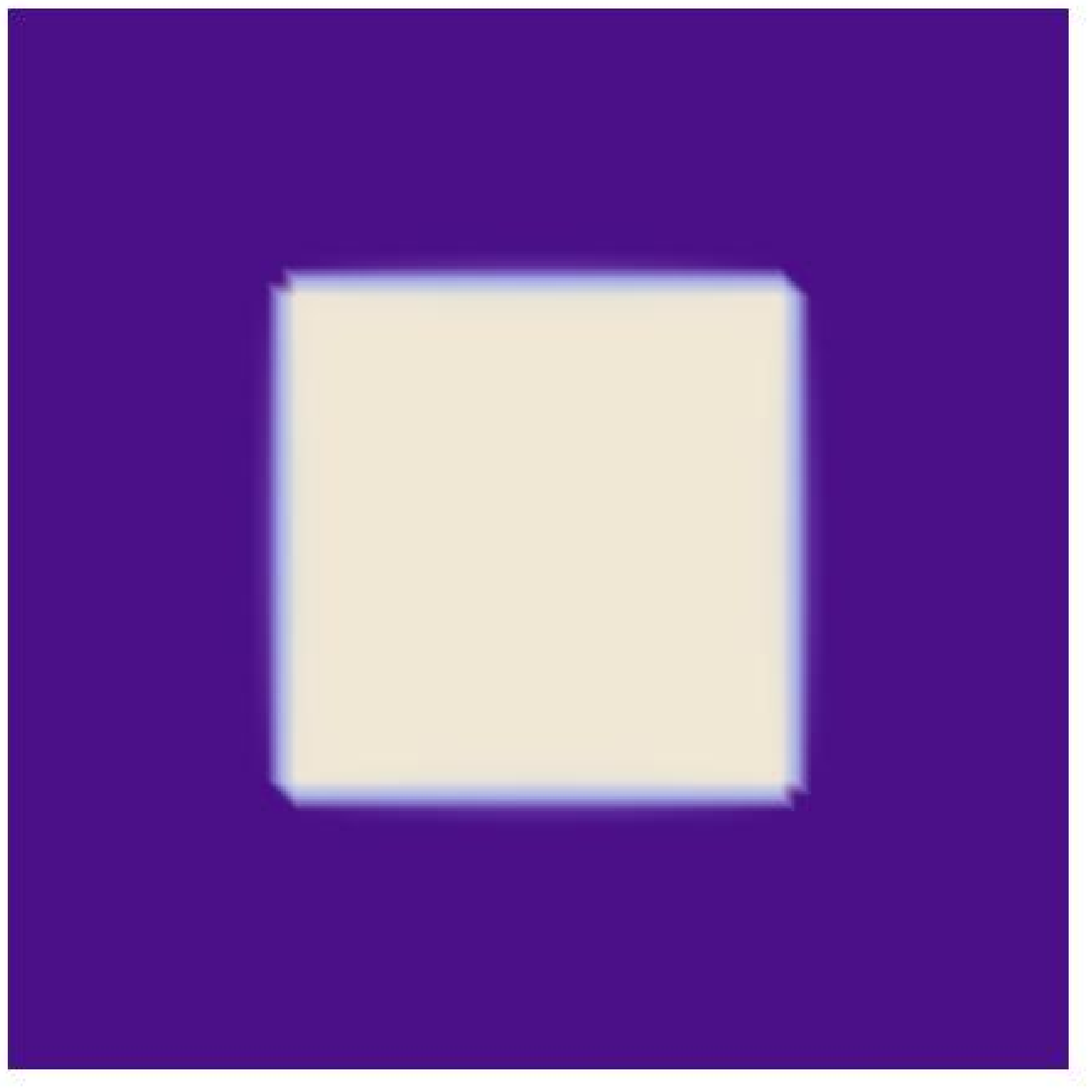}
              \put(-72,-15){$\alpha=0.01$}}
        \caption{\small Density profiles for $\rho_{i,j,k}$ for the periodic initial data corresponding to different values of $\alpha$. We fixed
        $b=c=d=1$ and vary $a$ accordingly to get the $\alpha$ shown in the profiles (see \eqref{aldef})}
        \label{densityprofiles}
\end{figure}

More generally, when $1>\sigma,\tau>0$, we still expect frozen corner phases and a central facet phase induced by the pinning on one
sublattice of configurations of pairs of parallel dimers, corresponding to the smallest weights among $a,b$ and $c,d$. 
The plot of the arctic curve \eqref{fortress} shows that a disordered phase separates the facet 
from the frozen corners. In both the frozen corners and the facet, the convergence of $\rho_{i,j,k}$ for $i/k=u,j/k=v$ fixed 
is exponential in $k$.
As explained above, the behavior of $\rho^{(0,0)}$ in the disordered phase is connected to the singularity $x=y=z=1$ of both
numerator \eqref{numQ} and denominator \eqref{denomD}. We find that the leading orders in the $t$ expansion for
$x\to 1-t x,y\to 1-t y,z\to1-t z$ are respectively $Q^{(0,0)}\sim t $ if $\sigma\neq \frac{1}{2}$, and $Q^{(0,0)}\sim t^2 $ otherwise,
while $D\sim t^4$ in all cases. We deduce that if $\sigma=\frac{1}{2}$ then $\rho_{i,j,k}$ for $i/k=u,j/k=v$ fixed tends to $0$ 
algebraically, as $k^{-1}$, and diverges at the boundary of the temperate zone (both along the facet border and the frozen corners border).
However if $\sigma\neq \frac{1}{2}$, we find that $\rho_{i,j,k}$ tends to a scaling function without any global rescaling.
We display in Fig.\ref{densityprofiles} a picture of the values of $|\rho_{i,j,k}^{(0,0)}|$ for size $k=85$, 
$-k\leq i,j \leq k$ and both $i$ and $j$ even.

\section{Toroidal initial data II: the $m$-toroidal case}

In this section, we introduce the $T$-system with initial data wrapped on a torus involving $4m$ 
arbitrary initial values. This particular choice is exactly solvable, and leads to the exact derivation
of higher degree arctic curves.

\subsection{Exact solution of the $T$-system with $m$-toroidal initial data}

\begin{defn}
Let us consider the following condition on the initial data $\{t_{i,j}\}_{i,j\in\Z}$ of the 
$T$-system \eqref{Tsystem}:
\begin{equation}\label{init2m} t_{i+m,j-m}=t_{i,j}\qquad {\rm and} \qquad 
t_{i+2,j+2}=t_{i,j}\qquad (i,j\in \Z)
\end{equation}
Initial data with this property will be called {\it $m$-toroidal boundary conditions}. 
The corresponding torus of the $\Z^2$ plane is generated by the two vectors
$\vec{e}_1=(m,-m,0)$ and $\vec{e}_2=(2,2,0)$.
\end{defn}

It is easy to show that any solution of the $T$-system with $m$-toroidal boundary conditions satisfies the same 
toroidal conditions, namely that:
$T_{i+m,j-m,k}=T_{i,j,k}$ and
$T_{i+2,j+2,k}=T_{i,j,k}$
for all $i,j\in\Z$ and $k\in \Z_+$, and $i+j+k=1$ mod 2.

Quite remarkably, there is an explicit expression for the solution of the $T$-system for $m$-toroidal boundary conditions. 

\begin{figure}
\centering
\includegraphics[width=10.cm]{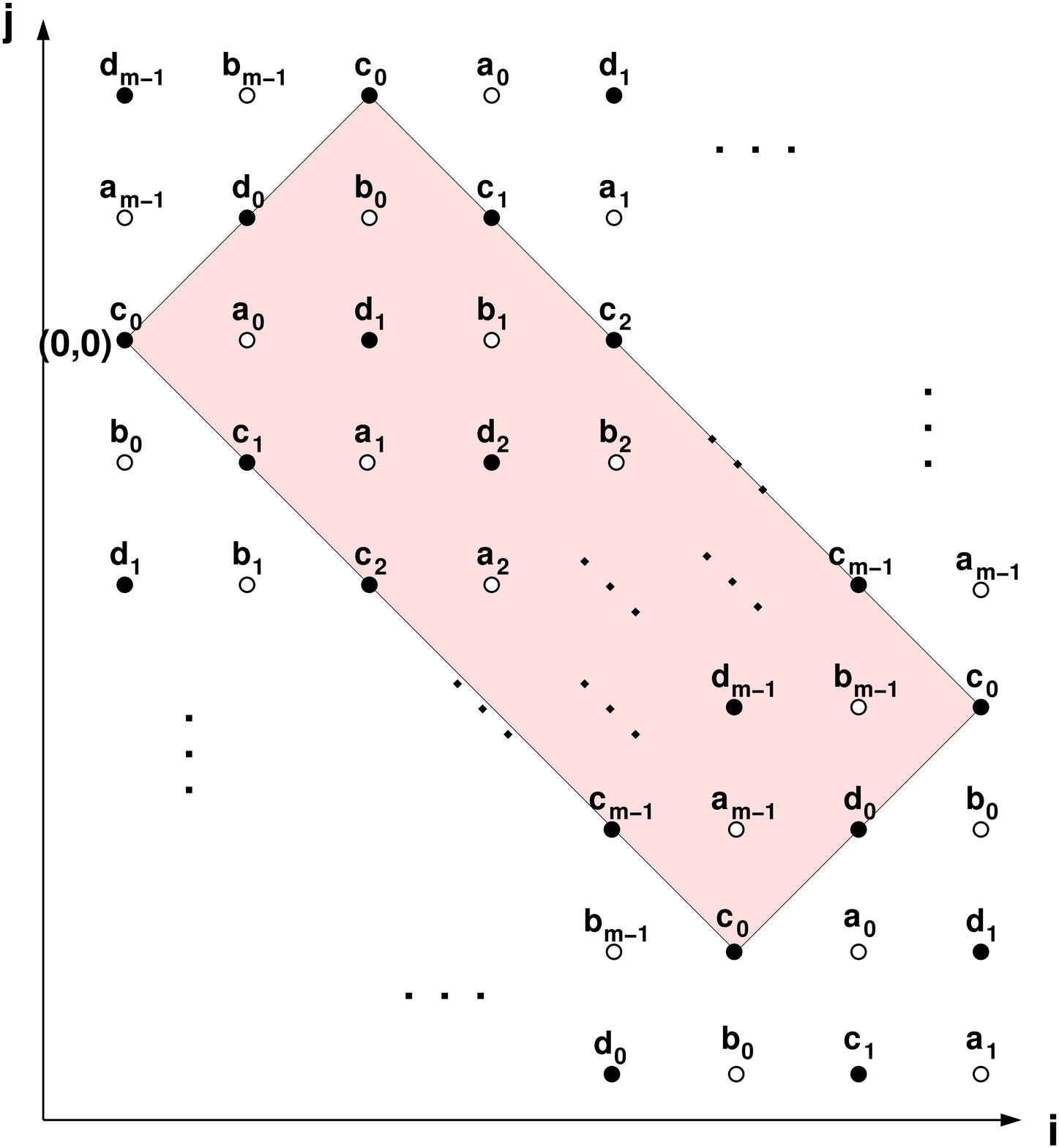}
\caption{\small The initial data of the $m$-toroidal $T$-system in the $(i,j)$ plane. The empty circles correspond to $k=0$
and the filled circles to $k=1$. We have shaded a fundamental domain for the corresponding torus.}
\label{fig:grid2m}
\end{figure}

Let us first denote
respectively by $a_i,b_i,c_i,d_i$ the initial data corresponding to a fundamental domain in the planes $k=0$ and $k=1$,
namely (see Fig.\ref{fig:grid2m} for an illustration):
\begin{eqnarray}
a_i &=& T_{i+1,-i,0}=t_{i+1,-i}\qquad 
b_i = T_{i+2,-i+1,0}=t_{i+2,-i+1}\nonumber \\
c_i &=& T_{i,-i,1}=t_{i,-i}\qquad \quad
d_i = T_{i+1,-i+1,1}=t_{i+1,-i+1}\label{abcd}
\end{eqnarray}
for $i\in \Z$.
The sequences $a_i,b_i,c_i,d_i$ $i\in \Z$ are clearly periodic with period $m$, as a direct consequence of the conditions \eqref{init2m}.
Let us further introduce two $m$-periodic sequences $x_i,y_i$ defined as:
\begin{equation}
x_i=\frac{c_{i} d_{i+1}+c_{i+1}d_{i}}{a_i b_i} \qquad {\rm and}\qquad y_i=  \frac{a_{i-1}b_{i}+a_{i}b_{i-1}}{c_i d_i}\quad (i\in \Z)
\end{equation}
Note that we have:
\begin{equation*} T_{i+1,-i,2}=x_i b_i  \qquad T_{i+2,-i+1,2}=x_i a_i
\end{equation*}
We also define for $n\geq 1$ and $i\in \Z$ the quantities:
\begin{equation}
u_{n,i}=\prod_{\ell=0}^{n-1} \left(x_{i-\ell-1}\right)^{\frac{n+1}{2}-\big\vert \frac{n-1}{2}-\ell\big\vert}\qquad 
v_{n,i}=\prod_{\ell=0}^{n-1} \left(y_{i-\ell-1}\right)^{\frac{n+1}{2}-\big\vert \frac{n-1}{2}-\ell\big\vert}
\end{equation}
with the convention that $u_{0,i}=u_{-1,i}=u_{-2,i}=1$ and similarly for $v$.
Finally, let $\theta_{i,j,k}$ be defined for $i,j\in \Z$ and $k\in \Z_+$ as:
\begin{equation}
\theta_{i,j,k}=T_{i+\lfloor \frac{k}{2} \rfloor,j+\lfloor \frac{k}{2} \rfloor,k\, {\rm mod}\, 2}
\end{equation}
In particular, for $k=0,1$ we have from the initial data \eqref{initdat}:
\begin{equation}\label{inithet}
 \theta_{i,j,i+j+1\, {\rm mod}\, 2}=t_{i,j} . \end{equation}
 
With the above definitions, the following theorem gives the exact value of the solution 
$T_{i,j,k}$ of the $T$-system with $m$-toroidal boundary conditions.

\begin{thm}\label{solmper}
With the above definitions for $u_{n,i},v_{n,i},\theta_{i,j,k}$,
the solution $T_{i,j,k}$ to the $T$-system with $m$-toroidal boundary conditions  \eqref{init2m} given by \eqref{abcd}  reads explicitly:
\begin{equation}\label{exactsol} T_{i,j,k}=u_{k-1,\frac{i-j+k-1}{2}}\,  v_{k-2,\frac{i-j+k-1}{2}} \, \theta_{i,j,k} 
\end{equation}
Moreover, we have the following explicit values for the cross-ratios $L_{i,j,k}$ and $R_{i,j,k}$, with the notations
$\al=\frac{i-j}{2}$, $\beta=\frac{i-j-1}{2}$, and $\delta^{[p]}_{i,j}=\delta_{i-j,0\, {\rm mod}\, p}$:
\begin{eqnarray*}
L_{i,j,k}&=&\frac{T_{i+1,j,k}T_{i-1,j,k}}{T_{i,j,k+1}T_{i,j,k-1}}=
\delta^{[4]}_{i+j+k,0}\left(\delta^{[2]}_{k,0}\, \frac{a_\al b_{\al-1}}{a_\al b_{\al-1}+a_{\al-1} b_{\al}}+\delta^{[2]}_{k,1}\,
\frac{c_{\beta+1} d_{\beta}}{c_\beta d_{\beta+1}+c_{\beta+1} d_{\beta}}  \right)\\
&&\qquad \qquad\qquad\quad +\delta^{[4]}_{i+j+k,2}\left(\delta^{[2]}_{k,0}\,  \frac{a_{\al-1} b_{\al}}{a_\al b_{\al-1}+a_{\al-1} b_{\al}} +\delta^{[2]}_{k,1}\,
\frac{c_{\beta} d_{\beta+1}}{c_\beta d_{\beta+1}+c_{\beta+1} d_{\beta}}  \right)\\
R_{i,j,k}&=&\frac{T_{i,j+1,k}T_{i,j-1,k}}{T_{i,j,k+1}T_{i,j,k-1}}= 1-L_{i,j,k}
\end{eqnarray*}
\end{thm}
\begin{proof}
For $T_{i,j,k}$ as in the statement of the theorem, let us compute the ratios
$L_{i,j,k}$ and $R_{i,j,k}$.
To this end, we note that by definition:
\begin{equation}\label{ratiou}\frac{u_{k,\al}}{u_{k-1,\al}}=\prod_{\ell=0}^{\lfloor \frac{k-1}{2} \rfloor} x_{\al-k+\ell} \quad {\rm and}\quad 
\frac{u_{k-1,\al-1}\, u_{k-1,\al}}{u_{k-2,\al-1}\, u_{k,\al}}
=\left\{ \begin{matrix}1 & {\rm if}\, k \, {\rm even}\\ 1/x_{\al-\frac{k+1}{2}} & {\rm otherwise}\end{matrix}\right.\end{equation}
and similarly for $v$:
\begin{equation}\label{ratiov}
\frac{v_{k-2,\al-1}\, v_{k-2,\al}}{v_{k-3,\al-1}\, v_{k-1,\al}}
=\left\{ \begin{matrix}1/y_{\al-\frac{k}{2}} & {\rm if}\, k \, {\rm even}\\ 1 & {\rm otherwise}\end{matrix}\right.\end{equation}
Let us pick $i,j,k$ such that $i+j+k=0$ mod 2, and use the result for the cross-ratios of (\ref{ratiou}-\ref{ratiov}) for
$\al=\frac{i-j+k}{2}$.
Noting finally that 
\begin{eqnarray*}
\frac{\theta_{i+1,j,2k}\theta_{i-1,j,2k}}{\theta_{i,j,2k+1}\theta_{i,j,2k-1}}&=&\left\{ \begin{matrix} \frac{a_{\al} b_{\al-1}}{c_\al d_\al}& {\rm if}\, 
\frac{i+j+2k}{2} \, {\rm even} \\
 \frac{a_{\al-1} b_{\al}}{c_\al d_\al}& {\rm otherwise} \end{matrix} \right. \ \left( \al=\frac{i-j}{2}\right)\\
\frac{\theta_{i+1,j,2k-1}\theta_{i-1,j,2k-1}}{\theta_{i,j,2k}\theta_{i,j,2k-2}}&=&\left\{ \begin{matrix}
\frac{c_\beta d_{\beta+1}}{a_\beta b_\beta}& {\rm if}\,  
\frac{i+j+2k-1}{2} \, {\rm even} \\
\frac{c_{\beta+1} d_{\beta}}{a_\beta b_\beta} & {\rm otherwise} \end{matrix} \right. \left( \beta=\frac{i-j-1}{2}\right)\\
\end{eqnarray*}
we conclude that
\begin{eqnarray*}
L_{i,j,2k}&=& \left\{ \begin{matrix} \frac{a_\al b_{\al-1}}{a_\al b_{\al-1}+a_{\al-1} b_{\al}} &\  {\rm if}\ i+j+2k=0 \, {\rm mod}\, 4\\ 
 \frac{a_{\al-1} b_{\al}}{a_\al b_{\al-1}+a_{\al-1} b_{\al}} &\  {\rm if}\ i+j+2k=2 \, {\rm mod}\, 4\\ \end{matrix}\right.\qquad \left(\alpha=\frac{i-j}{2}\right)
\\
L_{i,j,2k-1}&=& \left\{ \begin{matrix} \frac{c_{\beta+1} d_{\beta}}{c_\beta d_{\beta+1}+c_{\beta+1} d_{\beta}} &\  
{\rm if}\ i+j+2k-1=0 \, {\rm mod}\, 4\\ 
 \frac{c_\beta d_{\beta+1}}{c_\beta d_{\beta+1}+c_{\beta+1} d_{\beta}} &\  
 {\rm if}\ i+j+2k-1=2 \, {\rm mod}\, 4\\ \end{matrix}\right. \left(\beta=\frac{i-j-1}{2}\right)
\end{eqnarray*}
Similarly we compute the quantities:
\begin{eqnarray*}
R_{i,j,2k}&=& \left\{ \begin{matrix}  \frac{a_{\al-1} b_{\al}}{a_\al b_{\al-1}+a_{\al-1} b_{\al}} &\  {\rm if}\ i+j+2k=0 \, {\rm mod}\, 4\\ 
\frac{a_\al b_{\al-1}}{a_\al b_{\al-1}+a_{\al-1} b_{\al}}  &\  {\rm if}\ i+j+2k=2 \, {\rm mod}\, 4\\ \end{matrix}\right. \qquad \left(\alpha=\frac{i-j}{2}\right)
\\
R_{i,j,2k-1}&=& \left\{ \begin{matrix} \frac{c_\beta d_{\beta+1}}{c_\beta d_{\beta+1}+c_{\beta+1} d_{\beta}} &\  
{\rm if}\ i+j+2k-1=0 \, {\rm mod}\, 4\\ 
 \frac{c_{\beta+1} d_{\beta}}{c_\beta d_{\beta+1}+c_{\beta+1} d_{\beta}} &\  {\rm if}\ i+j+2k-1=2 \, 
 {\rm mod}\, 4\\ \end{matrix}\right. \left(\beta=\frac{i-j-1}{2}\right)
\end{eqnarray*}
We conclude that $R_{i,j,k}+L_{i,j,k}=1$ for all $k\geq 0$ and $i,j\in \Z$, and therefore $T_{i,j,k}$ satisfies the $T$-system \eqref{Tsystem}.
Moreover, the initial values of $T_{i,j,k}$ are $T_{i,j,i+j+1\, {\rm mod}\, 2}=\theta_{i,j,i+j+1\, {\rm mod}\, 2}=t_{i,j}$ by \eqref{inithet}. The theorem follows.
\end{proof}

\begin{example}
For $m=1$, all sequences are constant, $a_i=a$, $b_i=b$, $c_i=c$, $d_i=d$, and $x_i=\frac{2ab}{cd}=x$, $y_i=\frac{2cd}{ab}=y$. 
Moreover, we have
$u_{n,i}=x^{\lfloor \frac{(n+1)^2}{4}\rfloor}$, $v_{n,i}=y^{\lfloor \frac{(n+1)^2}{4}\rfloor}$, hence:
$$T_{i,j,k}=2^{\frac{k(k-1)}{2}}
\left(\frac{ab}{cd}\right)^{\lfloor \frac{k}{2}\rfloor} \left\{\delta^{[2]}_{k,0}\left(a\, \delta^{[4]}_{i+j+k,1}+b\, \delta^{[4]}_{i+j+k,3}\right) 
+\delta^{[2]}_{k,1}\left(c\, \delta^{[4]}_{i+j+k,1}+d\, \delta^{[4]}_{i+j+k,3}\right)
\right\}$$
This solution is slightly more general that the uniform one (which would correspond to $a=b=c=d=1$ and  $T_{i,j,k}=2^{\frac{k(k-1)}{2}}$), 
but we easily compute:
$$R_{i,j,k}=L_{i,j,k}=\frac{1}{2} \qquad (i,j\in \Z;k\in \Z_+;i+j+k=0\, {\rm mod}\, 2)\, ,$$
therefore the general equation \eqref{densityequation} for the density reduces to that of the uniform case \eqref{densitytrivialdata}.
\end{example}

\begin{example}
For $m=2$, we find that 
$$ x_0=\frac{ c_1 d_0+c_0d_1}{a_0b_0} , \quad x_1=\frac{ c_1 d_0+c_0d_1}{a_1b_1}, 
\quad y_0=\frac{a_0b_1+a_1b_0}{c_0d_0},\quad y_1=\frac{a_0b_1+a_1b_0}{c_1d_1} .$$
and the solution reads:
\begin{eqnarray*}
T_{i,j,k}&=&\delta_{k,0}^{[2]}\left(a_{\frac{i-j-1}{2}}\, \delta_{i+j+k,1}^{[4]}+b_{\frac{i-j-1}{2}}\, \delta_{i+j+k,3}^{[4]}\right) 
(x_0x_1y_0y_1)^{\frac{k(k-2)}{8}}\left(x_{\frac{i-j-1}{2}}\right)^{\lfloor \frac{k+2}{4}\rfloor}\left(x_{\frac{i-j+1}{2}}\right)^{\lfloor \frac{k}{4}\rfloor}\\
&+&\delta_{k,1}^{[2]}\left(c_{\frac{i-j}{2}}\, \delta_{i+j+k,1}^{[4]}+d_{\frac{i-j}{2}}\, \delta_{i+j+k,3}^{[4]}\right)
(x_0x_1y_0y_1)^{\frac{k^2-1}{8}}\left(y_{\frac{i-j}{2}}\right)^{-\lfloor \frac{k-1}{4}\rfloor}\left(y_{\frac{i-j}{2}+1}\right)^{-\lfloor \frac{k+1}{4}\rfloor}
\end{eqnarray*}
Again, this solution is more general than that of the $2\times 2$ case \eqref{exasol22} 
(which would correspond to $c_0=d_1=a$, $c_1=d_0=b$, $a_1=b_0=c$, and $a_0=b_1=d$), 
but as we shall see below (Example \ref{twotwoex}), it has the same values of $R_{i,j,k}$ and $L_{i,j,k}$ 
and therefore the same equation for the density.
\end{example}

\subsection{Density: exact derivation}

In this section we consider the $T$-system with $m$-toroidal boundary conditions.
We define the density $\rho\equiv \rho^{(0,0)}$ as before as the response of the system to an infinitesimal perturbation of the initial data at position $(0,0)$,
with value $T_{0,0,1}=t_{0,0}=c_0$. More precisely,
we write:
$$ \rho_{i,j,k}=c_0\frac{\partial {\rm Log}\, T_{i,j,k}}{\partial t_{0,0}}\Big\vert_{t_{0,0}=c_0} $$

Our aim in this section is to compute $\rho_{i,j,k}$ explicitly (As before, the singularity locus of the generating function for $\rho$ will determine the suitable arctic curve.). Note that we have the following initial conditions:
\begin{equation}\label{initrho}
\rho_{i,j,0}=0\qquad {\rm and}\qquad  \rho_{i,j,1}=\delta_{i,0}\delta_{j,0}\qquad (i,j\in \Z)
\end{equation}
Next, differentiating the $T$-system relation w.r.t. to $t_{0,0}$ provides us with the following 
system of linear recursion relations for $\rho_{i,j,k}$:
\begin{equation}\label{systemrho}
\rho_{i,j,k+1}+\rho_{i,j,k-1}=L_{i,j,k} (\rho_{i+1,j,k}+\rho_{i-1,j,k})+R_{i,j,k}(\rho_{i,j+1,k}+\rho_{i,j-1,k}) 
\end{equation}
which, together with the initial conditions \eqref{initrho}, determine $\rho_{i,j,k}$ entirely.
The crucial remark here is that although this system is infinite, it has only finitely many distinct coefficients.
Indeed, from Theorem \ref{solmper}, we deduce the following simple:
\begin{cor}
The quantities $(L_{i,j,k},R_{i,j,k})$ for the solutions of the $T$-system with $m$-toroidal boundary conditions have
the following periodicities:
\begin{eqnarray*} 
L_{i+2,j+2,k}&=& L_{i,j,k}\qquad L_{i+m,j-m,k}=L_{i,j,k}\qquad L_{i+1,j+1,k+2}=L_{i,j,k} \\
R_{i+2,j+2,k}&=& R_{i,j,k}\qquad R_{i+m,j-m,k}=R_{i,j,k}\qquad R_{i+1,j+1,k+2}=R_{i,j,k}
\end{eqnarray*}
\end{cor}
\begin{proof}
The first two relations are clear, as this periodicity is inherited from that of the initial data \eqref{init2m}.
The last one is checked directly on the expressions for $L_{i,j,k},R_{i,j,k}$ of Theorem \ref{solmper}: one simply notices
that the translation $(i,j,k)\to (i+1,j+1,k+2)$ leaves $\al$ and $\beta$ invariant, and leaves also the quantity $i+j+k\to i+j+k+4$
invariant modulo $4$, and $k\to k+2$ invariant modulo 2. 
\end{proof}
In other words, coefficients of the system  \eqref{systemrho} are periodic in the $\Z^3$ lattice, with period vectors: 
 $\vec{e}_1=(2,2,0)$,
$\vec{e}_2=(m,-m,0)$ and $\vec{e}_3=(1,1,2)$.
This suggests to introduce the following generating functions, for $i,j\in \Z$ and $k\in \Z_+$:
$$  \rho^{(i,j,k)}(x,y,z) = \sum_{a,b\in\Z,c\geq 0} \rho_{i+2a+mb+c,j+2a-mb+c,k+2c} x^{i+2a+mb+c}y^{j+2a-mb+c}z^{k+2c}$$
while the total density generating function $\rho(x,y,z)=\sum_{i,j\in\Z,k\in\Z_+} x^iy^jz^k \rho_{i,j,k}$ is equal to the sum
$$ \rho(x,y,z)=\sum_{(i,j,k)\in \Pi_m}\rho^{(i,j,k)}(x,y,z)$$
where $\Pi_m$ is the set of integral points within the paralellepipedon based on $\vec{e}_i$, $i=1,2,3$, namely: 
$$\Pi_m=\{\sum_{i=1}^3 t_i \vec{e}_i,\ 0\leq t_i <1 \}\cap \{(i,j,k)\in \Z^3, \ i+j+k=1\, {\rm mod}\, 2\}$$
Note that with the above definition, $\rho^{(i,j,k)}(x,y,z)$ is periodic in $i,j,k$, namely it satisfies:
$$\rho^{(i+2a+mb+c,j+2a-mb+c,k+2c)}(x,y,z)=\rho^{(i,j,k)}(x,y,z)$$
for $a,b\in \Z$ and $c\in \Z_+$, for all $i,j\in \Z$ and $k\in \Z_+$.
For later use, we also note that, extending the above definition to $k=-1$ leads to: 
\begin{eqnarray*}\rho^{(i,j,-1)}(x,y,z)&=&-z^{-1}\delta_{i,0}\delta_{j,0}+ 
\sum_{a,b\in\Z,c\geq 1} \rho_{i+2a+mb+c,j+2a-mb+c,k+2c} x^{i+2a+mb+c}y^{j+2a-mb+c}z^{k+2c}\\
&=& -z^{-1}\delta_{i,0}\delta_{j,0}+\rho^{(i+1,j+1,1)}(x,y,z)
\end{eqnarray*}
where we have used $\rho_{i,j,1}+\rho_{i,j,-1}=0$ and the initial condition.
We finally obtain the following system for the generating functions $\rho^{(i,j,k)}(x,y,z)$, $k=0,1$ by use of the periodicities:
\begin{eqnarray*}
&&{\rm For}\ \  i+j=0\, {\rm mod}\, 2: \\
&&\quad z^{-1}\rho^{(i,j,1)}(x,y,z)+z\rho^{(i+1,j+1,1)}(x,y,z)
=L_{i,j,0} (x^{-1}\rho^{(i+1,j,0)}(x,y,z)+x\rho^{(i-1,j,0)}(x,y,z))\\
&&\qquad\qquad\qquad\qquad\qquad\qquad\qquad +R_{i,j,0}(y^{-1}\rho^{(i,j+1,0)}(x,y,z)+y\rho^{(i,j-1,0)}(x,y,z) )+\delta_{i,0}\delta_{j,0}\\
&&{\rm For}\ \  i+j=1\, {\rm mod}\, 2: \\
&&\quad z^{-1}\rho^{(i-1,j-1,0)}(x,y,z)+z\rho^{(i,j,0)}(x,y,z)
=L_{i,j,1} (x^{-1}\rho^{(i+1,j,1)}(x,y,z)+x\rho^{(i-1,j,1)}(x,y,z))\\
&&\qquad\qquad\qquad\qquad\qquad\qquad\qquad\qquad\qquad +R_{i,j,1}(y^{-1}\rho^{(i,j+1,1)}(x,y,z)+y\rho^{(i,j-1,1)}(x,y,z) )
\end{eqnarray*}
We may actually further restrict this system to a fundamental domain
$P_m$ of the $(i,j,k)$, $k=0,1$ planes modulo $\vec{e}_1$ and $\vec{e}_2$, which we take to be:
$$ P_m=\Big\{ (i+1,-i,0),(i+2,-i+1,0),(i,-i,1),(i+1,-i+1,1)\Big\}_{ i\in  \{0,1,...,m-1\}}$$
Let us define the four following $m$-periodic functions for $i\in \Z$:
\begin{eqnarray*}
\alpha_i(x,y,z)&=& \rho^{(i+1,-i,0)}(x,y,z),\qquad \beta_i(x,y,z)=\rho^{(i+2,-i+1,0)}(x,y,z) \\
\gamma_i(x,y,z)&=& \rho^{(i,-i,1)}(x,y,z),\qquad \quad \delta_i(x,y,z)=\rho^{(i+1,-i+1,1)}(x,y,z)
\end{eqnarray*}
and the coefficients
\begin{equation}
\lambda_i=L_{i,-i,0}=\frac{a_ib_{i-1}}{a_{i-1}b_i+a_ib_{i-1}} \qquad \mu_i=L_{i+2,-i+1,1}=\frac{c_{i+1}d_i}{c_id_{i+1}+c_{i+1}d_i}
\end{equation}
so that $L_{i+1,-i+1,0}=1-\lambda_i$ and $L_{i+3,-i+2,1}=1-\mu_i$. Note also that, as is readily seen from their definition,
these coefficients are not independent, as they must satisfy the relations:
\begin{equation}\label{relamu}
\prod_{i=0}^{m-1} \left(\frac{1}{\lambda_i} -1\right)=\prod_{i=0}^{m-1} \left(\frac{1}{\mu_i} -1\right)=1
\end{equation}
We may summarize the above results into:

\begin{thm}\label{sysrho}
The density generating functions $\al_i,\beta_i,\gamma_i,\delta_i$, for $i\in \{0,1,...,m-1\}$, 
are uniquely determined as the solutions of the following $4m\times 4m$
linear system:
\begin{eqnarray}
z^{-1}\, \al_i+z\, \beta_i-\lambda_i\, (x^{-1}\,\gamma_{i+1}+x\,\delta_i )-(1-\lambda_i)(y^{-1}\,\gamma_{i}+y\,\delta_{i+1} )&=&0 
\nonumber \\
z^{-1}\, \beta_i+z \, \al_i-(1-\lambda_i) (x^{-1}\,\delta_{i+1}+x\,\gamma_i )-\lambda_i\,(y^{-1}\,\delta_{i}+y\,\gamma_{i+1} )&=&0 
\nonumber \\
z^{-1}\, \gamma_i+z \, \delta_i-\mu_i\, (x^{-1} \,\al_i +x\, \beta_{i-1})-(1-\mu_i)(y^{-1}\,\al_{i-1}+y\,\beta_{i} )&=&\delta_{i,0} 
\nonumber \\
z^{-1}\, \delta_i+z \, \gamma_i-(1-\mu_i) (x^{-1} \,\beta_i +x\,\al_{i-1})-\mu_i\,(y^{-1}\,\beta_{i-1}+y\,\al_{i} )&=&0
\label{sysgenro}
\end{eqnarray}
for $i,j\in \{0,1,...,m-1\}$, 
subject to the periodicity conditions $\al_m=\al_0$, $\al_{-1}=\al_{m-1}$, and similarly for the $\beta$, $\gamma$, $\delta$'s.
\end{thm}

The solution of the general system for the $m$-periodic densities of Theorem \ref{sysrho} is 
always a rational fraction of $x,y,z$, with denominator given by the determinant of the system \eqref{sysgenro}.
The matrix of coefficients may be rewritten in block form as:
$$ M=\begin{pmatrix}
 z^{-1}\, I & z I & -M(x,y) & -{\bar M}(y^{-1},x^{-1}) \\
 z \, I &  z^{-1}I & -M(y^{-1},x^{-1}) & -{\bar M}(x,y)\\
 -P(x,y) & -{\bar P}(y^{-1},x^{-1}) &  z^{-1}I & z \, I \\
 -{\bar P}(x,y) & -P(y^{-1},x^{-1}) &  z \, I &  z^{-1}I 
\end{pmatrix}$$
where all the entries are $m\times m$ matrices, with:
$$P(x,y)=\begin{pmatrix} 
\frac{\mu_0}{x} & 0 & 0 &\cdots & 0 & \frac{1-\mu_0}{y}\\
\frac{1-\mu_1}{y} &\frac{\mu_1}{x} & 0 &  & & 0 \\
0 & \frac{1-\mu_2}{y} &\frac{\mu_2}{x} &\ddots & & 0\\
\vdots & \ddots & \ddots & \ddots &\ddots &\vdots  \\
\vdots & & \ddots & \ddots &\ddots  &0  \\
0 & \cdots &\cdots & 0& \frac{1-\mu_{m-1}}{y} & \frac{\mu_{m-1}}{x}
\end{pmatrix} $$
$$
M(x,y)=\begin{pmatrix} \frac{1-\lambda_0}{y} & \frac{\lambda_0}{x} & 0 &\cdots & 0 & 0\\
0 &\frac{1-\lambda_1}{y} &  \frac{\lambda_1}{x}  &\ddots  & & 0 \\
0 & 0 &\frac{1-\lambda_2}{y} &\frac{\lambda_2}{x}  &\ddots & \vdots\\
\vdots & \ddots & \ddots & \ddots &\ddots &0 \\
0 & & \ddots & \ddots &\ddots  &\frac{\lambda_{m-2}}{x}  \\
 \frac{\lambda_{m-1}}{x}  & \cdots & \cdots & 0& 0 & \frac{1-\lambda_{m-1}}{x}
\end{pmatrix} $$
and ${\bar P}(x,y)$ is $P(x,y)$ with $\mu_i$ and $1-\mu_i$ interchanged, while ${\bar M}(x,y)$ is $M(x,y)$ with 
$\lambda_i$ and $1-\lambda_i$ interchanged.

\begin{example}
For $m=1$, the relation \eqref{relamu} gives $\lambda_0=\mu_0=\frac{1}{2}$ as expected. Moreover, 
denoting by $\alpha\equiv \al_0$, etc., the linear system of Theorem \ref{sysrho} reduces to:
\begin{eqnarray*}
z^{-1}\, \al+z\, \beta-\frac{1}{2} (x^{-1}\,\gamma+x\,\delta)-\frac{1}{2}(y^{-1}\,\gamma+y\,\delta )&=&0 
\\
z^{-1}\, \beta+z \, \al-\frac{1}{2}(x^{-1}\,\delta+x\,\gamma)-\frac{1}{2}(y^{-1}\,\delta+y\,\gamma )&=&0 
\\
z^{-1}\, \gamma+z \, \delta-\frac{1}{2}(x^{-1} \,\al+x\, \beta)-\frac{1}{2}(y^{-1}\,\al+y\,\beta )&=&1
\\
z^{-1}\, \delta+z \, \gamma-\frac{1}{2}(x^{-1} \,\beta +x\,\al)-\frac{1}{2}(y^{-1}\,\beta+y\,\al)&=&0
\end{eqnarray*}
The total density $\rho=\al+\beta+\gamma+\delta$ therefore satisfies:
$$\Big(z^{-1}+z-\frac{1}{2}(x+x^{-1}+y+y^{-1})\Big) \rho(x,y,z)=1 $$
which matches \eqref{genfuntrivialdata}.
\end{example}

\begin{figure}
\centering
\includegraphics[width=12.cm]{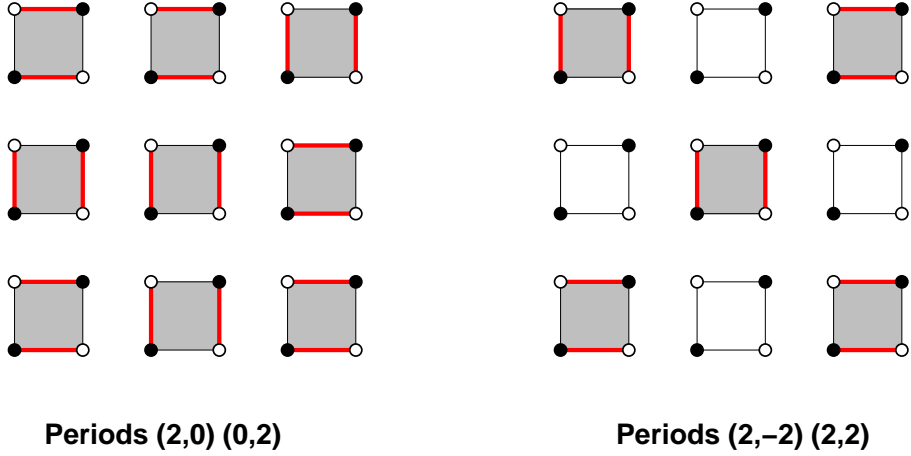}
\caption{\small Comparison between the facet phases of the $2\times2$ (left) and $2$-toroidal (right) $T$-system. In the latter
the pinning of configurations for small $c_0$ is on a larger square sublattice (of $c_0$-type faces). But as a result, the sublattice 
of $d_1$-type faces decouples entirely, and these also have two arbitrary parallel dimer configurations, which makes the phases
identical configurationwise.}
\label{fig:compfacet}
\end{figure}

\begin{example}\label{twotwoex}
For $m=2$, the relation  \eqref{relamu} gives $\lambda_1=1-\lambda_0$ and $\mu_1=1-\mu_0$. 
The $8\times 8$ linear system of Theorem \ref{sysrho} reads:
\begin{eqnarray*}
z^{-1}\, \al_0+z\, \beta_0-\lambda_0\, (x^{-1}\,\gamma_{1}+x\,\delta_0 )-\lambda_1\, (y^{-1}\,\gamma_{0}+y\,\delta_{1} )&=&0 \\
z^{-1}\, \beta_0+z \, \al_0-\lambda_1\, (x^{-1}\,\delta_{1}+x\,\gamma_0 )-\lambda_0\, (y^{-1}\,\delta_{0}+y\,\gamma_{1} )&=&0 \\
z^{-1}\, \al_1+z\, \beta_1-\lambda_1\, (x^{-1}\,\gamma_{0}+x\,\delta_1 )-\lambda_0\, (y^{-1}\,\gamma_{1}+y\,\delta_{0} )&=&0 \\
z^{-1}\, \beta_1+z \, \al_1-\lambda_0\, (x^{-1} \,\delta_0+x\,\gamma_1 )-\lambda_1 \, (y^{-1} \,\delta_1+y\, \gamma_0 )&=&0 \\
z^{-1}\, \gamma_0+z \, \delta_0-\mu_0\, (x^{-1} \,\al_0 +x\, \beta_{1} )-\mu_1\, (y^{-1}\,\al_{1}+y\,\beta_{0} )&=&1 \\
z^{-1}\, \delta_0+z \, \gamma_0-\mu_1\, (x^{-1} \,\beta_0 +x\,\al_{1} )-\mu_0\, (y^{-1}\,\beta_{1}+y\,\al_{0} )&=&0 \\
z^{-1}\, \gamma_1+z \, \delta_1-\mu_1\, (x^{-1} \,\al_1+x\, \beta_{0} )-\mu_0\, (y^{-1}\,\al_{0}+y\,\beta_{1} )&=&0 \\
z^{-1}\, \delta_1+z \, \gamma_1-\mu_0\,  (x^{-1} \,\beta_1 +x\,\al_{0} )-\mu_1\, (y^{-1}\,\beta_{0}+y\,\al_{1} )&=&0
\end{eqnarray*}
This boils down to the following $4\times 4$ system for a new set of generating functions: 
$\al=\al_0+\beta_1$, $\beta=\beta_0+\al_1$, $\gamma=\gamma_0+\delta_1$ and $\delta=\delta_0+\gamma_1$:
\begin{eqnarray*}
z^{-1}\, \al+z\, \beta-\lambda_0\, (x^{-1}+x)\delta-\lambda_1\, (y^{-1}+y)\gamma&=&0 \\
z^{-1}\, \beta+z \, \al-\lambda_1\, (x^{-1}+x)\gamma-\lambda_0\, (y^{-1}+y)\delta&=&0 \\
z^{-1}\, \gamma+z \, \delta-\mu_0\, (x^{-1}+x)\al-\mu_1\, (y^{-1}+y)\beta&=&1 \\
z^{-1}\, \delta+z \, \gamma-\mu_1\, (x^{-1}+x)\beta-\mu_0\, (y^{-1}+y)\al&=&0
\end{eqnarray*}
with $\lambda_0+\lambda_1=1=\mu_0+\mu_1$.
Note that this system is equivalent to that of \eqref{2x2system}, with $\lambda_1=\sigma$ and $\mu_1=\tau$. 
The simple reason for this is that the condition \eqref{relamu} has 
induced a more restrictive periodicity condition on the coefficients $\lambda_i,\mu_i$, namely $\lambda_1=1-\lambda_0$ and $\mu_1=1-\mu_0$,
which is equivalent to that of Section \ref{sectwotwo}. Another indication is found by comparing the facet phases in both models.
We have represented in Fig.\ref{fig:compfacet} the two facet phases. The pinning of dimer configurations on $c_0$ type faces
for $c_0$ small induces a pinning on $d_1$ faces as well, which makes the two phases identical.
\end{example}

\subsection{Arctic curves and phase diagram/limit shape}

\begin{figure}
        \centering
             \hbox{ 
                \includegraphics[width=3.85cm]{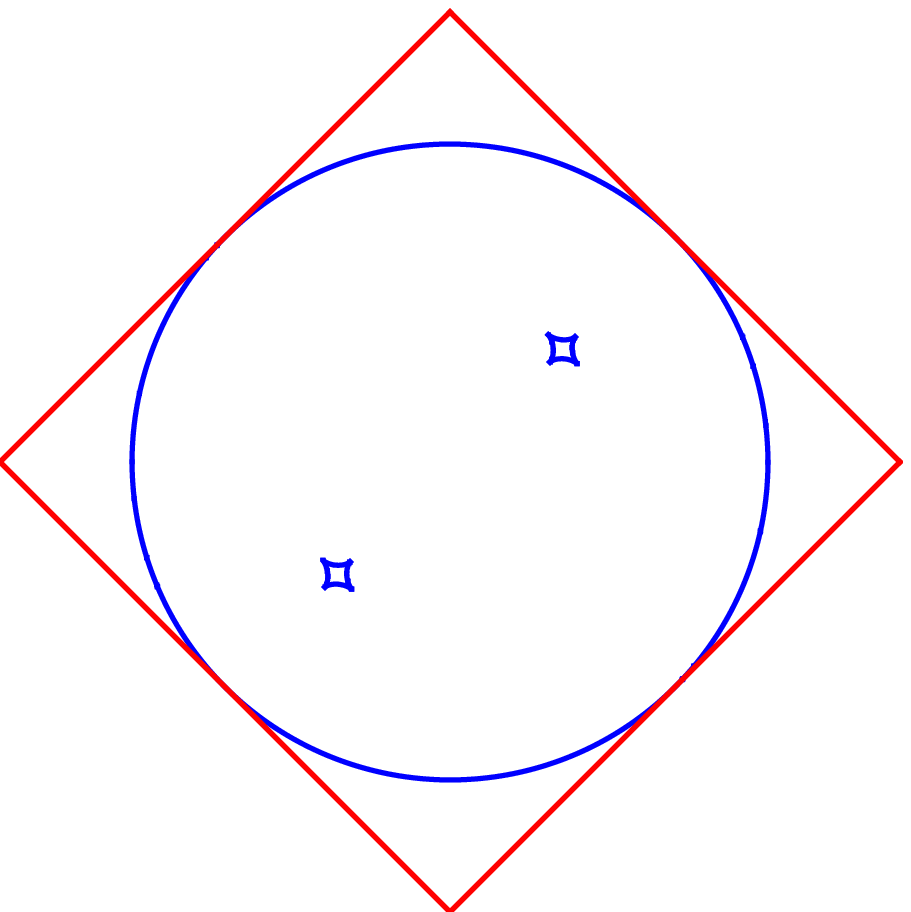}
                \put(-75,-15){$\lambda_1=4/9$}
               \hskip .2cm
             \includegraphics[width=3.85cm]{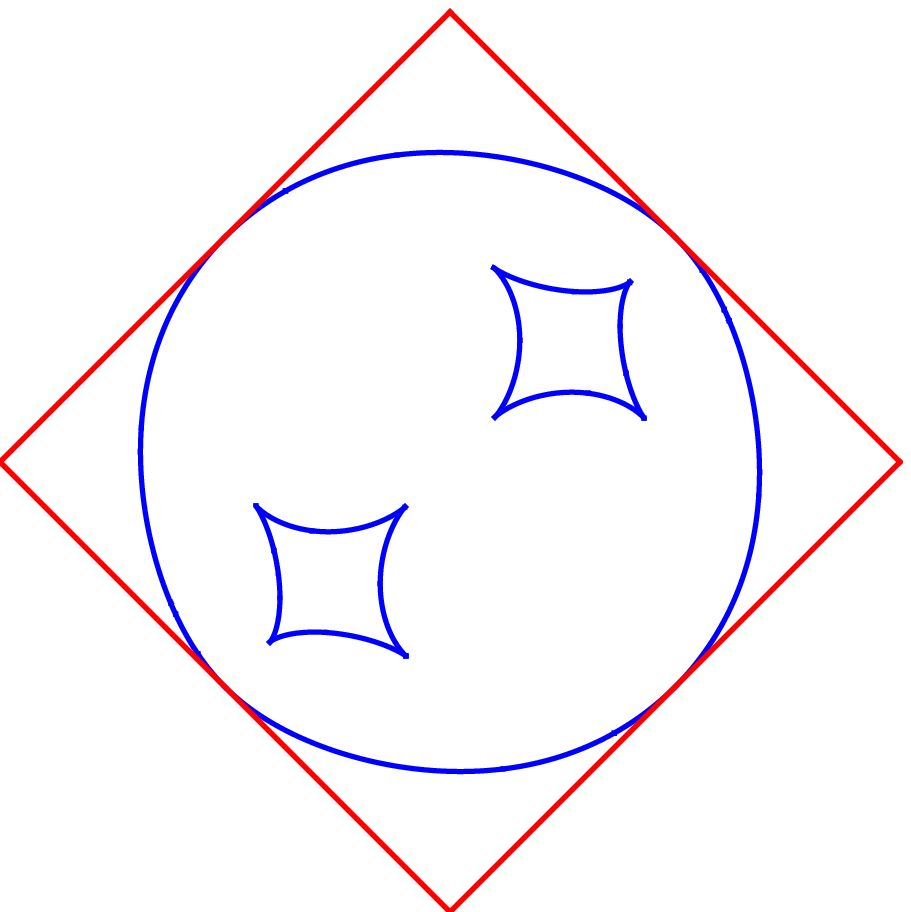}
             \put(-75,-15){$\lambda_1=1/5$}
             \hskip .2cm
             \includegraphics[width=3.85cm]{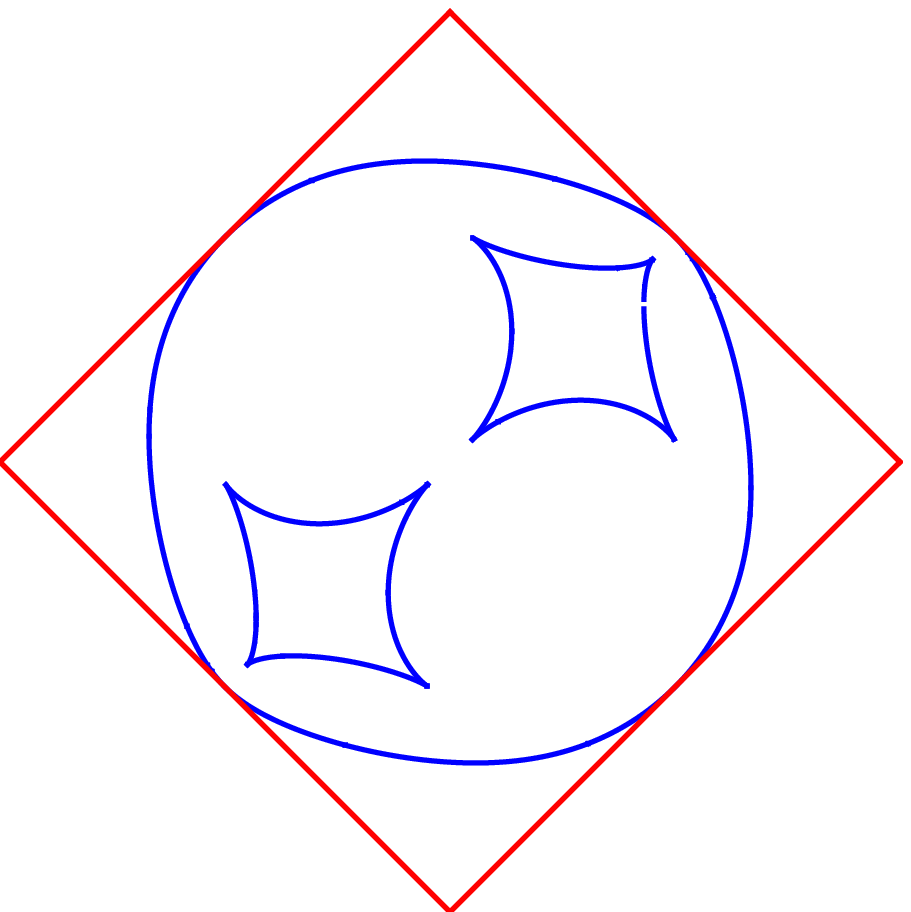}
             \put(-75,-15){$\lambda_1=19/10$}
             \hskip .2cm
             \includegraphics[width=3.85cm]{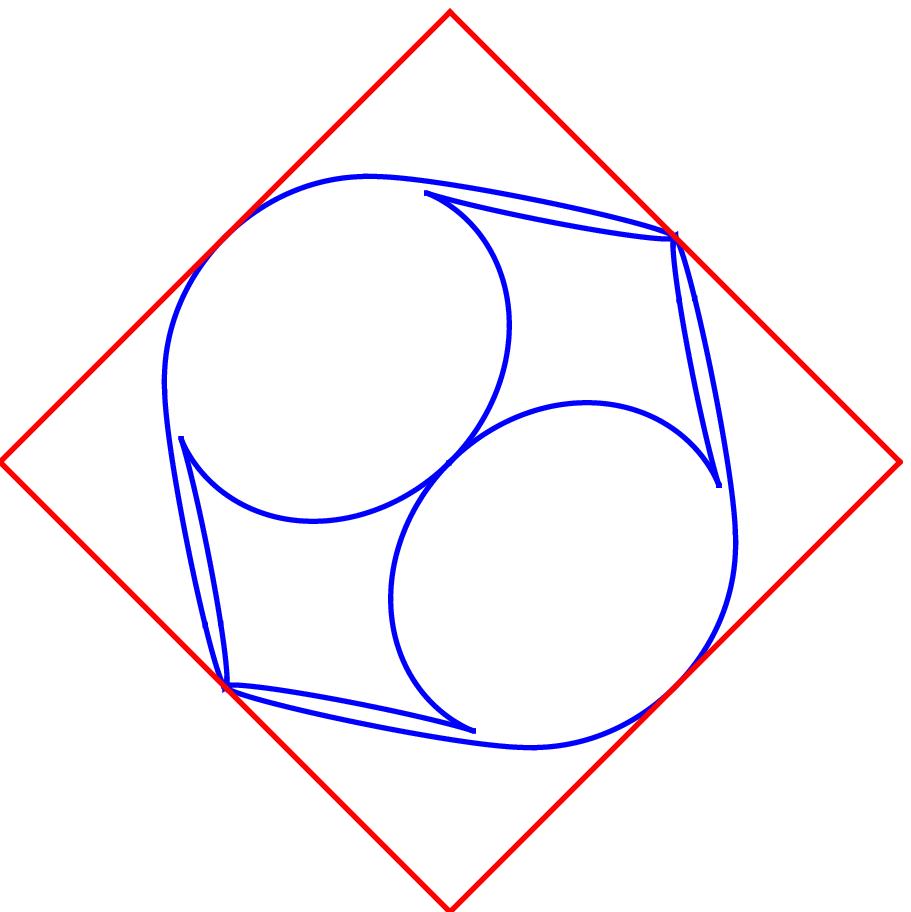}
             \put(-75,-15){$\lambda_1=200/201$}}
        \caption{\small Arctic curves for the $3$-toroidal initial data corresponding to different values of $\lambda_1$,
        where $\lambda_0=1/2$, $\lambda_2=1-\lambda_1$ and $\mu_0=\mu_1=\mu_2=1/2$.}
        \label{periodicarcticcurve2x3-1}
\end{figure}

\begin{figure}
        \centering

             \hbox{ 
                \includegraphics[width=3.85cm]{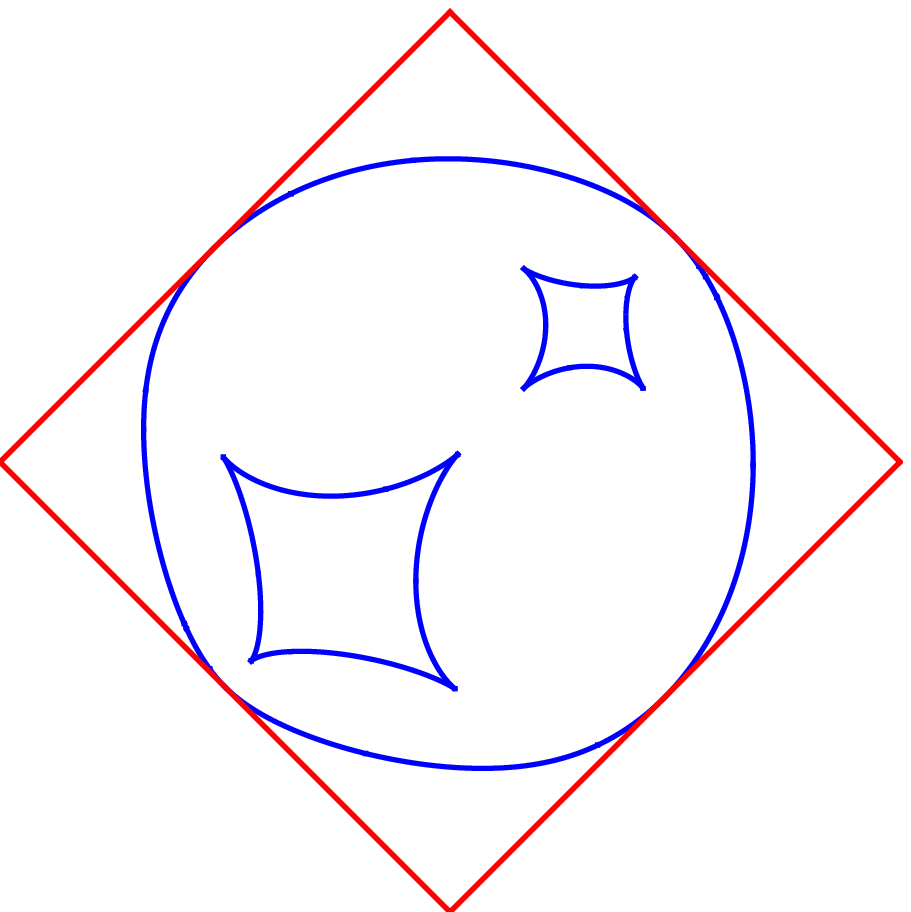}
                \put(-75,-15){$\mu_1=1/5$}
               \hskip .2cm
             \includegraphics[width=3.85cm]{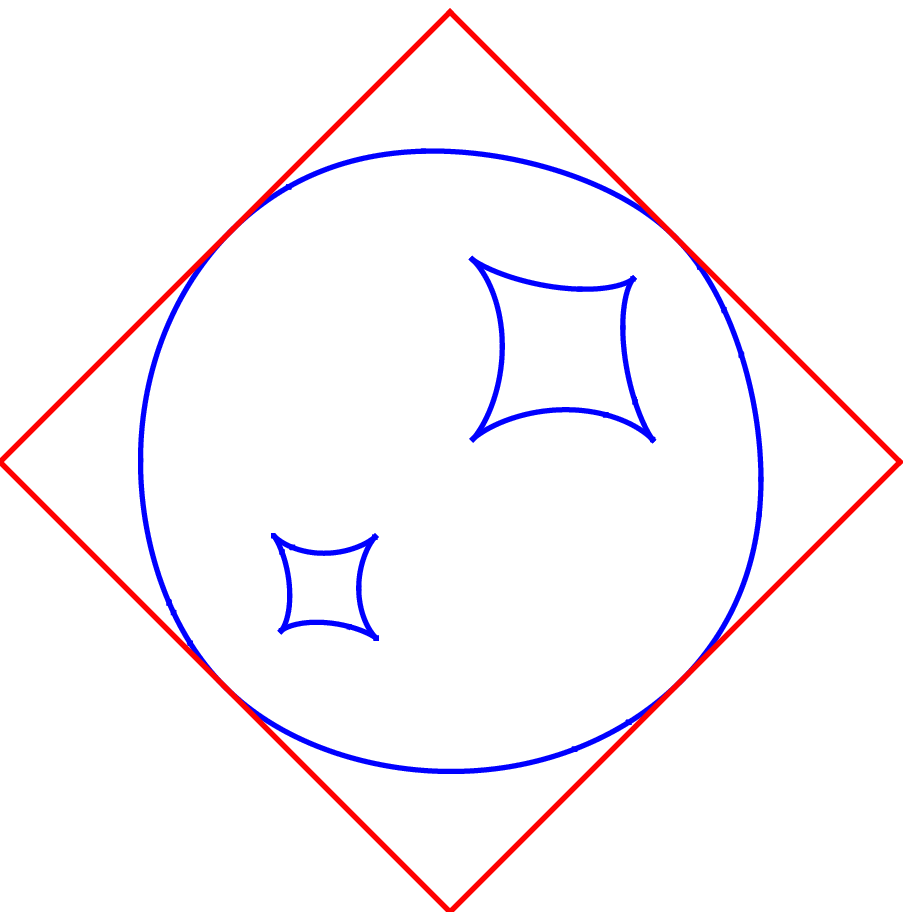}
             \put(-75,-15){$\mu_1=2/3$}
             \hskip .2cm
             \includegraphics[width=3.85cm]{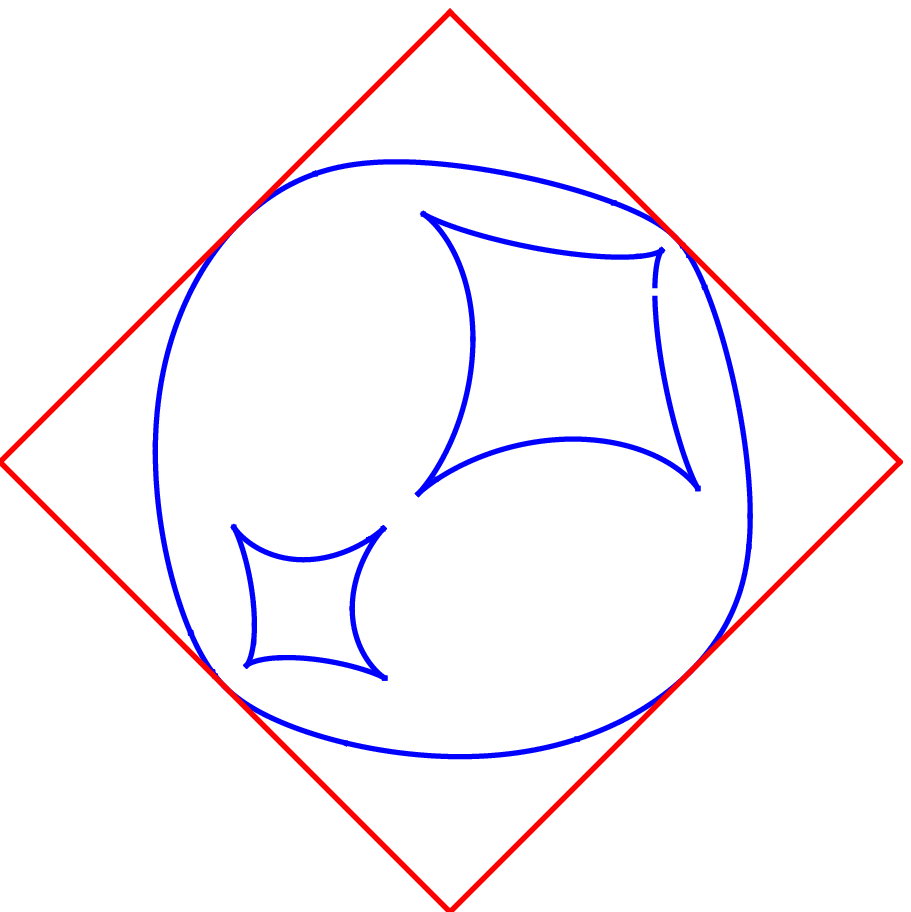}
             \put(-75,-15){$\mu_1=9/10$}
             \hskip .2cm
             \includegraphics[width=3.85cm]{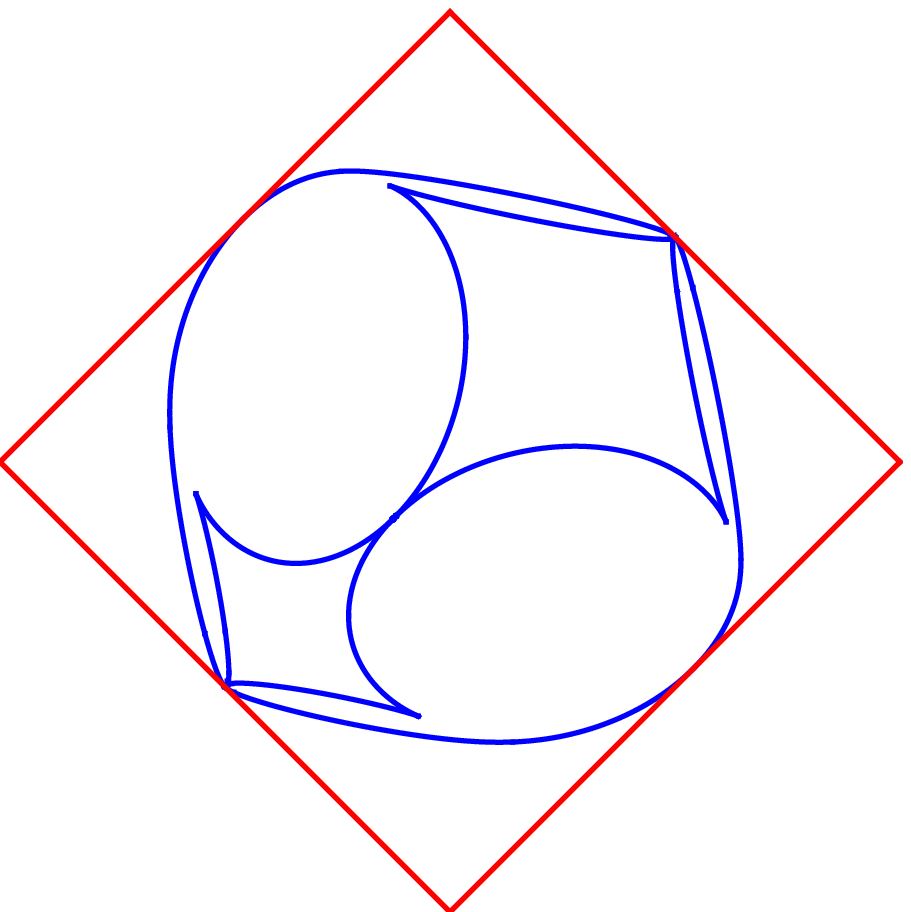}
             \put(-75,-15){$\mu_1=99/100$}}
        \caption{\small Arctic curves for the $3$-toroidal initial data corresponding to different values of $\mu_1$.
        Where $\lambda_0=1/2$, $\lambda_1=1/4$, $\lambda_2=1-\lambda_1=3/4$, $\mu_0=1/2$ and $\mu_2=1-\mu_1$.}
        \label{periodicarcticcurve2x3-2}
\end{figure}

It is clear that the solution of the general system for the $m$-periodic densities of Theorem \ref{sysrho} is 
always a rational fraction of $x,y,z$, with denominator given by the determinant of the system \eqref{sysgenro}. 
Applying the usual analysis to the singularity locus of this denominator
yields some algebraic arctic curve of higher degree. 
The details being cumbersome, we display here various plots
of these curves for $m=3,4$.

\begin{figure}
\centering
\includegraphics[width=14.cm]{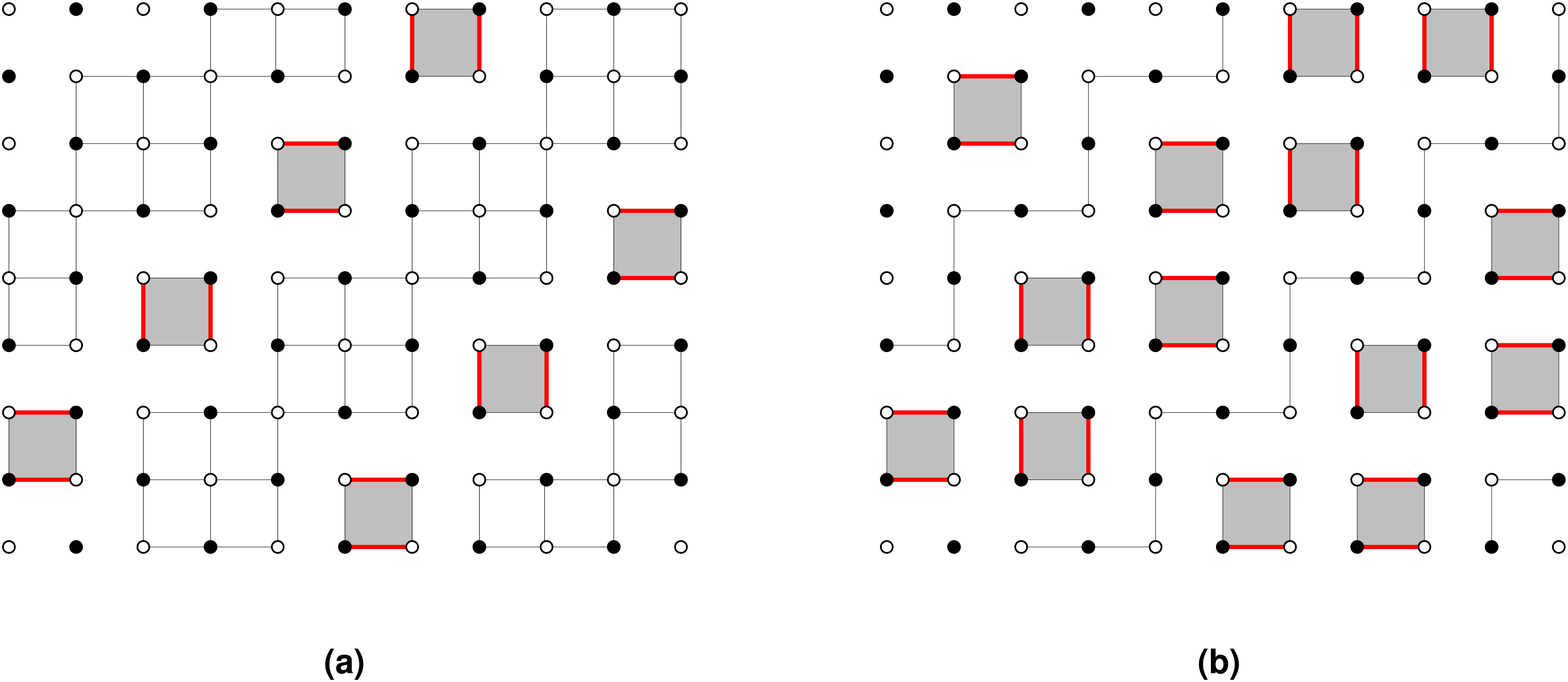}
\caption{\small Typical dominant states for the $3$-toroidal partition function, in the cases where: (a) $c_0\to 0$ while all other face weights remain finite (we have shaded the $c_0$-type faces); and (b) $c_0\sim d_1\to 0$ while all other face weights remain finite (both types of faces are shaded).}
\label{fig:ground23}
\end{figure}

\subsubsection{Case $m=3$}

\begin{figure}
        \centering
             \hbox{ 
                \includegraphics[width=3.85 cm]{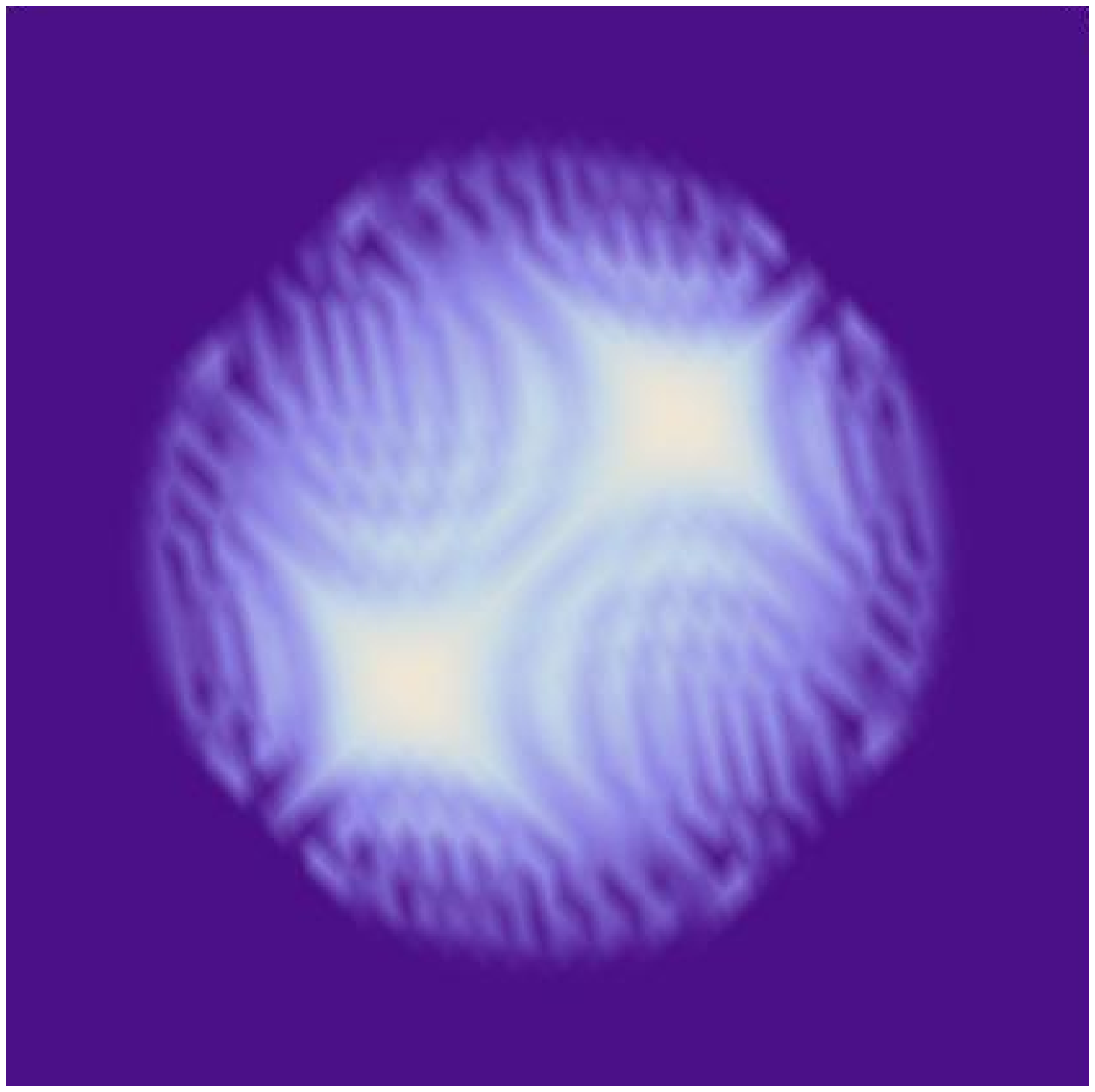}
               \put(-100,-15){$\mu_0=5/9 (c_0=4/5)$}
              \hskip .2cm
               \includegraphics[width=3.85 cm]{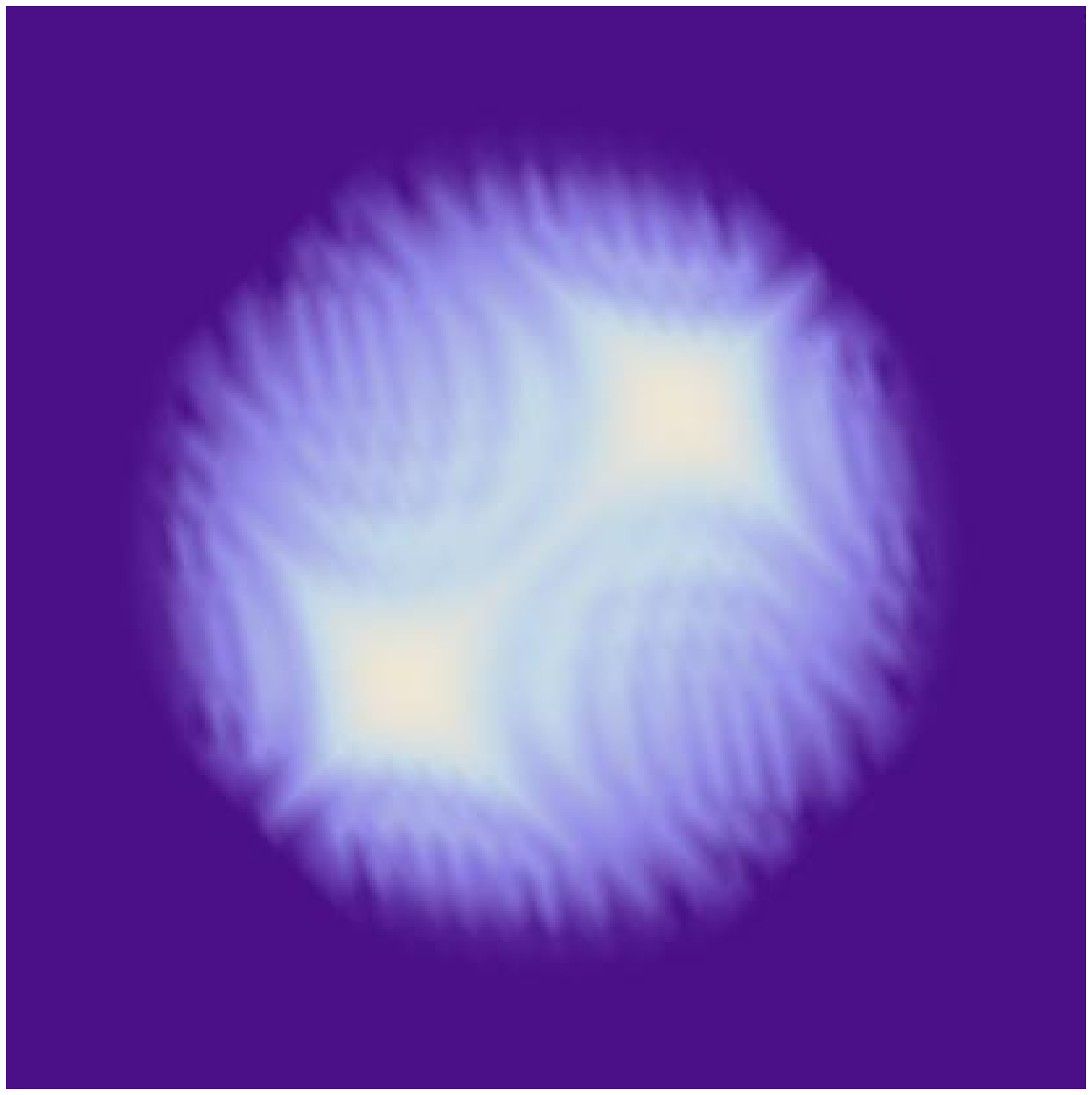}
               \put(-100,-15){$\mu_0=3/5 (c_0=2/3)$}
	      \hskip .2cm
              \includegraphics[width=3.85 cm]{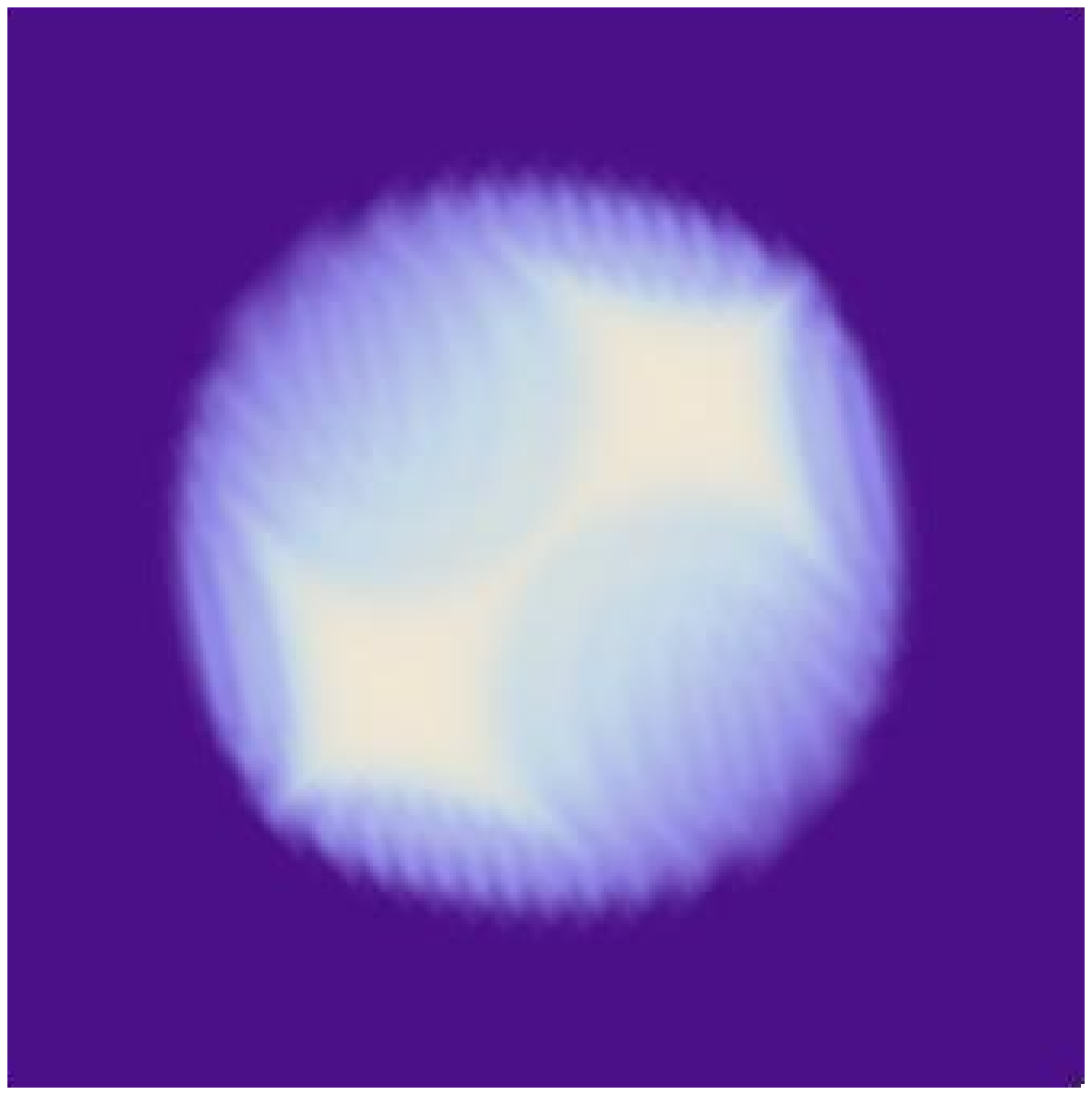}
              \put(-100,-15){$\mu_0=4/5 (c_0=1/4)$}
              \hskip .2cm
              \includegraphics[width=3.85 cm]{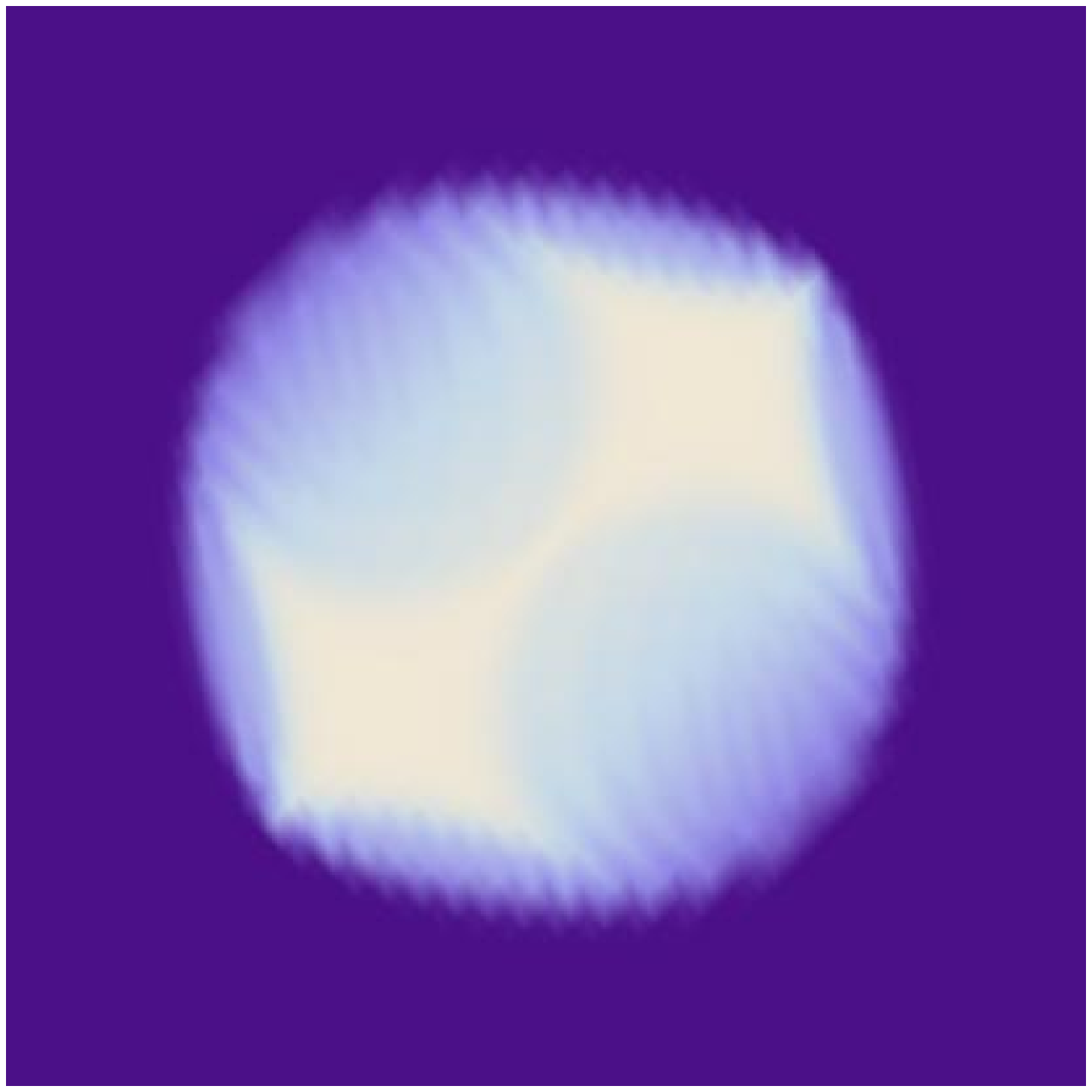}
              \put(-102,-15){$\mu_0=9/10 (c_0=1/9)$}}
        \caption{\small Density profiles for $\rho_{i,j,k}$ for the $3$-toroidal initial data corresponding to different values of $\mu_0$.
        Where $\mu_1=1/2$, $\mu_2=1-\mu_0$ and $\lambda_0=\lambda_1=\lambda_2=1/2$. (We have fixed $a_0=a_1=a_2=b_0=b_1=b_2=c_1=c_2=d_0=d_1=d_2=1$
        and only vary $c_0$.)}
        \label{densityprofiles2x3}
\end{figure}

For $m=3$, the arctic curve is found generically to be the zero locus of a polynomial of degree $14$ in $u,v$. In the physical range of parameters
$\lambda_i,\mu_i\in [0,1]$, we find generically 3 disconnected pieces, a first curve tangent to the square $|u|+|v|=1$ in four points, and two
inner pieces, each with 4 cusps, thus defining three inner regions in addition to the 4 frozen corners.
As before, we expect the two innermost regions to correspond to facet type phases, 
where the configurations get pinned to the faces with the smallest weights. As before, the density tends to 0 exponentially in the corners, and as a power of $k$ in the disordered region.

Assuming say that $c_0\to 0$ while the other faces weights remain finite,
we expect the  crystalline state depicted in Fig.\ref{fig:ground23} (a) to be dominant. If we take $c_0\sim d_1\to 0$, the crystalline state of 
Fig.\ref{fig:ground23} (b) tends to dominate the partition function. In all cases, we have a pinning of the configurations on the shaded 
faces with smallest weights,
all tending to be occupied by two parallel dimers, with 2 possibilities on each shaded face.

These however are not compatible with the boundary conditions of the Aztec graph, hence the formation of facets. Note that no facet occupies the 
center of the Aztec graph, except in the limiting case where one of the weights vanishes
(see Fig.\ref{periodicarcticcurve2x3-2} right, with $\mu_1=\frac{99}{100}$, $\mu_2=\frac{1}{100}$), in which case the two 
facets meet at a quadruple point in the center, which in the limit is the intersection of two tangent ellipses.

We have displayed in Fig. \ref{periodicarcticcurve2x3-1} the arctic curves occurring when the $\mu$ parameters are all trivial ($=1/2$)
while the $\lambda$'s vary. The curve is symmetric, and the two facets have identical size. In Fig.\ref{periodicarcticcurve2x3-2} however,
we have let both $\lambda$'s and $\mu$'s vary, and we see that the relative sizes of the facets vary as well. The qualitative
explanation for this is that one of the facet phases is more compatible with certain corners than others.

Like in the $2\times 2$ periodic  initial data case of Sect.\ref{22periosec}, the behavior of $\rho_{i,j,k}$ for large $k$  is found by comparing how the numerator and denominator of the expression for $\rho^{(0,0)}(x,y,z)$ behave in the vicinity of the singular point $x=y=z=1$. We find that
$\rho^{(0,0)}(1-t x,1-t y,1-t z)$ behaves as $t^{-3}$ when $\lambda_0\neq \lambda_2$, and as $t^{-2}$ otherwise, which means that if
$\lambda_0\neq \lambda_2$ $\rho_{i,j,k}$ tends to $0$ algebraically as $k^{-1}$, whereas it goes to a scaling function without overall scaling
otherwise. We have represented the profiles of the density function $k \rho^{(0,0)}_{i,j,k}$ for size $k=77$ in Fig.\ref{densityprofiles2x3}.

\subsubsection{Case $m=4$}

\begin{figure}
        \centering
             \hbox{ 
                \includegraphics[width=3.75cm]{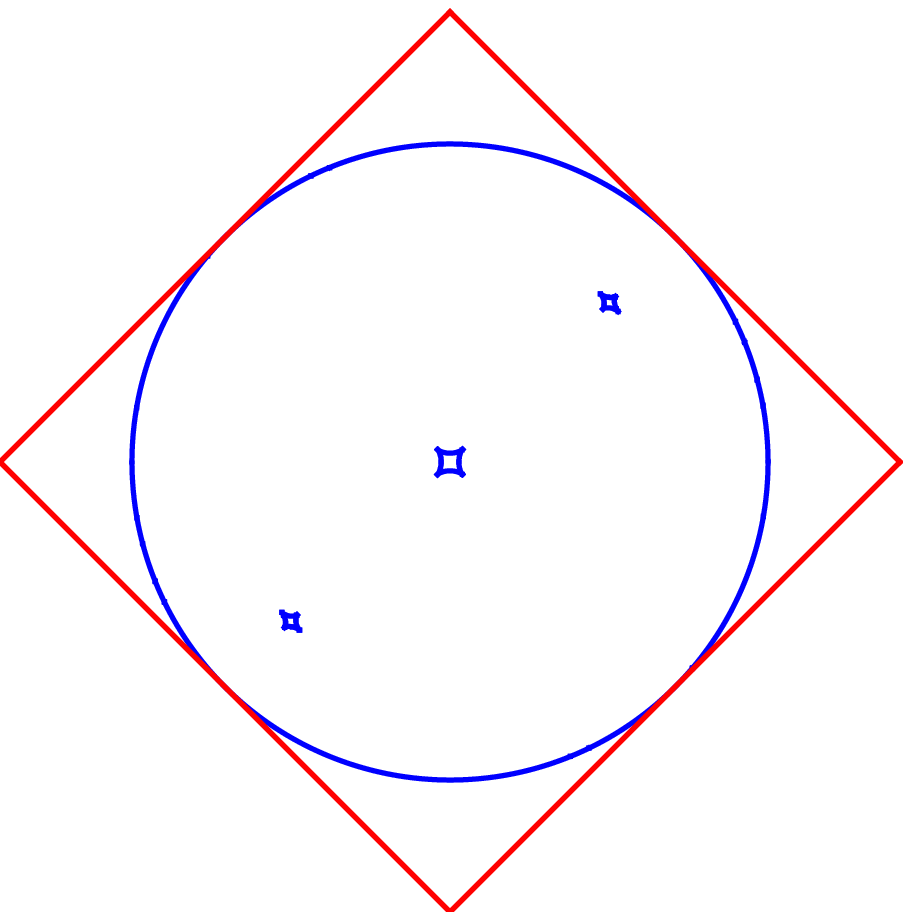}
                \put(-75,-15){$\lambda_2=4/9$}
                \hskip .2cm
             \includegraphics[width=3.75cm]{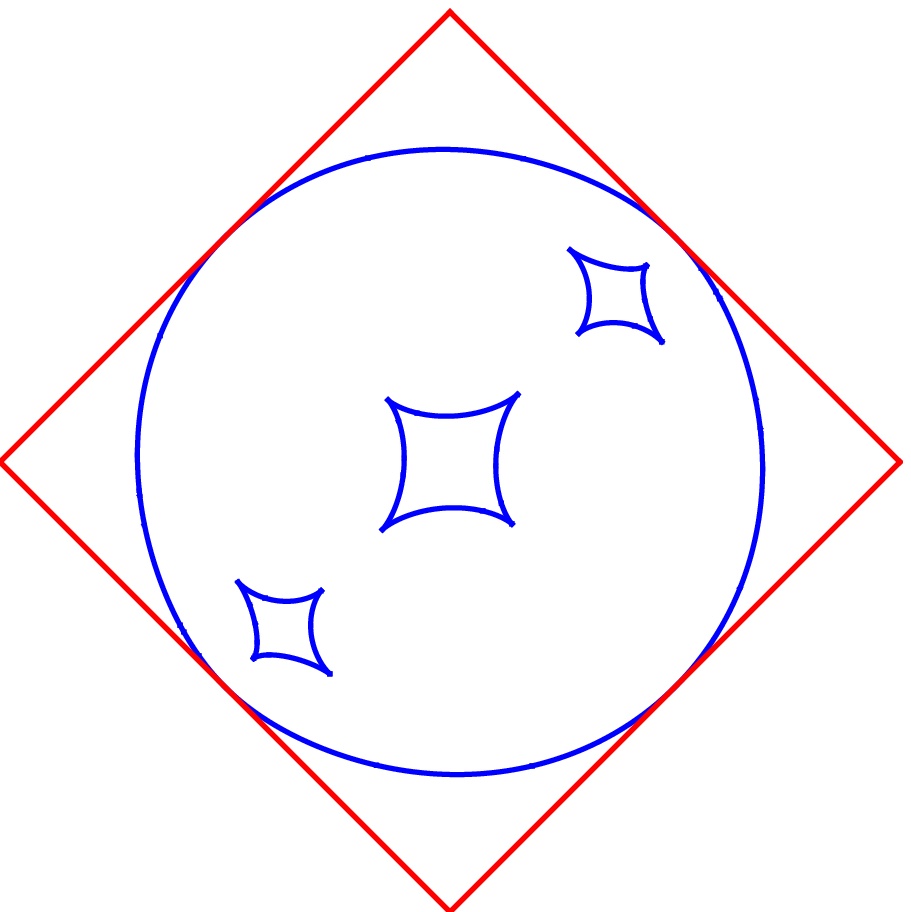}
             \put(-75,-15){$\lambda_2=1/5$}
               \hskip .2cm
             \includegraphics[width=3.75cm]{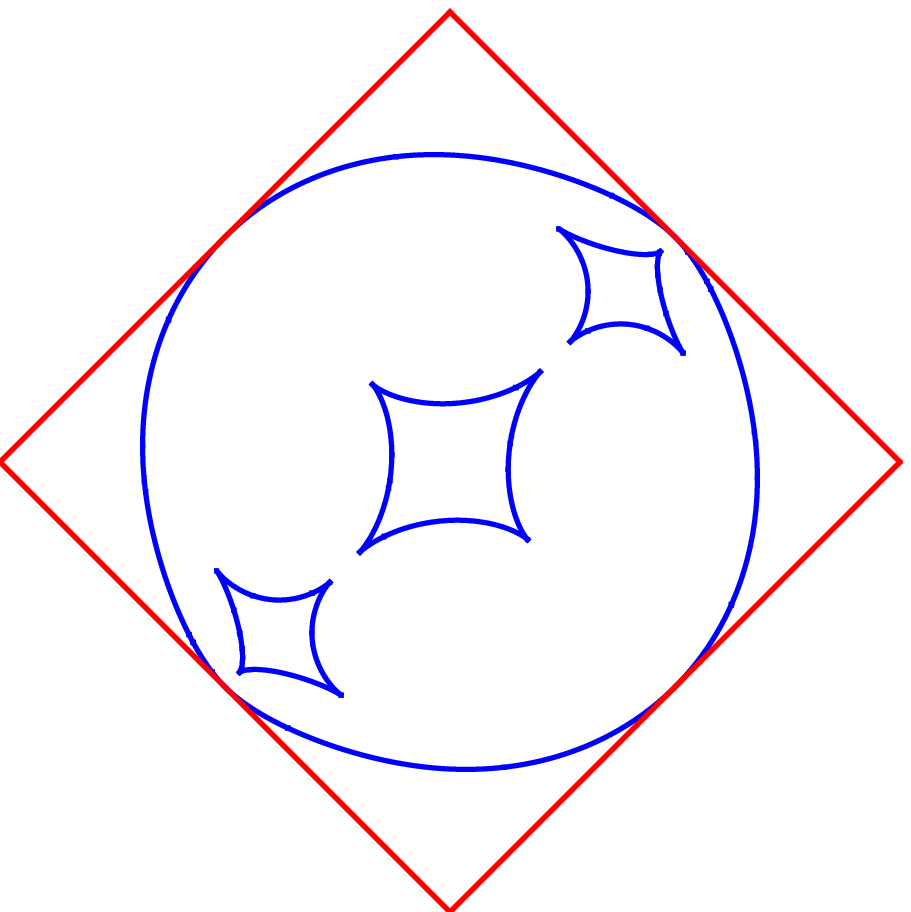}
             \put(-75,-15){$\lambda_2=9/10$}
             \hskip .2cm
             \includegraphics[width=3.75cm]{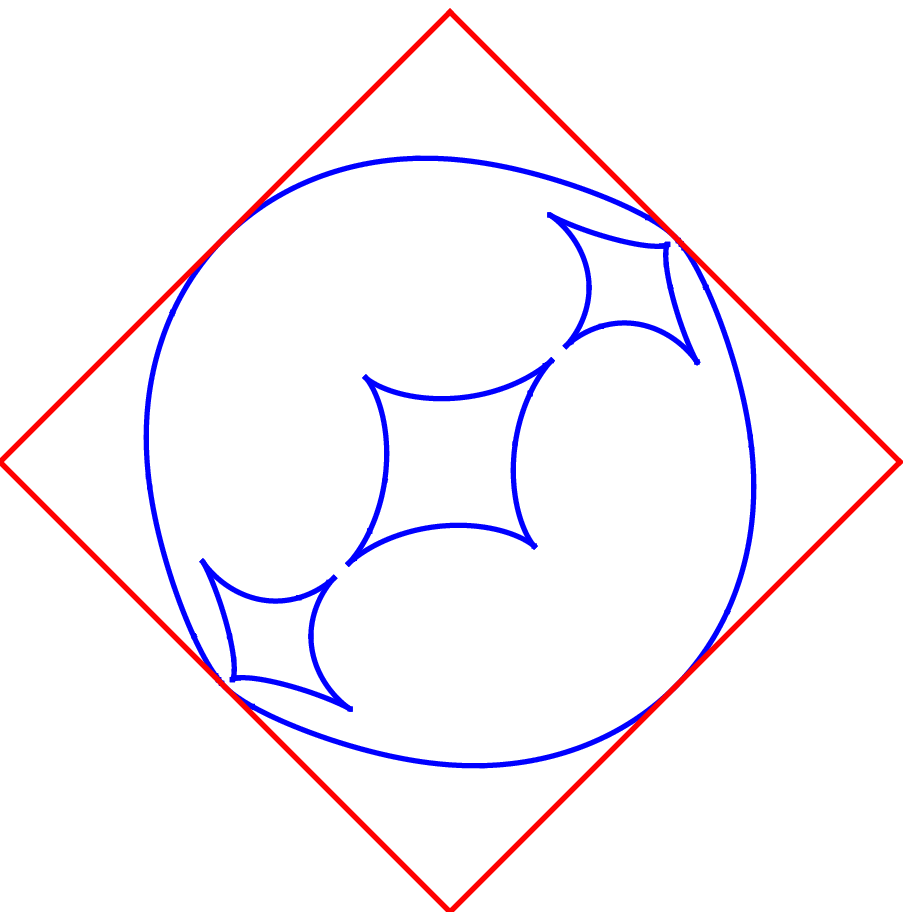}
             \put(-75,-15){$\lambda_2=19/20$}}

        \caption{\small Arctic curves for the $4$-toroidal initial data corresponding to different values of $\lambda_2$.
        Where $\lambda_0=1/2$, $\lambda_1=1/2$, $\lambda_3=1-\lambda_2$ and $\mu_0=\mu_1=\mu_2=\mu_3=1/2$.}
        \label{periodicarcticcurve2x4-1}
\end{figure}

\begin{figure}
        \centering

             \hbox{ 
                \includegraphics[width=3.75cm]{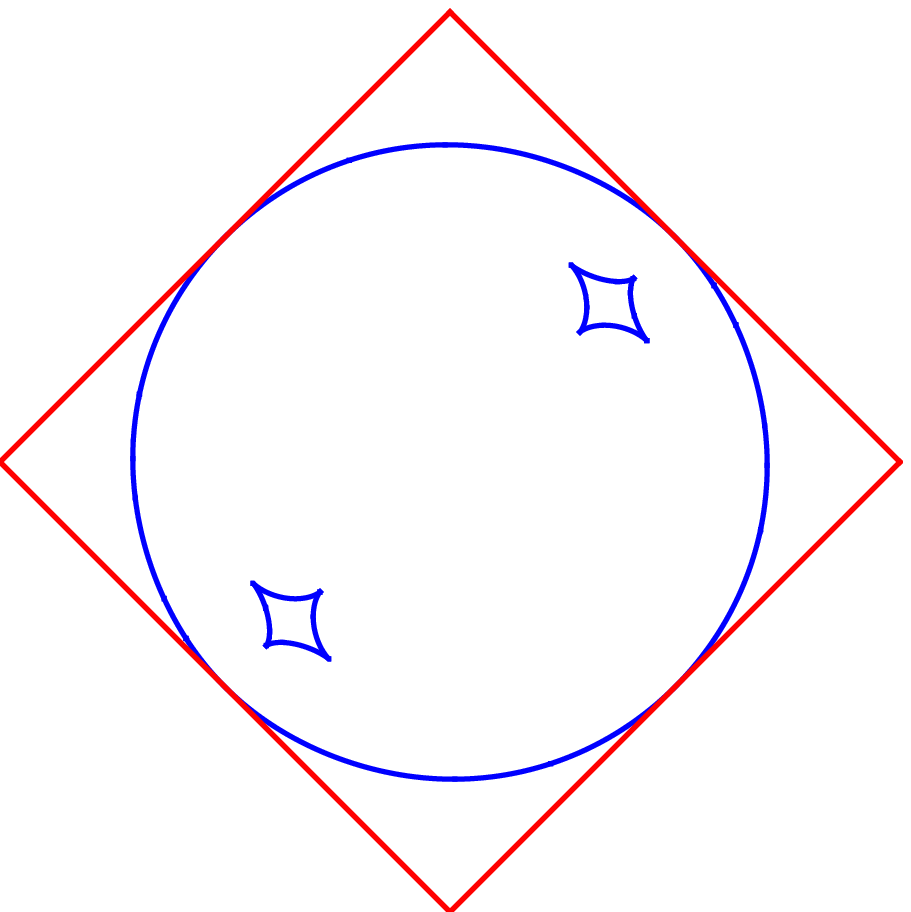}
               \put(-75,-15){$\lambda_2=1/2$}
               \hskip .2cm
                \includegraphics[width=3.75cm]{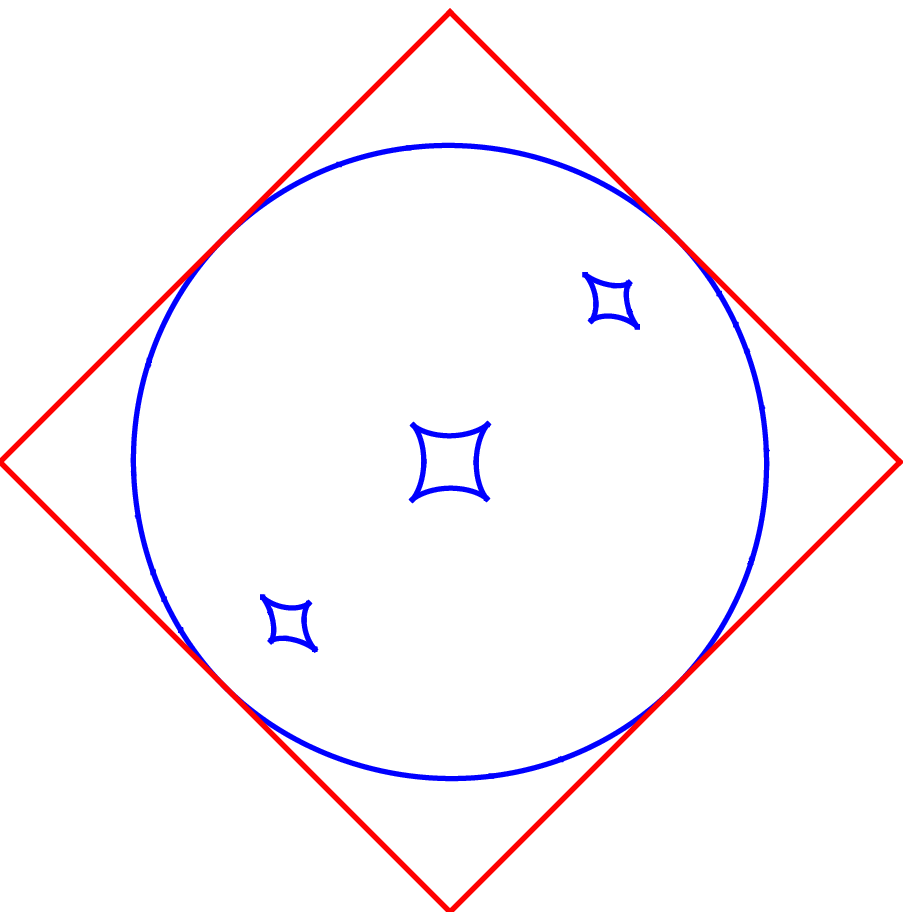}
                \put(-75,-15){$\lambda_2=1/3$}
               \hskip .2cm
             \includegraphics[width=3.75cm]{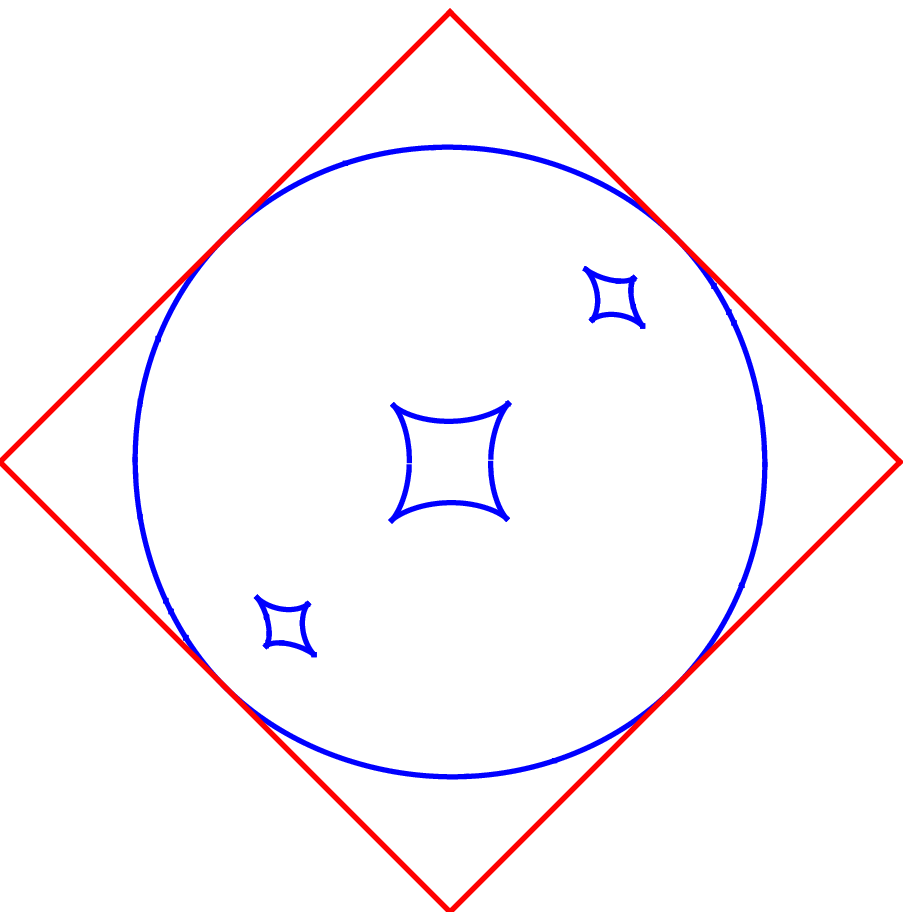}
             \put(-75,-15){$\lambda_2=1/4$}
             \hskip .2cm
             \includegraphics[width=3.75cm]{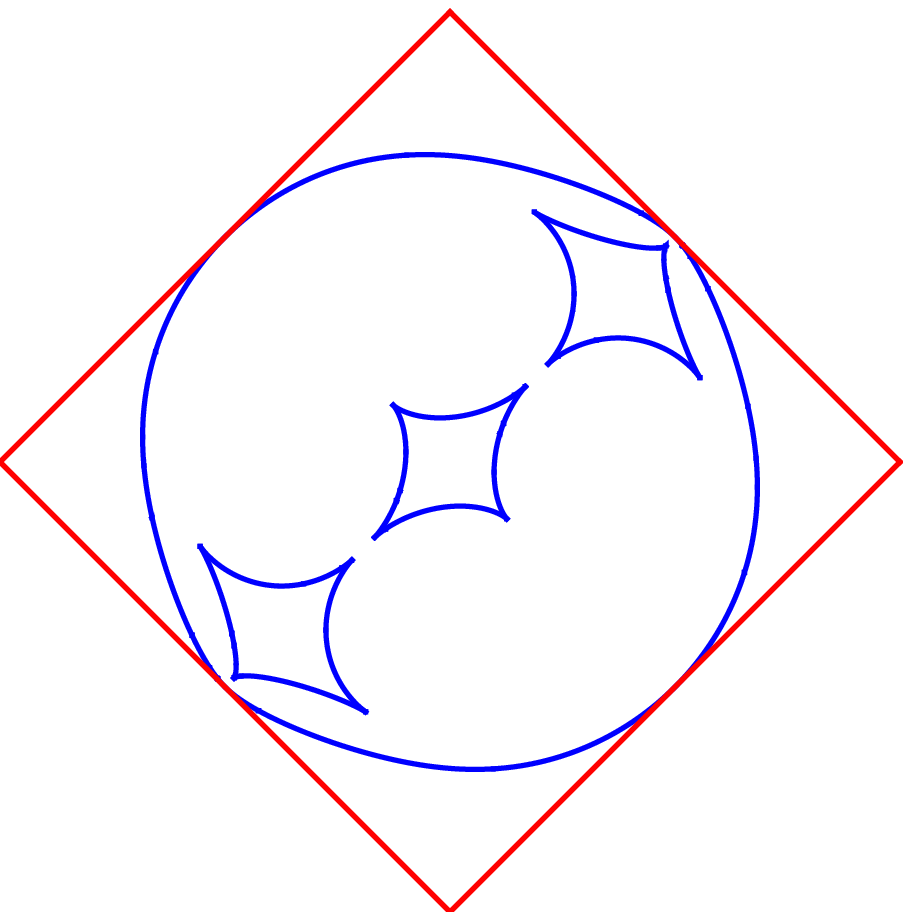}
             \put(-75,-15){$\lambda_2=9/10$}}

         \caption{\small Arctic curves for the $4$-toroidal initial data corresponding to different values of $\lambda_2$.
        Where $\lambda_0=1/2$, $\lambda_1=2/3$, $\lambda_3=\frac{2}{1+\lambda_2}-1$ and $\mu_0=\mu_1=\mu_2=\mu_3=1/2$.}
        \label{periodicarcticcurve2x4-2}
\end{figure}

\begin{figure}
        \centering

             \hbox{ 
                \includegraphics[width=3.75cm]{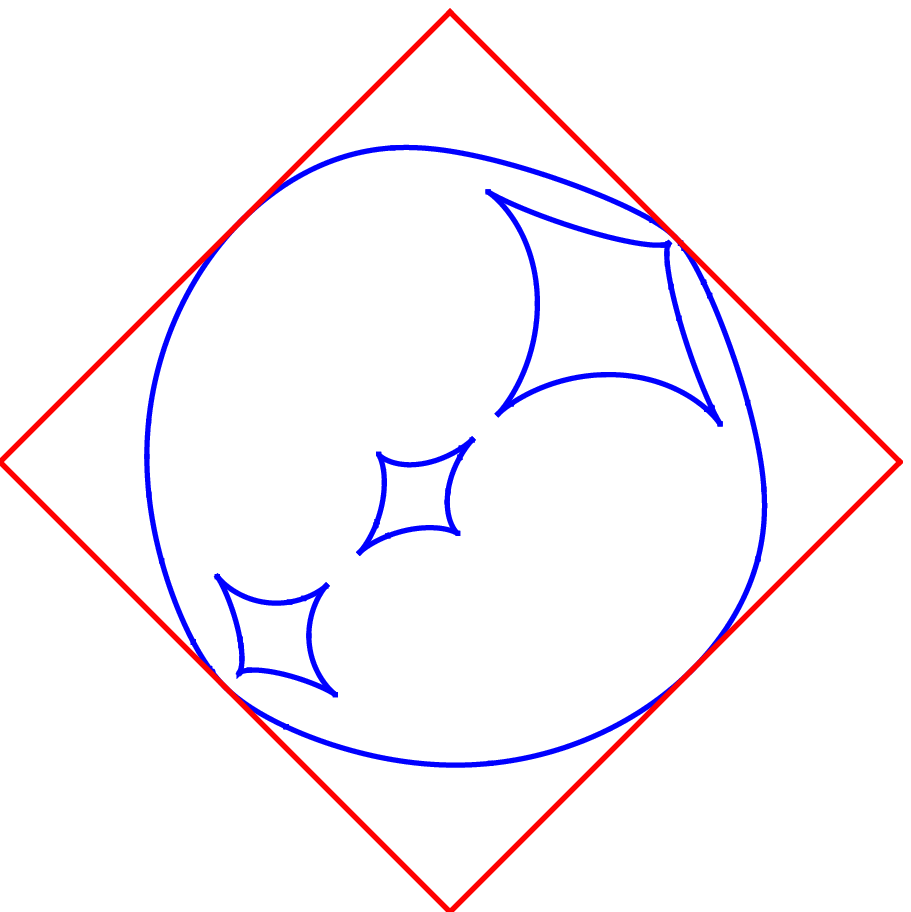}
                \put(-75,-15){$\mu_2=1/3$}
               \hskip .2cm
               \includegraphics[width=3.75cm]{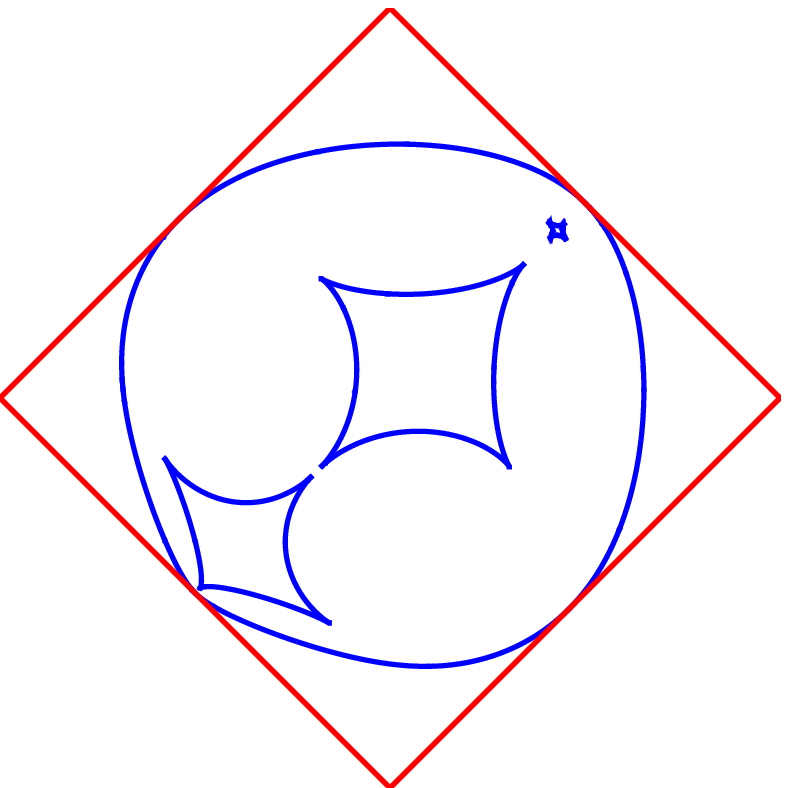}
               \put(-75,-15){$\mu_2=2/3$}
               \hskip .2cm
             \includegraphics[width=3.75cm]{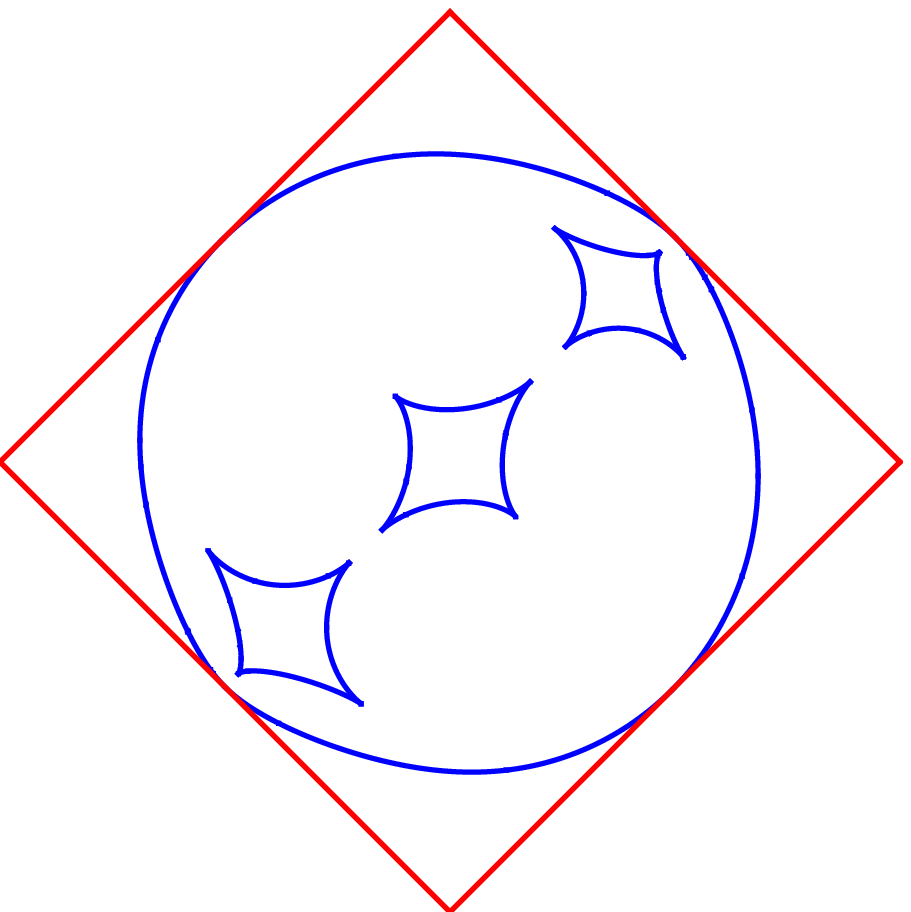}
             \put(-75,-15){$\mu_2=4/5$}
             \hskip .2cm
             \includegraphics[width=3.75cm]{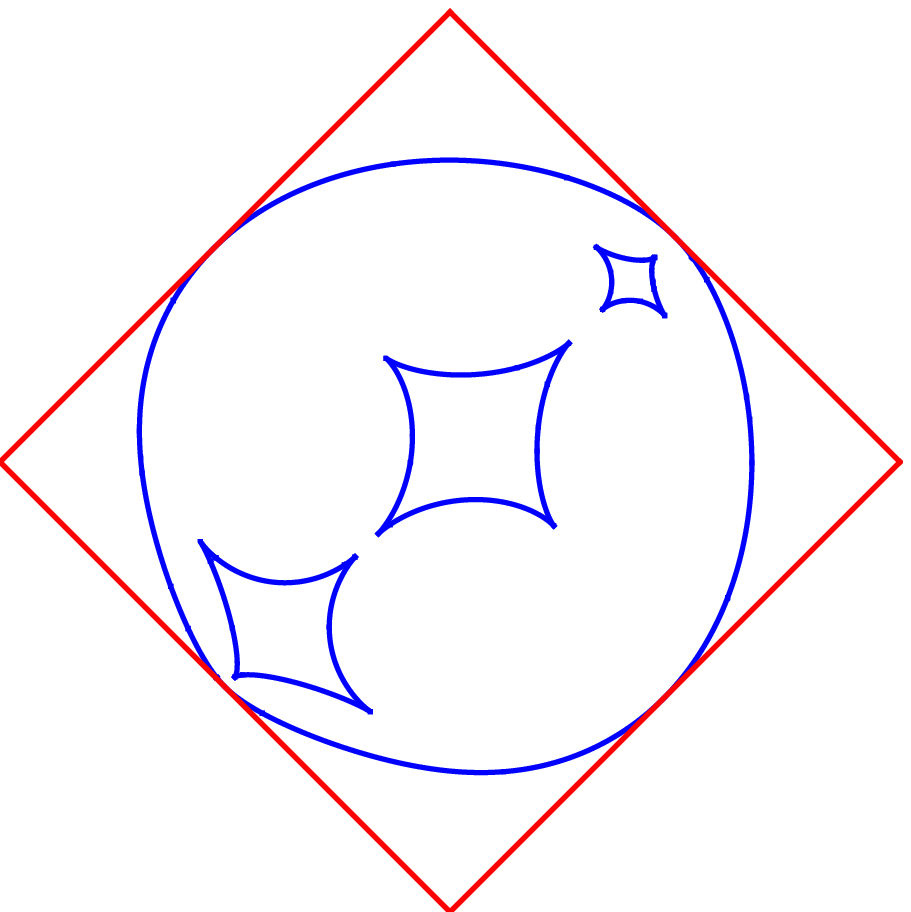}
             \put(-75,-15){$\mu_2=9/10$}}

         \caption{\small Arctic curves for the $4$-toroidal initial data corresponding to different values of $\mu_2$.
        Where $\lambda_0=1/2$, $\lambda_1=2/3$, $\lambda_2=4/5$, $\lambda_3=1/9$ and $\mu_0=1/2$, $\mu_1=1/4$ and 
        $\mu_3=\frac{3(1-\mu_2)}{3-2\mu_2}$.}
        \label{periodicarcticcurve2x4-3}
\end{figure}

The structure of phases is similar to the cases $m=2,3$, except that we now have $3$ inner facet regions along the diagonal of the square domain.
As before, we first display in Figs.\ref{periodicarcticcurve2x4-1}-\ref{periodicarcticcurve2x4-2} the case when all $\mu$'s are trivial ($=1/2$) and $\lambda$'s vary. We see that the picture is again symmetric w.r.t. the center. In Fig. \ref{periodicarcticcurve2x4-3}, we let both $\lambda$'s and $\mu$'s vary, and observe that the relative sizes of the various facet regions vary.

Note finally that the case $m=4$ reduces to the case $m=2$ if we choose $\lambda_2=\lambda_0$ and $\lambda_3=\lambda_1=1-\lambda_0$,
and $\mu_2=\mu_0$, $\mu_3=\mu_1=1-\mu_0$, as this changes the periodicity of $R_{i,j,k}$, $L_{i,j,k}$ to $\vec{e_1}'=(2,2)$, $\vec{e_2}'=(2,-2)$ and $\vec{e_3}'=(1,1,2)$.

\begin{figure}
\centering
\includegraphics[width=6.cm]{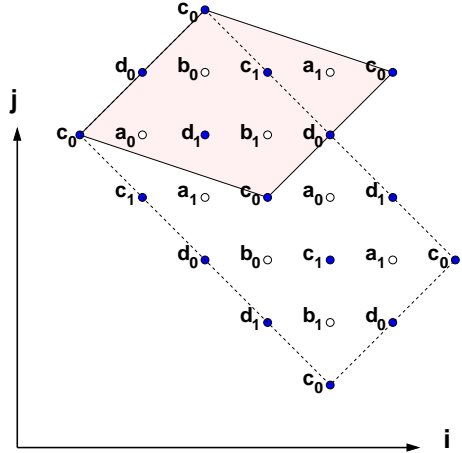}
\caption{\small For the $4$-toroidal case with $\lambda_0+\lambda_2=\lambda_1+\lambda_3=1$ and $\mu_0+\mu_2=\mu_1+\mu_3=1$,
the fundamental domain is {\it de facto} reduced by half. We have shaded the new fundamental domain, generated by $(2,2)$ and $(3,-1)$. }
\label{fig:4to3}
\end{figure}

Another interesting sub-case is when $\lambda_0+\lambda_2=\lambda_1+\lambda_3=1$ 
while $\mu_0+\mu_2=\mu_1+\mu_3=1$. In this case, we get a family of curves with one less facet region, 
which qualitatively look like those of the $m=3$ case (see Fig.\ref{periodicarcticcurve2x4-2} left for an example with $\lambda_0=\frac{1}{2}$
and $\lambda_1=\frac{2}{3}$ while $\mu_0=\mu_1=\frac{1}{2}$). 
Note that the choice of parameters above has changed the periodicity of 
$R_{i,j,k}$, $L_{i,j,k}$ to $\vec{e_1}'=(2,2)$, $\vec{e_2}'=(3,-1)$ and $\vec{e_3}'=(1,1,2)$, thus effectively dividing by $2$ the volume of the fundamental domain in $\Z^3$ for the $4$-toroidal case. One way of realizing the relations between the $\lambda$'s and the $\mu$'s 
is to take $c_{i+2}=d_i$, $d_{i+2}=c_i$ and $b_{i+2}=a_i$, $a_{i+2}=b_i$ for all $i$. The torus of initial data in this case is twice as small,
generated by $(2,2)$ and $(3,-1)$ as shown in Fig.\ref{fig:4to3}.

Note that despite its qualitative similarity with the $m=3$ case, this is different, as the two fundamental domains are inequivalent.

\subsubsection{General case and quotients}

In general, we expect generically $m-1$ facet regions along the diagonal of the square $|u|+|v|=1$.
Like in the cases $m=2,3,4$, we find that $\rho^{(0,0)}(1-tx,1-t y,1-t z)$ behaves like $t^{-3}$ for $\lambda_0\neq \lambda_{m-1}$
and like $t^{-2}$ otherwise, hence we expect that $\rho_{i,j,k}$ scales like $k^{-1}$ for $\lambda_{0}=\lambda_{m-1}$
and tends directly to a scaling function without overall rescaling otherwise.

For even $m=2p$, we have, like for $m=4$, two sub-cases of interest. 

The first one is a reduction to the $p$-toroidal case, by picking parameters $\lambda_i,\mu_i$
such that $\lambda_{i+p}=\lambda_i$ and $\mu_{i+p}=\mu_i$ for $i=0,1,..,p-1$. 
The second is by picking $\lambda_{i+p}=1-\lambda_i$ and $\mu_{i+p}=1-\mu_i$, $i=0,1,...,p-1$. 
The new periodicity of $R_{i,j,k}$, $L_{i,j,k}$ is $\vec{e_1}'=(2,2)$, $\vec{e_2}'=(p+1,p-1)$ and 
$\vec{e_3}'=(1,1,2)$, again dividing by $2$ the volume of the fundamental domain of $\Z^3$ of the $m$-toroidal case.
For this reason, we call the corresponding boundary conditions the $\Z_2$ quotient of the $m=2p$-toroidal case.

More generally, if $m=p q$ for some positive integers $p,q$, we may take $\lambda_{i+q}=\lambda_i$ and $\mu_{i+q}=\mu_i$ for all $i$
to reduce the $m$-toroidal case to the $q$-toroidal one. Similarly, for $p=2\ell$ even,
picking $\lambda_{i+q}=1-\lambda_i$ and $\mu_{i+q}=1-\mu_i$, for all $i$
reduces the $m$-toroidal case to one with periodicities $\vec{e_1}'=(2,2)$, $\vec{e_2}'=(q+1,q-1)$ and $\vec{e_3}'=(1,1,2)$,
that is to the $\Z_2$ quotient of the $2q$-toroidal case.

\section{Conclusion and discussion}

\subsection{Summary and perspectives}

In this paper we have obtained the exact solutions of the octahedron equation with initial data satisfying 
$m$-toroidal boundary conditions, namely  some specific doubly periodic initial conditions. We have used 
this solution to compute a density function of the associated dimer model on an Aztec graph, and investigate its singularities
in the limit of large size of the graph.

The generic result is a phase diagram with three types of phases for the dimer configurations: frozen corners with no entropy, 
disordered intermediate region, and facets with order and entropy.

Our analysis uses exclusively the octahedron equation and the properties of its particular solutions. 
Analogous equations have been 
considered recently such as the cube recurrence \cite{CS} related to combinatorial groves and the hexahedron equation 
related to double-dimers \cite{KEN}. 
Although not directly related to a dimer model, these display the same arctic phenomena \cite{KPS,KEN}. 
It would be interesting to see whether analogous exact solutions such as the $m$-toroidal ones of the present paper, 
exist in these other cases. If so, we expect some special patterns of frozen, disordered and facet phases to occur.

Another direction of generalization would consist of considering different geometries of initial data. In \cite{DF13}, 
arbitrary initial data on stepped surfaces for the octahedron equation were investigated. Given a stepped surface 
$(i,j,k_{i,j})_{i,j\in \Z}$ with $|k_{i,j+1}-k_{i,j}|=|k_{i+1,j}-k_{i,j}|=1$ and $i+j+k_{i,j}=1$ mod 2 for all $i,j$, 
these consist of the following initial data assignments:
$$ T_{i,j,k_{i,j}} =t_{i,j} \qquad (i,j\in \Z) $$
for some fixed parameters $t_{i,j}$, $i,j\in \Z$. It was shown in \cite{DF13} that the solution $T_{i,j,k}$ of the octahedron equation with such initial conditions
is the partition function of a dimer model on a graph obtained by taking the {\it shadow} of the point $(i,j,k)$ onto the initial data surface,
and attaching to the resulting graph edges weights expressed in terms of the local parameters $t_{i,j}$.
We may now consider the succession of partition functions $T_{i,j,k}$ for domains of growing size 
as $k$ increases. For sufficiently nice surfaces such as ``flat" periodic structures with a fixed average rational normal vector, 
we expect the thermodynamic limit $k\to \infty$  of these models to make sense. Special solutions of the octahedron equation
should still be amenable to the study of arctic curve phenomena.
The same could possibly hold for the so-called brane tiling models \cite{MUSIK}.

Our new solutions should allow for an investigation of the behavior of tilings at the boundary between phases, in the same spirit
as Ref. \cite{KJ}, where it was shown in the uniform case that the ``North Polar Region" boundary converges to the Airy process,
allowing for a connection to eigenvalue distributions of large random matrix ensembles.
In particular, we expect new universality classes to govern more singular regions, such as the quadruple points arising when some
weights tend to $0$, corresponding to the identification of two cuspidal points of two neighboring facet boundaries
(see the rightmost picture of  Figs. \ref{periodicarcticcurve2x3-1}  and \ref{periodicarcticcurve2x3-2} for instance), and reminiscent 
of the configuration leading to the tacnode process \cite{BD,ADM}.

\subsection{Cluster algebra and arctic curves}

\begin{figure}
\centering
\includegraphics[width=16.cm]{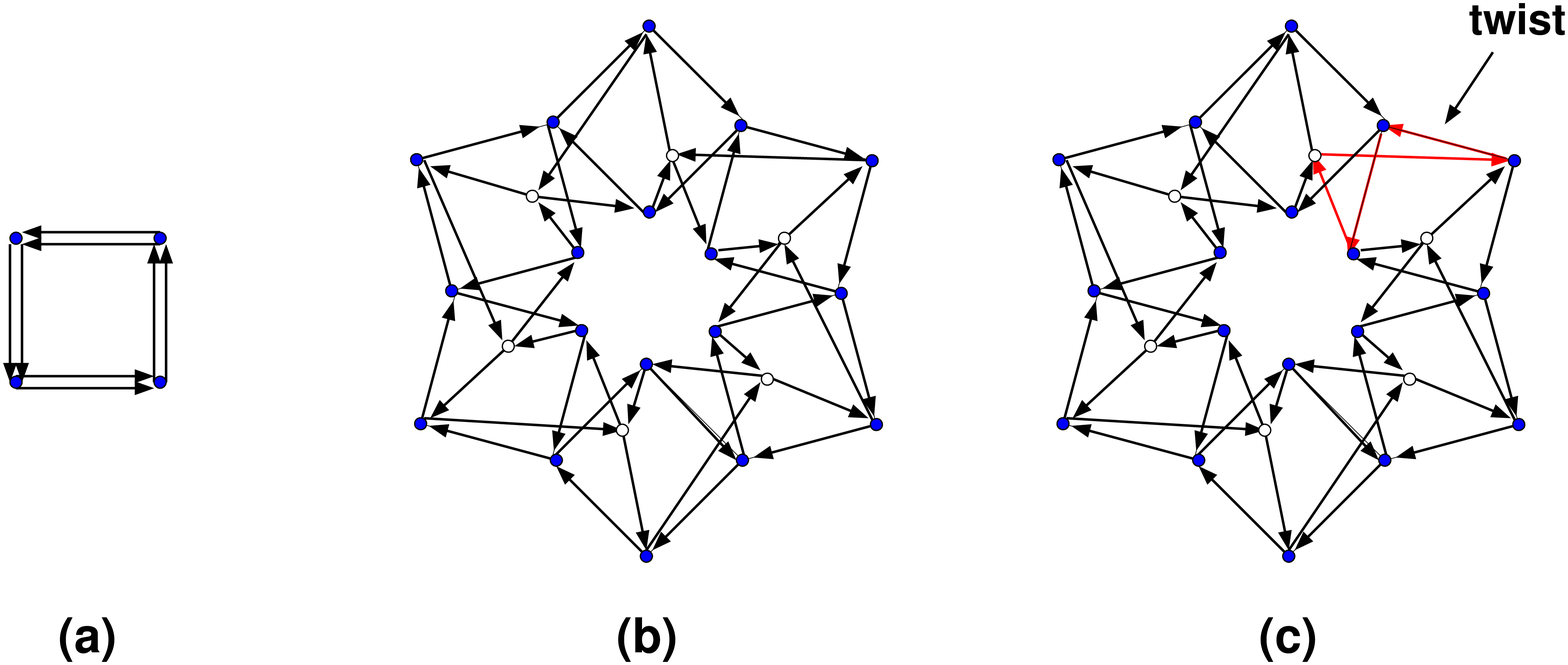}
\caption{\small The quivers obtained by folding the $T$-system quiver onto a torus, for the $2\times 2$ case (a), the $m=6$-toroidal case (b) and the $\Z_2$ quotient of the $m=12$-toroidal case (the quiver is simply obtained from the $m=6$ case by a twist, i.e. reversal of 4 arrows as indicated).}
\label{fig:quivers}
\end{figure}

The $T$-system equation  is known to be a particular mutation in an infinite rank  cluster algebra, with quiver given by 
an antiferromagnetic orientation of the edges of the square lattice $\Z^2$, namely such that every other square face 
is oriented clockwise. The cluster algebra describes rational transformations (called mutations) of variables $(x_i)$ 
attached to the vertices  $i$ of the quiver as follows. There is a mutation $\mu_k$ for each vertex $k$ of the quiver, and 
the action of $\mu_k$ on $(x_i)$ is \cite{FZI}:
\begin{eqnarray}
\mu_k(x_i)&=&x_i\qquad {\rm if}\ i\neq k\nonumber \\
\mu_k(x_k)&=& \frac{1}{x_k}\left(\prod_{i\to k} x_i +\prod_{k\to j} x_j \right)\label{muta}
\end{eqnarray}
where the first product is over all the tails of arrows with head $k$, and the second is over all heads of arrows whose tail is $k$.
The quiver mutates under $\mu_k$ as well (see \cite{FZI} for a precise definition: this is a so-called cluster algebra of geometric type, 
without coefficients).

The doubly periodic situation we have investigated in this paper corresponds to finite rank {\it folded} cluster
algebras, for which the above quiver has been folded, by identifying the vertices modulo the corresponding torus generators. 
We have represented in Fig.\ref{fig:quivers} (a) the folded quiver for the $2\times 2$ periodic case of Section \ref{sectwotwo},
in Fig.\ref{fig:quivers} (b) the folded quiver for the $m=6$-toroidal case, and in in Fig.\ref{fig:quivers} (c) the folded quiver for 
the $\Z_2$ quotient of the $m=12$-toroidal case. The latter is obtained from the $m=6$ quiver via a twist, namely the reversal of 4 arrows 
connecting two neighboring pairs of vertices. Alternatively, if we view the $m$-toroidal case quiver as a ribbon made of $m$ oriented octahedra, 
then the $\Z_2$ quotient of the $2m$-toroidal case quiver is the corresponding M\" obius strip.
Note that all the quivers are naturally bipartite.

The corresponding $T$ system with toroidal boundary conditions is simply obtained by restricting to compound mutations
in which one performs all mutations at vertices of a given parity, and alternating between the two parities (any such compound mutation reflects 
all the arrows of the quiver). Let us denote by $(x_i)_{i\in I}$ and $(B_{i,j})_{i,j\in I}$ respectively the cluster variables and exchange matrix 
elements coding the quiver ($B$ is a skew-symmetric integer matrix, such that $B_{i,j}>0$ counts the number of arrows from vertex $i$ to $j$). 
Here $I$ is a finite set, for instance $I=\{1,2,...,4m\}$ in the $m$-toroidal case.

The special property of the above quivers and their compound mutations is that the ``coefficient" variables defined as
\begin{equation} y_i=\prod_{k\in I} x_k^{B_{i,k}} \qquad (i\in I) \end{equation}
take only finitely many distinct values under arbitrary compound mutations. 
Indeed, in all cases, we may identify $y_i$ with the ratio $L_{i,j,k}/R_{i,j,k}$, 
which indeed takes only finitely many distinct values. Let us call {\it y-finite} the corresponding cluster algebra.

This property was instrumental for determining the density $\rho$ exactly. More generally, we may define an analogue 
of such a density for any cluster algebra. Pick an initial cluster $(x_i)_{i\in I}$, and a particular cluster variable, say $x_{i_0}$. 
For any mutated cluster
$(x_i')_{i\in I}$ we may define the density:
$$ \rho^{(i_0)}(x';x)_j=x_{i_0} \partial_{x_{i_0}} \, \log \, x_j' $$
It is easy to derive a linear recursion relation for $\rho_j$ by differentiating the mutation relation \eqref{muta}. We get:
$$ \rho_k'+\rho_k= L \sum_{i\to k} \rho_i +R \sum_{k\to j}\rho_j $$
with $L/R=y_k$ and $L+R=1$. 

Assume that we choose some special set of (possibly compound) mutations, such that the variables
$y_i$ only take finitely many values under iteration of these mutations, 
then we see that $\rho$ can be determined as the solution of a periodic linear system, with initial data $\rho^{(i_0)}(x;x)_j=\delta_{j,i_0}$. Denoting by $x^{(k)}$ the image of $x$ under the $k$-th iterate of these mutations
we may form the generating series: $\rho(\{w_i\}_{i\in I};z)=\sum_{k\geq 0, j\in I} \rho^{(i_0)}(x^{(k)};x)_j \, w_j\, z^k$. 
As $\rho^{(i_0)}(x^{(k)};x)_j$ solve a  periodic linear system, the function $\rho(\{w_i\}_{i\in I};z)$ is necessarily a rational fraction of $w_i,z$.
The denominator of this fraction governs the singularities of $\rho(\{w_i\}_{i\in I};z)$ at large $k$. 
The interesting case is if $I$ is infinite, and $x_i$
doubly or multiply periodic, say $I=\{ \vec{a}=(a_1,a_2,...,a_r)\in \Z^r \}$ and $x_{\vec{a}+\vec{e}_i}=x_{\vec{a}}$
for $r$ linearly independent vectors $\vec{e_i}$, $i=1,2,...,r$.  Then we can choose finitely many catalytic variables say $w_1,...,w_r$
for the generating function $\rho(\{w_i\}_{i\in[1,r]};z)=\sum_{k\geq 0,\vec{a}\in \Z^r}\rho^{(i_0)}(x^{(k)};x)_{\vec{a}} \, w_{\vec{a}} \, z^k$
where $w_{\vec{a}}=\prod_{i=1}^r w_i^{a_i}$. This multivariate generating function should display an arctic curve, 
obtained by blowing up the vicinity of the point $w_i=1$, $z=1$ and taking the algebraic dual.

This makes the search for y-finite cluster algebras worthwhile, as each of them will lead to interesting algebraic ``arctic" curves.

\appendix
\section{$m=3$ and $m=4$ arctic curves}
In this appendix we include explicit expressions for the limit shape curves for the cases $m=3$ and $m=4$. We only give the expression for a 
specific value of the parameters, since the expressions are in general very cumbersome.
\begin{itemize}
 \item Arctic curve for m=3 with $\lambda_0=1/2$, $\lambda_1=1/4$, $\lambda_2=3/4$, $\mu_0=1/2$, $\mu_1=1/5$ and $\mu_2=4/5$. This corresponds
 to the left most curve in Fig. \ref{periodicarcticcurve2x3-2}.
 \tiny{\begin{align}
  P(u,v)=&603358073569688095393738000 u^{14}+1822971422522481873814304800 vu^{13}+302414835014281399576977600 u^{13}\nonumber\\
  &+7658013562515635323215886000 v^2u^{12}+626386479045976264625165760 v u^{12}+65648625922043130480407960u^{12}\nonumber\\
  &+8502660801885990861442260800 v^3 u^{11}-4016291377989674598197523840 v^2u^{11}-8955889812423159779663425824 v u^{11}\nonumber\\
  &-648516348371464166524636080u^{11}+24870815123962290558144794000 v^4 u^{10}-961218355287663519951292800 v^3u^{10}\nonumber\\
  &-6515407606857043381218037200 v^2 u^{10}-172367781226698452854372560 vu^{10}+1099108080544208467044202281 u^{10}\nonumber\\
  &+8664424609796383599417068000 v^5u^9-17028590399764390279389912000 v^4 u^9-11010604932056552215403730080 v^3u^9\nonumber\\
  &+8631816024097405173283346160 v^2 u^9+6130342332781365103023636918 vu^9+300038159586641951467587240 u^9\nonumber\\
  &+48101067368389417583947124400 v^6u^8-8874912221343735420641284800 v^5 u^8-30642500612723566034952420120 v^4u^8\nonumber\\
  &+4089814979226593490453920400 v^3 u^8+2890833100949061663021542421 v^2u^8-1537862122247709326673670200 v u^8\nonumber\\
  &-1276791684735224437145235252u^8+165638160476224209249648000 v^7 u^7-15185547478970793846022129920 v^6u^7\nonumber\\
  &+10619324480232243252222805440 v^5 u^7+17121927414923400963351428640 v^4u^7-4642084019946561205079466936 v^3 u^7\nonumber\\
  &-4201308893745605384096673600 v^2u^7+97128658780698750571038384 v u^7+16658271644450437458125640u^7\nonumber\\
  &+48101067368389417583947124400 v^8 u^6-15185547478970793846022129920 v^7u^6-54696534109775129942931200352 v^6 u^6\nonumber\\
  &+8498087480515562992290313440 v^5u^6+20848735934263779279940738242 v^4 u^6-2928090072842649426783830400 v^3u^6\nonumber\\
  &-1125942030946106640101862864 v^2 u^6+881693827811784667334364120 vu^6+410818358444129895320450118 u^6\nonumber\\
  &+8664424609796383599417068000 v^9u^5-8874912221343735420641284800 v^8 u^5+10619324480232243252222805440 v^7u^5\nonumber\\
  &+8498087480515562992290313440 v^6 u^5-16598910777434586615901305852 v^5u^5-3118943690894703413913413040 v^4 u^5\nonumber\\
  &+4436727620735139576883870032 v^3u^5+378779090210933672213353800 v^2 u^5-894275420028329313474734772 vu^5\nonumber\\
  &-28143830188642461399955080 u^5+24870815123962290558144794000 v^{10}u^4-17028590399764390279389912000 v^9 u^4\nonumber\\
  &-30642500612723566034952420120 v^8u^4+17121927414923400963351428640 v^7 u^4+20848735934263779279940738242 v^6u^4\nonumber\\
  &-3118943690894703413913413040 v^5 u^4-6585025120215513060415620600 v^4u^4+224576822600011254994156440 v^3 u^4\nonumber\\
  &+730062356407169871489508026 v^2u^4-87999348446432687845418760 v u^4-39991576579826072416315884u^4\nonumber\\
  &+8502660801885990861442260800 v^{11} u^3-961218355287663519951292800 v^{10}u^3-11010604932056552215403730080 v^9 u^3\nonumber\\
  &+4089814979226593490453920400 v^8u^3-4642084019946561205079466936 v^7 u^3-2928090072842649426783830400 v^6u^3\nonumber\\
  &+4436727620735139576883870032 v^5 u^3+224576822600011254994156440 v^4u^3-771752886154129578670446744 v^3 u^3\nonumber\\
  &+54105975565681638845373840 v^2u^3+158742939499283087522192736 v u^3+3181828983737934822021000u^3\nonumber\\
  &+7658013562515635323215886000 v^{12} u^2-4016291377989674598197523840 v^{11}u^2-6515407606857043381218037200 v^{10} u^2\nonumber\\
  &+8631816024097405173283346160 v^9u^2+2890833100949061663021542421 v^8 u^2-4201308893745605384096673600 v^7u^2\nonumber\\
  &-1125942030946106640101862864 v^6 u^2+378779090210933672213353800 v^5u^2+730062356407169871489508026 v^4 u^2\nonumber\\
  &+54105975565681638845373840 v^3u^2-205856416682486477443753704 v^2 u^2-4541013871098771634821000 vu^2\nonumber\\
  &-417838190775940873949175 u^2+1822971422522481873814304800 v^{13}u+626386479045976264625165760 v^{12} u\nonumber\\
  &-8955889812423159779663425824 v^{11}u-172367781226698452854372560 v^{10} u+6130342332781365103023636918 v^9u\nonumber\\
  &-1537862122247709326673670200 v^8 u+97128658780698750571038384 v^7u+881693827811784667334364120 v^6 u\nonumber\\
  &-894275420028329313474734772 v^5u-87999348446432687845418760 v^4 u+158742939499283087522192736 v^3 u\nonumber\\
  &-4541013871098771634821000v^2 u+925709140319743466938350 v u-636257259784396800000 u\nonumber\\
  &+603358073569688095393738000v^{14}+302414835014281399576977600 v^{13}+65648625922043130480407960v^{12}\nonumber\\
  &-648516348371464166524636080 v^{11}+1099108080544208467044202281v^{10}+300038159586641951467587240 v^9\nonumber\\
  &-1276791684735224437145235252v^8+16658271644450437458125640 v^7+410818358444129895320450118 v^6\nonumber\\
  &-28143830188642461399955080v^5-39991576579826072416315884 v^4+3181828983737934822021000 v^3\nonumber\\
  &-417838190775940873949175v^2-636257259784396800000 v+6507176520522240000\nonumber
 \end{align}}
\normalsize{ 
\item Arctic curve for m=4 with $\lambda_0=\lambda_1=1/2$, $\lambda_2=9/10$ $\lambda_3=1/10$ and $\mu_0=\mu_1=\mu_2=\mu_3=1/2$. This corresponds
to the third curve from the left in Fig. \ref{periodicarcticcurve2x4-1}.} 
\tiny{\begin{align*}
P(u,v)&=1865357057070562500 u^{20}+8040214205493930000 v u^{19}+24057908820831125400 v^2 u^{18}-4441125622088345250 u^{18}\\
&+39313216117293630480 v^3u^{17}-29939012844366018600 v u^{17}+73651180421168030644 v^4 u^{16}-12929933668024326890 v^2 u^{16}\\
&+10118310470530522825 u^{16}+72428652840795390912 v^5   u^{15}-123453887076899013696 v^3 u^{15}+7155651657787473900 v u^{15}\\
&+120374020497531686304 v^6 u^{14}-65490141412881018800 v^4 u^{14}+28101666374921987920 v^2u^{14}-7424978692843390100 u^{14}\\
&+54528844935004785600 v^7 u^{13}-119634727591822485216 v^5 u^{13}+82981649042010236868 v^3 u^{13}+16659576911242363500 vu^{13}\\
&+232890902059778826120 v^8 u^{12}-200076579710447004960 v^6 u^{12}+70320961099229135980 v^4 u^{12}-55379407385229146900 v^2 u^{12}\\
&+1294150941317351875u^{12}+963632570756269152 v^9 u^{11}+53942752458379867200 v^7 u^{11}-52476298064771112660 v^5 u^{11}\\
&+13172800480291575480 v^3 u^{11}-8797446414920899800 vu^{11}+343772140948551525776 v^{10} u^{10}-557579149718173524388 v^8 u^{10}\\
&+298100820891210187760 v^6 u^{10}-64636470915389906508 v^4u^{10}+38201320674387580830 v^2 u^{10}+354087698981500350 u^{10}\\
&+963632570756269152 v^{11} u^9+135203610372643741200 v^9 u^9-222149669702554081500 v^7u^9+118161884723925156948 v^5 u^9\\
&-48376121742836702328 v^3 u^9+222939375275280000 v u^9+232890902059778826120 v^{12} u^8-557579149718173524388 v^{10}u^8\\
&+605093816463271297814 v^8 u^8-235160946416126975500 v^6 u^8+55243479544063681389 v^4 u^8-8633087504875621410 v^2 u^8\\
&-67353790853352825u^8+54528844935004785600 v^{13} u^7+53942752458379867200 v^{11} u^7-222149669702554081500 v^9 u^7\\
&+119117082588260646160 v^7 u^7-71814538701551870384 v^5u^7+20230021731703841536 v^3 u^7+281342047822398300 v u^7\\
&+120374020497531686304 v^{14} u^6-200076579710447004960 v^{12} u^6+298100820891210187760 v^{10}u^6-235160946416126975500 v^8 u^6\\
&+124954126399081716836 v^6 u^6-23662723202246442204 v^4 u^6+106708566476157900 v^2 u^6-11636145655350000u^6\\
&+72428652840795390912 v^{15} u^5-119634727591822485216 v^{13} u^5-52476298064771112660 v^{11} u^5+118161884723925156948 v^9 u^5\\
&-71814538701551870384 v^7u^5+16371208262322883456 v^5 u^5-1896684382137151100 v^3 u^5+36851525477017500 v u^5\\
&+73651180421168030644 v^{16} u^4-65490141412881018800 v^{14}u^4+70320961099229135980 v^{12} u^4-64636470915389906508 v^{10} u^4\\
&+55243479544063681389 v^8 u^4-23662723202246442204 v^6 u^4+3656070264108895450 v^4u^4-24666825302795000 v^2 u^4\\
&-212561857484375 u^4+39313216117293630480 v^{17} u^3-123453887076899013696 v^{15} u^3+82981649042010236868 v^{13}u^3\\
&+13172800480291575480 v^{11} u^3-48376121742836702328 v^9 u^3+20230021731703841536 v^7 u^3-1896684382137151100 v^5 u^3\\
&-25956260718995000 v^3u^3+328736881500000 v u^3+24057908820831125400 v^{18} u^2-12929933668024326890 v^{16} u^2\\
&+28101666374921987920v^{14} u^2-55379407385229146900 v^{12}u^2+38201320674387580830 v^{10} u^2-8633087504875621410 v^8 u^2\\
&+106708566476157900 v^6 u^2-24666825302795000 v^4 u^2+529284795718750 v^2 u^2+1366328125000u^2+8040214205493930000 v^{19} u\\
&-29939012844366018600 v^{17} u+7155651657787473900 v^{15} u+16659576911242363500 v^{13} u-8797446414920899800 v^{11}u\\
&+222939375275280000 v^9 u+281342047822398300 v^7 u+36851525477017500 v^5 u+328736881500000 v^3 u-15885000000000 v u\\
&+1865357057070562500v^{20}-4441125622088345250 v^{18}+10118310470530522825 v^{16}-7424978692843390100 v^{14}\\
&+1294150941317351875 v^{12}+354087698981500350v^{10}-67353790853352825 v^8-11636145655350000 v^6-212561857484375 v^4\\
&+1366328125000 v^2+97656250000
\end{align*}}
\end{itemize}

\end{document}